\documentclass[11pt]{article}
\usepackage[utf8]{inputenc}
\usepackage[margin=1in]{geometry}
\usepackage{graphicx}
\usepackage{amsmath}
\usepackage{amsthm}
\usepackage{amssymb}
\usepackage{amsfonts}
\usepackage{amsthm}
\usepackage{tikz}
\usepackage{comment}
\usepackage{caption}
\usetikzlibrary{decorations.markings}
\usepackage{microtype}
\usepackage{mathrsfs}
\usepackage{hyperref}
\usepackage{enumitem}
\usepackage[toc,page]{appendix}
\usepackage{cancel}
\usepackage[normalem]{ulem}
\usepackage{float}
\usepackage{mathtools}

\newtheorem{thm}{Theorem}[section]
\newtheorem*{thm*}{Theorem}
\newtheorem*{question}{Question}
\newtheorem{lemma}[thm]{Lemma}
\newtheorem{prop}[thm]{Proposition}
\newtheorem{proposition}[thm]{Proposition}

\newtheorem{claim}[thm]{Claim}
\newtheorem{corollary}[thm]{Corollary}
\newtheorem{conj}[thm]{Conjecture}

\theoremstyle{definition}
\newtheorem{df}[thm]{Definition}
\newtheorem{definition}[thm]{Definition}
\newtheorem{rem}[thm]{Remark}
\newtheorem{remark}[thm]{Remark}
\newtheorem{exam}[thm]{Example}

\DeclareMathOperator*{\des}{des}

\newcommand{\C}{\mathbb{C}}
\newcommand{\R}{\mathbb{R}}

\newcommand{\Z}{\mathbb{Z}}

\newcommand{\E}{\mathbb E}
\newcommand{\D}{\mathcal D}

\newcommand{\M}{\mathcal M}
\renewcommand{\P}{\mathcal P}
\newcommand{\PP}{\mathscr P}

\newcommand{\ga}{\gamma}

\newcommand{\eps}{\varepsilon}

\renewcommand{\th}{\theta}
\newcommand{\la}{\lambda}
\newcommand{\ka}{\kappa}
\newcommand{\pa}{\partial}
\newcommand{\myeq}{\stackrel{?}{=}}

\newcommand{\limtwo}{\lim_{\substack{N\to\infty,\, \th\to 0\\ \th N\to\ga}}}
\newcommand{\setzeroes}{x_1=\dots=x_N=0}
\newcommand{\wt}{\widetilde}

\newcommand{\Tmc}{T^\ga_{m\to\ka}}
\newcommand{\Tcm}{T^\ga_{\ka\to m}}

\newcommand{\ii}{\mathbf i}

\renewcommand{\d}{\mathrm d}

\newcommand{\GT}{\mathcal G}

\newcommand{\vx}{\vec{x}}

\numberwithin{equation}{section}

\begin{document}

\title{Matrix addition and the Dunkl transform at high temperature}

\author{Florent Benaych-Georges\thanks{Universit\'{e} Paris Descartes and Capital Fund Management. \href{mailto:florent.benaych@gmail.com }{florent.benaych@gmail.com}}, Cesar Cuenca\thanks{Harvard University. \href{mailto:cesar.a.cuenk@gmail.com }{cesar.a.cuenk@gmail.com}} , and Vadim Gorin\thanks{University of Wisconsin--Madison and Institute for Information Transmission Problems of RAS. \href{mailto:vadicgor@gmail.com}{vadicgor@gmail.com} }}

\date{}
\maketitle

\begin{abstract}
  We develop a framework for establishing the Law of Large Numbers for the eigenvalues in the random matrix ensembles as the size of the matrix goes to infinity simultaneously with the beta (inverse temperature) parameter going to zero. Our approach is based on the analysis of the (symmetric) Dunkl transform in this regime. As an application we obtain the LLN for the sums of random matrices as the inverse temperature goes to 0. This results in a one-parameter family of binary operations which interpolates between classical and free convolutions of the probability measures. We also introduce and study a family of deformed cumulants, which linearize this operation.
\end{abstract}

\tableofcontents

\section{Introduction}

\subsection{Overview}


\label{Section_Overview}

This text comes out of two circles of ideas. On one side, we are interested in $\beta$--ensembles of random matrix theory, where $\beta=1,2,4$ correspond to matrices with real/complex/quaternionic entries, but many distributions admit natural extensions to general real values of $\beta>0$. The theoretical physics tradition refers to the $\beta$ parameter as the inverse temperature. The matrices of interest are $N\times N$ and self--adjoint; we study the $N\to\infty$ asymptotic behavior of their eigenvalues on large scales. It was noticed by many authors (first, at classical $\beta=1,2,4$ and later for all $\beta>0$, see, e.g.\, \cite{BenArous-Guionnet, Johansson_CLT, BG} for the results of the latter type) that in the global regime, when we deal with all eigenvalues together and describe the asymptotics of their empirical measures through Laws of Large Numbers and Central Limit Theorems, the only dependence of the answers on $\beta$ is in simple normalization prefactors. In other words, the limits as $N\to\infty$ essentially do not depend on $\beta$, as long as $\beta>0$ remains fixed. Recently, it was shown that the situation changes, if one varies $\beta$ together with $N$ in such a way that $\beta N$ tends to a constant $2 \gamma>0$ as $N\to\infty$ (high-temperature regime). \cite{ABG,ABMV,TT_Jacobi} prove that for all classical ensembles of random matrices (Gaussian/Wigner, Laguerre/Wishart, and Jacobi/MANOVA) there is a different Law of Large Numbers in the high-temperature regime, and the resulting limit shapes non-trivially depend on the $\gamma$ parameter. A subsequent wave produced many more results in the $\beta N\to 2 \gamma$ asymptotic regime, such as the study of local statistics in \cite{KS,BGP,Pa}, or of central limit theorems in \cite{NT,HL}, or of the loop equations in \cite{FM}, or of the spherical integrals in \cite{MP}, or of the 2D systems in \cite{AB}, or connections to Toda chain in \cite{Spohn}, or of dynamic versions in \cite{NTT}; this is very far from the complete list of results and we refer to the previously mentioned articles for further references.

From another side, a classical tool of the probability theory for establishing asymptotic theorems is by using the characteristic functions or Fourier transforms. In the last 10 years, a Fourier approach has been developed for the strongly correlated $N$--particle systems (with distributions of random-matrix type) in the series of papers \cite{GP, BuG1, BuG2, BuG3, NovakM, Huang, GS, C,Ahn}. The central idea is to replace the exponents in the Fourier transform by symmetric functions of the representation-theoretic origin (such as Schur symmetric polynomials or multivariate Bessel functions) and to further connect the partial derivatives of the logarithm of the new transform to the asymptotic behavior of the particle system (mostly, in the global regime) by using differential operators diagonalized by these symmetric functions.

\smallskip

In this article we develop a theory of integral transforms of $N$--tuples of real numbers (which should be thought of as eigenvalues of a random $N\times N$ matrix) using multivariate Bessel functions of general parameter $\theta=\tfrac{\beta}{2}>0$ and generalizing conventional Fourier transform at $\theta=0$; such transforms are also known as symmetric Dunkl transforms in the special functions literature, see \cite{A} for a review. We prove a very general theorem stating that  the partial derivatives of the logarithms of our transforms at $0$ have prescribed limits as $N\to\infty$, $\theta\to 0$, $\theta N\to \gamma$ if and only if the associated random $N$--tuples satisfy a form of the Law of Large Numbers as $N\to\infty$, see Theorem \ref{thm_small_th}. In our theory these partial derivatives play the same role as cumulants in classical probability and free cumulants in the free probability. We further develop a combinatorial theory of our new $\gamma$--cumulants in Theorems \ref{theorem_cumuls_moms} and \ref{thm:mom_cums2}.

We present several applications of our theory:
\begin{itemize}
 \item We recover previous results about Gaussian and Laguerre $\beta$--ensembles of random matrices as $\beta\to 0$, $N\to\infty$, $\beta N\to2\gamma$, and recast them in the framework of $\gamma$--cumulants, see Section \ref{Section_GbE}, Example \ref{laguerre_exam}, and  Remark \ref{Remark_Gauss_Laguerre}.
 \item We investigate eigenvalues of the sum of two independent self--adjoint matrices in the limit $\beta N\to 2\gamma$. We prove the Law of Large Numbers in this regime and encounter a new operation of $\gamma$-convolution, interpolating between usual convolution at $\gamma=0$ and free (additive) convolution at $\gamma=\infty$, see Theorem \ref{Theorem_gamma_convolution}.

 \item We obtain the Law of Large Numbers for ergodic Gibbs measures on the $\beta$--corners branching graph of \cite{OV, AN} in the regime $\beta N\to 2\gamma$, see Theorem \ref{Theorem_ergodic}. The limits are infinitely-divisible with respect to $\gamma$--convolution.

 \item We find that each probability measure $\mu$ gives rise to a 1-parametric family of probability measures $\mu^{\tau, \gamma}$, $\tau\in [1,+\infty)$ which are $\beta N\to 2\gamma$ limits of empirical measures of spectra of $\lfloor N/\tau\rfloor\times \lfloor N/\tau\rfloor$ submatrix of $N\times N$ matrix whose spectrum approximates $\mu$ as $N\to\infty$. An intriguing property of the family is that all these measures are constructed from the same sequence of numbers, which are interpreted as the $\tfrac{\gamma}{\tau}$--cumulants of $\mu^{\tau, \gamma}$.

\end{itemize}

\subsection{Addition of matrices as $\theta=\tfrac{\beta}{2}\to 0$}
\label{Section_addition_intro}

Rather than explaining our results in the most general and abstract setting, we focus on describing a particular application which was the original motivation for this work: the addition of random matrices.  We start from a classical question. Let $A$ and $B$ be two self-adjoint $N\times N$ matrices with (real) eigenvalues $a_1\le a_2\le\dots\le a_N$ and $b_1\le b_2\le \dots \le b_N$, respectively. What can we say about the eigenvalues $c_1\le c_2\le \dots\le c_N$ of the sum $C=A+B$?

 The deterministic version of this problem asks to describe all possible values for $c_1\le \dots\le c_N$ if $A$ and $B$ are allowed to vary arbitrarily while preserving their eigenvalues. This question was first posed by Weyl \cite{Weyl} in 1912 and it took the full XX century before it was completely resolved, see \cite{KT} for a review. The answer is given by a convex set determined by the equality $\sum_{i=1}^N c_i=\sum_{i=1}^N a_i + \sum_{i=1}^N b_i$ (coming from ${\rm Trace}(C)={\rm Trace}(A)+{\rm Trace}(B)$) and a large list of inequalities satisfied by $c_1,\dots,c_N$: the simplest ones are well-known, for instance, $c_N\le a_N+b_N$, but there are many much more delicate relations.

 The stochastic version of the same problem starts from random and independent matrices $A$ and $B$. We assume that $A$ is sampled from the uniform measure on the set of all matrices with prescribed eigenvalues\footnote{Say, we deal with complex Hermitian matrices. Then this set is an orbit of the unitary group $U(N)$ under the action by conjugations, and the uniform measure on the orbit is the image of the Haar (uniform) measure on $U(N)$ with respect to this action.} $a_1\le\dots\le a_N$ and, similarly, $B$ is a uniformly random matrix with eigenvalues $b_1\le\dots\le b_N$. Then the  eigenvalues $c_1\le \dots\le c_N$ are random and we would like to obtain some description of them, with the most interesting questions pertaining to the situation of a very large $N$. The first asymptotic answer as $N\to\infty$ was obtained by Voiculescu  in the context of the free probability theory.

 \begin{thm}[\cite{Vo3}; see also \cite{Col, ColSn}]  \label{Theorem_Voi} Suppose that $A$ and $B$ are independent $N\times N$ uniformly random self-adjoint matrices with  spectra $a_1(N)\le \dots \le a_N(N)$ and $b_1(N)\le\dots\le b_N(N)$, respectively, and let $c_1(N)\le \dots \le c_N(N)$ be the (random) eigenvalues of $C=A+B$. Suppose that for two  probability measures $\mu_A$, $\mu_B$, we have:
 $$
   \lim_{N\to\infty} \frac{1}{N}\sum_{i=1}^N \delta_{a_i(N)}=\mu_A,\qquad
   \lim_{N\to\infty} \frac{1}{N}\sum_{i=1}^N \delta_{b_i(N)}=\mu_B.
 $$
 Then the random empirical measures $\frac{1}{N}\sum_{i=1}^N \delta_{c_i(N)}$ converge as $N\to\infty$ (weakly, in probability) to a deterministic measure $\mu_C:=\mu_A \boxplus \mu_B$, which is called the free convolution of $\mu_A$ and $\mu_B$.
 \end{thm}

In order to use this theorem, it is important to be able to efficiently describe the measure $\mu_A\boxplus \mu_B$. Let us briefly present two points of view on such description and refer to textbooks \cite{NS,MS} for more details. The first point of view is analytic and it relies on the notion of the Voiculescu $R$--transform of a probability measure $\mu$, defined through:
$$
 R_\mu(z)=(G_\mu(z))^{(-1)}-\frac{1}{z},\qquad G_\mu(z)=\int_{\mathbb R} \frac{1}{z-x} \mu(dx),
$$
where $G_\mu(z)$ is the Stieltjes transform of $\mu$ and $(G_\mu(z))^{(-1)}$ is the functional inverse.  For a compactly supported $\mu$, $R_\mu(z)$ is holomorphic in a complex neighborhood of $0$. The measure $\mu_A\boxplus \mu_B$ is determined by:
\begin{equation}
 \label{eq_free_conv_R}
 R_{\mu_A\boxplus \mu_B}(z)=R_{\mu_A}(z)+R_{\mu_B}(z).
\end{equation}
The relation \eqref{eq_free_conv_R} is a free probability version of the linearization of conventional convolution by logarithms of the characteristic functions: if $\xi$ and $\eta$ are independent random variables, then
\begin{equation}
 \label{eq_conv_lin}
\ln \E e^{\ii t (\xi+\eta)}= \ln \E e^{\ii t \xi}+ \ln \E e^{\ii t \eta}.
\end{equation}
An alternative combinatorial approach to the free convolution uses free cumulants of a probability measure $\mu$ denoted $\kappa^\mu_n$, $n=1,2,\dots$. They are defined as certain explicit polynomials in the moments of the measure $\mu$. Simultaneously, the free cumulants are coefficients of the Taylor-series expansion of $R_\mu(z)$ at the origin, so \eqref{eq_free_conv_R} gets restated as
\begin{equation}
 \label{eq_free_conv_cum}
 \kappa^{\mu_A\boxplus \mu_B}_n=\kappa^{\mu_A}_n+\kappa^{\mu_B}_n, \quad n=1,2,\dots.
\end{equation}
This relation is a free probability version of the statement that conventional cumulants of a sum of independent random variables are sums of the cumulants of the summands.

\bigskip

Note that in Voiculescu's Theorem \ref{Theorem_Voi} we never specified, whether we deal with real symmetric, or complex Hermitian, or quaternionic Hermitian random matrices. And in fact, the theorem remains exactly the same in all these settings, which are usually referred as the $\beta=1,2,4$ cases in the random matrix literature. What we would like to do is to go one step further and to extend the setting of the Theorem \ref{Theorem_Voi} to the general $\beta$ setting. However, there is no (skew-)field of general real dimension $\beta>0$, and therefore, there are no independent random matrices $A$ and $B$ over such field, which we could add. Hence, we first need to address a question:

\begin{question}
 What does it mean to add two independent self-adjoint $\beta$-random matrices $A$ and $B$?
\end{question}

Our answer to this question is based on the Fourier point of view on the addition of matrices. Suppose that $Q=[Q_{ij}]_{i,j=1}^N$ is a random real symmetric matrix. Its Fourier-Laplace transform is a function of another (deterministic) matrix $X$ given by:
\begin{equation}
  \chi_Q(X)=\E \exp\Bigl({\rm Trace} (XQ)\Bigr)=\E \exp\biggl(\,\sum_{i,j=1}^N x_{ij} Q_{ji}\,\biggr).
\end{equation}
Let us assume that the law of $Q$ is invariant under conjugations by orthogonal matrices (which is the case for all three matrices $A$, $B$, and $C$ in the Theorem \ref{Theorem_Voi}). In addition assume that the matrix $X$ is normal (i.e.\ $X X^*=X^* X$), which implies that $X$ can be diagonalized by orthogonal conjugations\footnote{If we know that $Q$ is invariant under orthogonal conjugations and we know the values of $\chi_Q(X)$ for all normal $X$, then we can uniquely determine the law of $Q$. In fact it is sufficient to take $X$ to be symmetric (or $\mathbf i$ times symmetric).}. In this situation, conjugating $X$ and noting invariance of the trace, we see that $\chi_Q(X)$ is a function of the eigenvalues $x_1,\dots,x_N$ of $X$ and we can write it as $\chi_Q(x_1,\dots,x_N)$.

If we specialize to the case when $Q$ is a uniformly random real symmetric matrix with deterministic eigenvalues $q_1\le \dots\le q_N$, then $\chi_Q$ is known as a \emph{multivariate Bessel function} at $\theta=\tfrac{\beta}{2}=\tfrac12$:
\begin{equation}
\label{eq_Bessel_as_Fourier}
 \chi_Q(x_1,\dots,x_N)= B_{(q_1,\dots,q_N)}\bigl(x_1,\dots,x_N; \, \tfrac{1}{2}\bigr).
\end{equation}
Going further, the definition of $\chi_Q$ and linearity of the trace immediately imply that for independent conjugation-invariant matrices $A$ and $B$ we have
\begin{equation}
\label{eq_Fourier_product}
 \chi_{A+B}(x_1,\dots,x_N)=\chi_A(x_1,\dots,x_N) \chi_{B}(x_1,\dots,x_N).
\end{equation}
Moreover, we can take \eqref{eq_Fourier_product} as a \emph{definition} of $A+B$: the matrix $A+B$ is defined as a random $N\times N$ real symmetric matrix, whose law is invariant under orthogonal conjugations, and whose Fourier-Laplace transform is given by the right-hand side of \eqref{eq_Fourier_product}.

The same argument can be given for complex Hermitian matrices and for quaternionic Hermitian matrices with the only difference being that the parameter of the Bessel functions in \eqref{eq_Bessel_as_Fourier} changes to $\theta=1$ and $\theta=2$, respectively. But in fact, the multivariate Bessel functions $B_{(q_1,\dots,q_N)}\bigl(x_1,\dots,x_N; \, \theta\bigr)$ make sense for any real $\theta>0$, see Section \ref{sec:bessel} for a formal definition. They are intimately connected to many topics, in particular, they are eigenfunctions of rational Calogero-Sutherland Hamiltonian and of (symmetric versions of) Dunkl operators; they are also limits of Jack and Macdonald symmetric polynomials.
We are now ready to define the general $\beta$-analogue of addition of random matrices:

\begin{definition} \label{Def_theta_addition} Fix $\theta=\tfrac{\beta}{2}>0$. Given deterministic $N$--tuples of reals $\mathbf a=(a_1\le\dots\le a_N)$ and $\mathbf b=(b_1\le \dots\le b_N)$, we define a random $N$--tuple $\mathbf c=(c_1\le \dots \le c_N)$ by specifying its law through
\begin{equation}
\label{eq_def_theta_addition}
  \E B_{(c_1,\dots,c_N)}(x_1,\dots,x_N;\, \theta)= B_{(a_1,\dots,a_N)}(x_1,\dots,x_N;\, \theta) B_{(b_1,\dots,b_N)}(x_1,\dots,x_N;\, \theta), \qquad x_1,\dots,x_N\in\mathbb C.
\end{equation}
We say that $\mathbf c$ is the eigenvalue distribution for the $\theta$--sum of independent Hermitian matrices with spectra $\mathbf a$ and $\mathbf b$. We write $\mathbf c= \mathbf a +_{\theta} \mathbf b$.
\end{definition}

For example, when $a_1 = \cdots = a_N = a$, the multivariate Bessel function is $B_{(a,\dots,a)}(x_1,\dots,x_N;\, \theta) = \exp(a(x_1 + \cdots + x_N))$. On the other hand, we have the identity
\begin{equation}\label{const_seq}
B_{(b_1+a,\dots,b_N+a)}(x_1,\dots,x_N;\, \theta) = \exp(a(x_1 + \cdots + x_N))\, B_{(b_1,\dots,b_N)}(x_1,\dots,x_N;\, \theta),
\end{equation}
as follows from Definition \ref{Definition_Bessel_function} below.
So in the case that $\mathbf a_{\textrm{const}} = (a\le\cdots\le a)$ is the constant sequence, and $\mathbf b = (b_1\le \cdots \le b_N)$ is arbitrary, then by comparing \eqref{eq_def_theta_addition} and \eqref{const_seq}, we conclude that $\mathbf c = \mathbf a_{\textrm{const}} +_{\theta} \mathbf b$ is the Dirac delta mass at the point $(b_1+a\le \cdots \le b_N+a)$. For more general sequences $\mathbf a$ we are not aware of similarly simple expressions for $\mathbf a +_{\theta} \mathbf b$.

Let us remark that the uniqueness of the law of $(c_1,\dots,c_N)$ defined through \eqref{eq_def_theta_addition} is not hard to prove by expressing expectations of various test functions through expectations of multivariate Bessel functions.\footnote{For a reader who is not familiar with the theory of multivariate Bessel functions, we remark that at $N=1$, $B_{(a)}(z;\, \theta)=\exp(az)$. Hence, choosing $z=\ii t$, the Bessel functions turn into the exponents $\exp (\ii a t)$ and uniqueness turns into the well-known uniqueness of a measure with a given Fourier transform.}
In the existence part, there is a caveat. It is known that \eqref{eq_def_theta_addition} defines $\mathbf c$ as a compactly supported generalized function (or distribution), see \cite{Tri}, \cite[Section 3.6]{A}. It is also straightforward to see that the total mass of distribution of $\mathbf c$ is $1$ by inserting $x_1=\dots=x_N=0$ into \eqref{eq_def_theta_addition}. However, the \emph{positivity} of the law of $\mathbf c$, i.e.\ the fact there exists a \emph{(probability) measure} on $\mathbf c$'s, such that \eqref{eq_def_theta_addition} holds, is a well-known open question.

\begin{conj}\label{conj_pos}
Given any $\theta > 0$, and $N$-tuples $\mathbf a=(a_1\le \dots \le a_N)$, $\mathbf b = (b_1\le \dots \le b_N)$, there exists a probability measure on $N$-tuples $\mathbf c=(c_1\le \dots \le c_N)$ such that
$$\E B_{(c_1,\dots,c_N)}(x_1,\dots,x_N;\, \theta)= B_{(a_1,\dots,a_N)}(x_1,\dots,x_N;\, \theta) B_{(b_1,\dots,b_N)}(x_1,\dots,x_N;\, \theta)$$
holds for any $x_1, \cdots, x_N\in\C$. In other words, the distribution $\mathbf c = \mathbf a +_\theta \mathbf b$ is realized by a probability measure.
\end{conj}

When $\theta=\frac{1}{2},1,2$, the conjecture is known to be true, since we have a construction for $\mathbf{c}$ as eigenvalues of bona fide random matrices. In the limiting cases $\theta\to\infty$ and $\theta\to 0$ (here $N$ is being fixed), the distribution $\mathbf{a}+_\theta\mathbf{b}$ turns into two explicit discrete probability measures, as we outline below. Conjecture \ref{conj_pos}, as well as its generalizations, have been mentioned in \cite[Conjecture 8.3]{Stanley_Jack}, \cite{Rosler_pos}, \cite[Conjecture 2.1]{GM}, \cite[Section 1.2]{Matveev} and are believed to be true, yet we do not address it in our paper. Instead we state our results in such a way that they continue to hold even if the conjecture was wrong.

To sum up, the binary operation $+_\th$ takes two deterministic $N$--tuples $\mathbf{a}$, $\mathbf{b}$ as input and outputs a distribution $\mathbf{c}$ on $\R^N$ (though conjecturally $\mathbf{c}$ is a random $N$--tuple).
Even though there is no matrix interpretation for the operation $+_\th$ for general values of $\th>0$, it is helpful to think of $\mathbf{a}, \mathbf{b}$ and $\mathbf{c}$ as spectra of (nonexistent) self-adjoint $N\times N$ matrices.
From our experience with random matrix theory, then the following question is natural: How does $(\mathbf a,\mathbf b)\mapsto \mathbf a +_\theta \mathbf b$ behave as $N\to\infty$? While this has not been written down in any published text, there are strong reasons to believe that as long as $\theta>0$ is kept fixed, we get the same free convolution as in Theorem \ref{Theorem_Voi}.\footnote{The reasons are: widespread independence of the Law of Large Number from the value $\beta$ for the random matrix $\beta$-ensembles, cf.\ \cite{BenArous-Guionnet, Johansson_CLT, BG}; the same answer in Theorem \ref{Theorem_Voi} for three values $\theta=\tfrac{\beta}{2}=\tfrac{1}{2},1,2$; existence of $\theta$--independent observables for $\mathbf a+_\theta \mathbf b$, see \cite[Theorem 1.1]{GM}; $\theta$--independence in a discrete version of the same problem, see \cite{Huang}.} There are two boundary cases which need separate consideration: $\theta\to\infty$ and $\theta\to 0$. The former was addressed in \cite{GM}, where it was proven that for fixed $N$, the $\theta\to\infty$ limit of $\mathbf a+_\theta \mathbf b$ is a deterministic operation known as finite free convolution; it was further shown in \cite{Marcus} that as $N\to\infty$ we again recover the free convolution of Theorem \ref{Theorem_Voi}. The final case $\theta\to 0$ turns out to be very different. The $\theta=0$ version of multivariate Bessel function is a simple symmetric combination of exponents:
\begin{equation}
\label{eq_Bessel_at_0}
 B_{(q_1,\dots,q_N)}\bigl(x_1,\dots,x_N; \, 0)=\frac{1}{N!}\sum_{\sigma\in S(N)} \prod_{i=1}^N \exp\bigl( x_i q_{\sigma(i)}\bigr),
\end{equation}
where the sum goes over $N!$ different permutations of $\{1,2,\dots,N\}$. The formula \eqref{eq_Bessel_at_0} implies a transparent probabilistic interpretation: $\mathbf c=\mathbf a+_0 \mathbf b$ is obtained by choosing a permutation $\sigma\in S(N)$ uniformly at random and letting $(c_1,\dots,c_N)$ be $(a_1+b_{\sigma(1)},\dots, a_N+b_{\sigma(N)})$ rearranged in the increasing order. From this interpretation it is not hard to see that as $N\to\infty$ the operation $\mathbf a+_0 \mathbf b$ becomes the usual convolution of the empirical measures corresponding to $\mathbf a$ and to $\mathbf b$.

Hence, we see a discontinuity in the $N\to\infty$ behavior of the operation $(\mathbf a,\mathbf b)\mapsto (\mathbf a+_\theta \mathbf b)$: at $\theta=0$ the limit is described by the conventional convolution, while at $\theta>0$ the limit is described by the free convolution. This motivates us to consider an intermediate scaling regime, in which $\theta$ goes to $0$ as $N\to\infty$. This is the topic of the following Theorem \ref{Theorem_gamma_convolution}, which is proven in Section \ref{Section_applications}.

\begin{definition}
\label{Def_mom_convergence}
 We say that real random vectors $\mathbf a(N)=(a_1(N)\le a_2(N)\le \dots \le a_N(N))$ converge as $N\to\infty$ in the sense of moments, if there exists a sequence of real numbers $\{m_k\}_{k\ge 1}$, such that for any $s=1,2,\dots$ and any $k_1,k_2,\dots,k_s\in \mathbb Z_{\ge 1}$, we have:
 \begin{equation}
 \label{eq_moments_convergence}
  \lim_{N\to\infty} \E\left[\prod_{i=1}^s\left( \frac{1}{N} \sum_{j=1}^N \bigl(a_j(N)\bigr)^{k_i}\right)\right] =\prod_{i=1}^s m_{k_i}.
 \end{equation}
 In this situation we write $\displaystyle \lim_{N\to\infty} \mathbf a(N)\stackrel{m}{=} \{m_k\}_{k\ge 1}$.
\end{definition}
Note that \eqref{eq_moments_convergence} implies that the random empirical measures $\frac{1}{N} \sum_{i=1}^N \delta_{a_i(N)}$ converge as $N\to\infty$ weakly, in probability towards a deterministic measure with moments $m_k$, as long the moments problem associated with $\{m_k\}_{k\ge 1}$  has a unique solution. Also note that we can use Definition \ref{Def_mom_convergence} in the situations where positivity of the distribution of $\mathbf a(N)$ is unknown: we may interpret $\E$ in \eqref{eq_moments_convergence} as the integral with respect to the distribution of $\mathbf a(N)$.

\begin{thm}\label{Theorem_gamma_convolution} Fix $\gamma>0$ and suppose that $\theta>0$ varies with $N$ in such a way that $\theta N\to \gamma$ as $N\to\infty$. Take two sequences of independent random vectors $\mathbf a(N)$, $\mathbf b(N)$, $N=1,2,\dots$, such that
$$
   \lim_{N\to\infty} \mathbf a(N)\stackrel{m}{=} \{ m_k^{\mathbf a} \}_{k\ge 1}, \qquad  \lim_{N\to\infty} \mathbf b(N)\stackrel{m}{=} \{ m_k^{\mathbf b} \}_{k\ge 1}.
$$
In addition, assume that $\mathbf a(N)$, $\mathbf b(N)$ satisfy the tail condition of Definition \ref{df_decaying}. Then
$$
  \lim_{\begin{smallmatrix} N\to\infty\\ \theta N\to \gamma\end{smallmatrix} } \bigl(\mathbf a(N)+_\theta \mathbf b(N)\bigr) \stackrel{m}{=} \{\tilde m_k \}_{k\ge 1},
$$
where we call $\{\tilde m_k\}_{k\ge 1}$ the $\gamma$-convolution of  $\{m_k^\mathbf a\}_{k\ge 1}$ and $\{m_k^\mathbf b\}_{k\ge 1}$ denoted through
$$
  \{\tilde m_k \}_{k\ge 1}=\{ m_k^\mathbf a \}_{k\ge 1} \boxplus_\gamma \{ m_k^\mathbf b \}_{k\ge 1}.
$$
\end{thm}

\bigskip

We further investigate the $\gamma$--convolution and establish the following properties:
\begin{enumerate}
\item There exist quantities called \emph{$\gamma$--cumulants}, with the $l$th $\gamma$--cumulant $\kappa_l^{(\gamma)}$ being a homogeneous polynomial of degree $l$ in the moments $m_1, \dots, m_l$ (where $m_k$ is treated as a variable of degree $k$), such that for each $l=1,2,\dots$
    \begin{equation}
    \label{eq_convolution_cumulants}
     \kappa_l^{(\gamma)} \left[ \{ m_k^\mathbf a \}_{k\ge 1} \boxplus_\gamma \{ m_k^\mathbf b \}_{k\ge 1} \right]=
     \kappa_l^{(\gamma)} \left[\{ m_k^\mathbf a \}_{k\ge 1}\right] +\kappa_l^{(\gamma)}\left[ \{ m_k^\mathbf b \}_{k\ge 1}\right].
    \end{equation}
\item Each moment $m_k$ can be expressed as a polynomial in $\kappa_l^{(\gamma)}$, $1\le l\le k$, whose coefficients are explicit polynomials in $\gamma$ with positive integer coefficients, see Theorem  \ref{theorem_cumuls_moms}.
\item A generating function of the $\gamma$-cumulants $\kappa_l^{(\gamma)}$ is related to a generating function of the moments $m_k$ through a simple relation, see Theorem \ref{thm:mom_cums2}.
\item As $\gamma\to 0$ the $\gamma$--convolution turns into the conventional convolution (i.e., if $\{m_k^\mathbf a\}_{k\ge 1}$ and $\{m_k^\mathbf b\}_{k\ge 1}$ are moments of two independent random variables, then $\lim_{\gamma\to 0}\,\{m_k^\mathbf a \}_{k\ge 1} \boxplus_\gamma \{ m_k^\mathbf b \}_{k\ge 1}$ gives moments of their sum). After proper renormalization the $\gamma$--cumulants turn into conventional cumulants, see Section \ref{Section_limit_to_0}.
\item As $\gamma\to \infty$ the $\gamma$--convolution turns into the free convolution of Theorem \ref{Theorem_Voi}. After proper renormalization the $\gamma$--cumulants turn into the free cumulants, see Section \ref{Section_limit_to_infinity}.
\end{enumerate}

\begin{remark}
Given a probability measure $\mu$ with finite moments $\{m_k\}_{k\ge 1}$, we say that the corresponding $\{ \ka_l^{(\ga)} \}_{l\ge 1}$ are the \emph{$\ga$-cumulants of $\mu$}. It is known that the only probability measures with finitely many nonzero classical cumulants are Dirac delta masses and Gaussian distributions, see e.g. \cite[Thm.\ 7.3.3]{L}.  There is no such result in free probability (see \cite[Thm.\ 2]{BeVo}): the semicircle distribution is a free probability analogue of the Gaussian distribution, but there are also very different measures with finitely many non-zero free cumulants. In our setting, the analogue of Gaussian/semicircle distributions are the measures for which only the first two $\gamma$-cumulants are nonzero, see Section \ref{Section_GbE} and Example \ref{gauss_exam}. Therefore, a natural open question is whether there are more examples of probability measures with finitely many nonzero $\gamma$-cumulants.
\end{remark}

\begin{remark}
 We do not discuss in this text the \emph{microscopic limits} of  $\mathbf a(N)+_\theta \mathbf b(N)$ as $N\to\infty$, $\theta N\to \gamma$, i.e.\ the asymptotic questions in which individual eigenvalues remain visible in the limit. Yet, we expect to see the Poisson point process in the bulk of the spectrum, as hinted by general universality considerations and the $\theta\to 0$ asymptotic results in \cite{KS,AD,BGP}.
\end{remark}

\subsection{Law of Large Numbers through Bessel generating functions}

Let us now outline the main technical tool, underlying the proof of Theorem \ref{Theorem_gamma_convolution} and other asymptotic results mentioned at the end of Section \ref{Section_Overview}.

Suppose that $\mathbf q=(q_1\le q_2\le\dots \le q_N)$ is a random $N$--tuple of reals. We define its Bessel generating function (BGF) through:
$$
 G_\theta(x_1,\dots,x_N; \mathbf q)=\E_{\mathbf q}\left[ B_{(q_1,\dots,q_N)}(x_1,\dots, x_N;\, \theta)\right].
$$
Our main result, Theorem \ref{thm_small_th}, establishes an equivalence of the following two conditions for random sequences $\mathbf q(N)=(q_1(N),\dots,q_N(N))$ as $N\to\infty$ and $\theta\to 0$ in such a way that $\theta N\to \gamma$:
\begin{enumerate}
\item Partial derivatives of arbitrary order in $x_1$ of $\ln\bigl(G_\theta(x_1,\dots,x_N; \mathbf q(N))\bigr)$ at $(0,\dots,0)$ converge to prescribed limits and partial derivatives in two (or more) different variables converge to $0$.
\item Random vectors $\mathbf q(N)$ converge in the sense of moments, as in Definition \ref{Def_mom_convergence}.
\end{enumerate}
The same theorem also establishes explicit polynomial formulas connecting the limiting value of the partial derivatives to the limiting values of the moments. The benefit of Theorem \ref{thm_small_th} is that it allows us to convert probabilistic information about $\mathbf q(N)$ into the analytic information about partial derivatives of its BGF and vice versa. For instance, Theorem \ref{Theorem_gamma_convolution} is then proven by three straightforward applications of Theorem \ref{thm_small_th}: to $\mathbf q(N)=\mathbf a(N)$, to $\mathbf q(N)=\mathbf b(N)$, and to $\mathbf q(N)=\mathbf a(N)+_\theta \mathbf b(N)$. This and several other applications of Theorem \ref{thm_small_th} are detailed in Section \ref{Section_applications}.

Similar methods to the ones used in our proof of Theorem \ref{thm_small_th} have led to recent results in the literature, and even though our Theorem \ref{thm_small_th} bears resemblance to these other results, there are important differences. For example, in \cite{BuG3} Bufetov and the third author developed a theory of Schur generating functions (SGF) for discrete $N$--particle systems as $N\to\infty$ (see also \cite{Huang} for an extension): they show that asymptotic information on partial derivatives of logarithms of SGF is in correspondence with asymptotic information on the moments in Law of Large Numbers as in Definition \ref{Def_mom_convergence} and with covariances in a version of the Central Limit Theorem for global fluctuations. This is different from our Theorem \ref{thm_small_th}: on the analytic side \cite{BuG3} requires more refined control on partial derivatives and on the probabilistic side \cite{BuG3} requires Central Limit Theorems in addition to Laws of Large Numbers.

In another similar framework related to multiplication of random matrices \cite{GS} established a statement in one direction: control on partial derivatives implies the Law of Large Numbers and Central Limit Theorem, but in that framework a statement in the opposite direction remains out of reach.

Going further, we show in Section \ref{Section_Appendix_LLN} that an analogue of Theorem \ref{thm_small_th} with fixed (rather than tending to $0$) $\theta$ is wrong: there is no direct correspondence between partial derivatives of the logatithm of BGF and asymptotics of moments; one probably needs to use in such situation more complicated (and not yet understood) combinations of mixed partial derivatives in several variables. Thus, Theorem \ref{thm_small_th} is not an extension of the results of previous papers, but rather a brand new statement.


\subsection{Connection to $\theta=\tfrac{\beta}{2}\to \infty$ limits}

One intriguing aspect of general $\beta$ random matrix theory is existence of dualities between parameters $\theta$ and $1/\theta$ (i.e.\ between $\beta$ and $4/\beta$). In the theory of symmetric polynomials such a duality manifests through the existence of an automorphism of the algebra of symmetric functions, which transposes the label of Jack symmetric polynomials and simultaneously inverts $\theta$, see \cite[Section 3]{Stanley_Jack}. In the study of classical ensembles of random matrices the duality appears as a symmetry in expectations of power sums of eigenvalues, see, e.g., \cite[Section 2.1]{DE}, \cite[Section 4.4]{FD}, \cite{For_du}, and references therein.

In our context, the duality suggests to look for a relation between $\theta\to 0$ limits of our paper and $\theta\to\infty$ limits. While this relation is not yet fully understood, we observe it in two forms.

First, the limit of the empirical measures of Gaussian $\beta$--ensembles as $\beta\to 0$, $N\to\infty$, $\beta N\to 2\gamma$ turns out to coincide with the orthogonality measure of the associated Hermite polynomials, see Remark \ref{Remark_aHerm}. Simultaneously, the same polynomials play an important role in the study of centered fluctuations of Gaussian $\beta$--ensembles as $\beta\to \infty$ with $N$ kept fixed, see \cite[Section 4.5]{Gorin_Klept} and \cite{AHV}.

Second, let us fix $N=d$ and send $\theta\to\infty$. \cite[Theorem 1.2]{GM} claims that in this regime the operation $\mathbf a+_\theta \mathbf b$ turns into the finite free convolution, which is a deterministic binary operation on $d$--tuples of real numbers. Further, \cite{AP} introduced for each $d$ a family of $d$ finite free cumulants $\kappa^{\mathrm{ff}}_{1;d}, \kappa^{\mathrm{ff}}_{2;d},\dots,\kappa^{\mathrm{ff}}_{d;d}$, which depend on a $d$--tuple of real numbers and play the same role for the finite free convolution, as our $\gamma$--cumulants play for the $\gamma$--convolution. Comparing the generating function of finite free cumulants from \cite{AP}, \cite{Marcus}, with the generating function of $\gamma$-cumulants of our Theorem \ref{thm:mom_cums2}, one sees\footnote{One should compare \cite[(3.1), (4.2)]{AP} with our pair of equations \eqref{eq_cums_moments_2} and notice that the conventions are slightly different: $(d)_n$ is a falling factorial in \cite{AP} and $(\gamma)_n$ is a rising factorial in our work.
One can also directly compare the formulas for the first four cumulants of \eqref{ex_c_to_m} and \eqref{ex_m_to_c} with similar formulas above Corollary 4.3 in the journal version of \cite{AP}. We are grateful to Octavio Arizmendi and Daniel Perales for pointing this connection to us.}that upon setting $\gamma=-d$, they are very similar and only differ by normalizations, see Section \ref{Section_gen_functions} for more details. However, it is important to note that in our setting $\gamma>0$, while in the setting of \cite{AP}, \cite{Marcus}, $d$ is a positive integer and, thus, $-d$ is a negative integer. Hence, a correct point of view is that our $\gamma$--cumulants and the finite free cumulants are analytic continuations of each other. It would be interesting to see whether this observation can be used to produce new formulas for finite free cumulants along the lines of our Theorem \ref{theorem_cumuls_moms}.

\subsection*{Acknowledgements}  The authors would like to thank Alexey Bufetov and Greta Panova for helpful discussions. We are thankful to Maciej Do\l \c ega for pointing us to the articles \cite{D}, \cite{BDEG}, and for sending us a draft of a new version of the latter paper. We thank Octavio Arizmendi and Daniel Perales for directing us to their work \cite{AP}. We are grateful to two anonymous referees for their feedback.
The work of V.G.\ was partially supported by NSF Grants DMS-1664619, DMS-1949820, by BSF grant 2018248, and by the Office of the Vice Chancellor for Research and Graduate Education at the University of Wisconsin--Madison with funding from the Wisconsin Alumni Research Foundation.

\section{Bessel generating functions}

We define here the Bessel generating function of a probability measure on $\R^N$ --- this is a one-parameter generalization of the characteristic function (or Laplace transform) of a probability measure.
The real parameter $\th$ is assumed to be positive, with $\th\to 0$ corresponding to the usual characteristic function. In this section, $\theta>0$ remains fixed.

\subsection{Difference and differential operators}

\label{Section_operators}

We work with functions of $N$ variables $x_1,\dots,x_N$.
Denote the operator that permutes the variables $x_i$ and $x_j$ by $s_{i, j}$. For instance,
$$
 [s_{1,2} f](x_1,x_2,x_3,x_4,\dots,x_N)=f(x_2,x_1,x_3,x_4,\dots,x_N).
$$
Define the \emph{Dunkl operators} by
\begin{equation}\label{dunkl_ops}
\D_i := \frac{\pa}{\pa x_i} + \theta\sum_{j : j\neq i}{\frac{1}{x_i - x_j}\circ (1 - s_{i, j})},\quad i = 1, \dots, N.
\end{equation}
These operators were introduced in \cite{Du}; see also \cite{Ki_lect,Rosler,Et_lect} for further studies. Their key property is commutativity:
$$\D_i\D_j = \D_j\D_i,\quad i, j = 1, \dots, N.$$
We often work with symmetrized versions of the Dunkl operators:
\begin{equation*}
\P_k := (\mathcal{D}_1)^k + \dots + (\mathcal{D}_N)^k,\quad k\in\Z_{\geq 1}.
\end{equation*}
Let $U\subseteq\C^N$ be any domain which is symmetric with respect to permutations of the axes.
If $f$ is a holomorphic function on $U$, then $\D_i f$ and $\P_k f$ are both well-defined and holomorphic on $U$.

\smallskip

We also need the \emph{degree-lowering operators} $d_1, \dots, d_N$, which are defined on monomials by
\begin{equation}
\label{eq_lowering_operator}
d_i(x_1^{r_1}\cdots x_N^{r_N}) := \begin{cases} x_1^{r_1}\cdots x_i^{r_i - 1}\cdots x_N^{r_N}, & \text{if}\ r_i\in\Z_{\geq 1}, \\ 0, & \text{if}\ r_i = 0, \end{cases}
\end{equation}
and extended by linearity to the space of polynomials of $N$ variables.
They can be further extended to the ring of germs of analytic functions at the origin $(0, \cdots, 0)\in\C^N$.

\subsection{Multivariate Bessel functions}\label{sec:bessel}

A central role in our studies is played by the simultaneous eigenfunctions of the operators $\P_k$ known as \emph{multivariate Bessel functions}.
They are given by very explicit formulas, which we describe next.

For each $N=1,2,\dots$, a \emph{Gelfand--Tsetlin pattern of rank $N$} is an array $\{y_{i}^k\}_{1\le i \le k \le N}$ of real numbers satisfying
$y^{k+1}_i\le y^{k}_i \le y^{k+1}_{i+1}$. Denote by $\GT_N$ the space of all Gelfand--Tsetlin patterns of rank $N$.

\begin{definition}\label{def_betacorner} Fix $\theta>0$. The \emph{$\th$-corners process with top row $a_1<\dots<a_N$}
is the probability distribution on the arrays
$\{y^k_i\}_{1\leq i\leq k\leq N}\in \GT_N$, such that $y^N_i=a_i$, $i=1,\dots,N$, and
the remaining $N(N-1)/2$ coordinates have the density
\begin{equation}
\label{eq_beta_corners_def}
\frac{1}{Z_{N; \th}} \cdot
\prod_{k=1}^{N-1} \left[\prod_{1\le i<j\le k} (y_j^k-y_i^k)^{2-2\theta}\right] \cdot \left[\prod_{a=1}^k \prod_{b=1}^{k+1}
|y^k_a-y^{k+1}_b|^{\theta-1}\right],
\end{equation}
where $Z_{N; \th}$ is the normalization constant:
\begin{equation}
\label{eq_normalization}
Z_{N; \th} =\left[\prod_{k=1}^N \frac{ \Gamma(\theta)^k}{\Gamma(k\theta)}\right] \cdot \prod_{1\le i < j \le N}
 (a_j-a_i)^{2\theta-1}.
\end{equation}
\end{definition}
\begin{remark}
 By taking  limits (in the space of probability measures on $\GT_N$), we can allow equalities and extend the definition to arbitrary $a_1\le a_2\le\dots\le a_N$.
\end{remark}
\begin{remark}
The distribution \eqref{eq_beta_corners_def} is the joint law of eigenvalues of principal corners of Hermitian conjugation-invariant real/complex/quaternion matrices at $\th=\frac{1}{2},1,2$, respectively, see \cite{Ner}.  This connects Definition \ref{def_betacorner} to the Laplace-Fourier point of view of Section \ref{Section_addition_intro}.
\end{remark}
\begin{remark}
The calculation of the normalization constant for a general $\th > 0$ is contained in \cite{Ner} (the author there does far more general calculations; see also \cite[Lem. 2.1]{C} for a short derivation of $Z_{N; \th}$ from Anderson's integral identity \cite{And}).
\end{remark}

\begin{definition} \label{Definition_Bessel_function}
The \emph{multivariate Bessel function} $B_{(a_1, \ldots, a_N)}(x_1, \ldots, x_N; \theta)$ is defined as the following (partial) Laplace transform of the
$\th$-corners process with top row $(a_1,\dots,a_N)$ from Definition \ref{def_betacorner}:
\begin{equation}\label{eq_Bessel_combinatorial}
 B_{(a_1,\dots,a_N)}(x_1,\dots,x_N;\,\theta)= \E_{\{y^k_i\}}\!\left[\exp\left(\sum_{k=1}^{N} x_k
 \cdot \left(\sum_{i=1}^{k} y_i^k-\sum_{j=1}^{k-1} y_j^{k-1}\right)  \right) \right]\!.
\end{equation}
The function $B_{(a_1, \dots, a_N)}(x_1, \dots, x_N)$ is defined for any reals $a_1<\dots <a_N$ and any complex numbers $x_1, \dots, x_N$.
\end{definition}

Often, we will abbreviate multivariate Bessel function as MBF.

It follows from the definition that
\begin{equation*}
B_{(a_1, \dots, a_N)}(0, \dots, 0; \th) = 1.
\end{equation*}

Our definition is called the \emph{combinatorial formula} for the multivariate Bessel functions; to our knowledge, the formula \eqref{eq_Bessel_combinatorial} first appeared in \cite{GK}. There are several alternative definitions of these functions.
For example, from the algebraic combinatorics point of view, they can be defined as limits of (properly normalized) Jack symmetric polynomials. Then \eqref{eq_Bessel_combinatorial} is a limit of the combinatorial formulas for the Jack polynomials, cf.\ \cite[Section 4]{Ok_Olsh_shifted_Jack}.

The MBF $B_{(a_1, \cdots, a_N)}(x_1, \cdots, x_N; \th)$, which was defined for ordered tuples $a_1<\dots<a_N$, can be extended to weakly ordered tuples $a_1\le \dots\le a_N$ by continuity: there is no singularity on the diagonals $a_i=a_j$. In fact, much more is true: $B_{(a_1,\dots,a_N)}(x_1,\dots,x_N;\th)$ admits an analytic continuation on the $2N+1$ variables $a_1, \dots, a_N$, $x_1, \dots, x_N, \th$, to an open subset of $\C^{2N+1}$ containing $\{(a_1, \cdots, a_N, x_1, \cdots, x_N, \th)\in\C^{2N+1} \mid \Re\th\ge 0\}$; see \cite{O}.
In particular, for a fixed $\th>0$, the MBF $B_{(a_1,\dots,a_N)}(x_1,\dots,x_N;\th)$ is an entire function on the variables $a_1, \cdots, a_N, x_1, \cdots, x_N$.

Another important property is that the MBF $B_{(a_1,\dots,a_N)}(x_1,\dots,x_N;\,\theta)$ is \emph{symmetric} in its arguments $x_1,\dots,x_N$ --- this is an immediate consequence of the fact that the MBFs are limits of properly normalized Jack symmetric polynomials, see e.g. \eqref{limit_to_bessel} below.
In the particular case $\theta=1$, the symmetry is also transparent from the following determinantal formula, which arises as the evaluation of the Harish-Chandra-Itzykson-Zuber (HCIZ) integral:
\begin{equation}\label{eq_Bessel_1}
  B_{(a_1,\dots,a_N)}(x_1,\dots,x_N;\,1)= 1!\cdot 2! \cdots (N-1)! \cdot  \frac{\det\bigl[ e^{a_i x_j}\bigr]_{i,j=1}^N}{\prod_{i<j} (x_i-x_j)(a_i-a_j)}.
\end{equation}

A link of MBF to the operators of Section \ref{Section_operators} is given by the following statement.

\begin{thm}[\cite{O}]\label{thm:opdam}
For each $k=1,2,\dots,$ and each $N$--tuple of reals $a_1\le a_2 \le \dots \le a_N$,
\begin{equation}\label{eqn:hypersystem}
\P_k B_{(a_1,\dots,a_N)} = \left(\sum_{i=1}^N{a_i^k}\right)\cdot B_{(a_1,\dots,a_N)}.
\end{equation}
\end{thm}

\subsection{Bessel generating functions} \label{Section_BGF}
Let $\M_N$ be the convex set of Borel probability measures on ordered $N$--tuples $a_1\le a_2\le \dots\le a_N$ of real numbers.

\begin{df}
The \emph{Bessel generating function} (or BGF) of $\mu\in\M_N$ is defined as a function of the variables $x_1, \dots, x_N$ given by:
\begin{equation}\label{BGF_def}
G_\th(x_1, \dots, x_N; \mu) := \int_{a_1\le a_2\le \dots \le a_N}{B_{(a_1, \dots, a_N)}(x_1, \dots, x_N; \theta)\mu(\d a_1, \dots, \d a_N)}.
\end{equation}
\end{df}

Because the MBFs $B_{(a_1,\dots,a_N)}(x_1,\dots,x_N;\theta)$ are symmetric functions on the variables $x_1,\dots, x_N$, so is $G_\th(x_1,\dots,x_N; \mu)$. Moreover,
$$
G_\th(0,\dots,0; \mu) = 1,
$$
as follows from $\mu$ being a probability measure and $B_{(a_1,\dots,a_N)}(0,\dots,0;\theta)=0$.

It will be important for us to assume that a BGF is defined in a \emph{complex} neighborhood of $(0,\dots,0)$. Unfortunately, this property fails for general measures, hence we need to restrict the class of measures that we deal with.\footnote{It is plausible that many of the results of our text extend to the situations where this restrictive condition fails.}

\begin{df}\label{df_decaying}
We say that a measure $\mu\in\M_N$ is \emph{exponentially decaying} with exponent $R>0$, if
\begin{equation}\label{int_bounded}
\int_{a_1\le a_2\le\dots\le a_N}{e^{N R \max_i |a_i| }\mu(\d a_1, \dots, \d a_N)} < \infty.
\end{equation}
\end{df}

\begin{lemma}\label{bgf_good}
If $\mu\in\M_N$ is exponentially decaying with exponent $R>0$, then the integral \eqref{BGF_def} converges for all $(x_1,\dots,x_N)$ in the domain
\begin{equation*}
\Omega_R := \left\{ (x_1, \dots, x_N)\in\C^N : |\Re x_i| < R,\ i = 1, \dots, N \right\},
\end{equation*}
and defines a holomorphic function in this domain.
\end{lemma}
\begin{proof}
 Note that if $\{y_i^k\}\in\GT_N$ satisfies $y_i^N=a_i$, $i=1,\dots,N$, then due to interlacing inequalities, for each $k=1,2,\dots,N$ we have
 $$
  \left|\sum_{i=1}^k y_i^k-\sum_{i=1}^{k-1} y_i^{k-1}\right|\le \max\left(|y_1^k|, |y_k^k|\right) \le \max_{i} |a_i|.
 $$
Hence, the integrand in the definition of the multivariate Bessel function \eqref{eq_Bessel_combinatorial} is upper bounded by
$$
 \exp\left( \sum_{j=1}^N |\Re x_j| \max_i |a_i|\right),
$$
which implies
$$
\left|B_{(a_1,\dots,a_N)}(x_1,\dots,x_N; \th)\right|\le \exp\left(NR \max_i |a_i|\right)\!, \quad (x_1,\dots,x_N)\in \Omega_R.
$$
Hence, \eqref{int_bounded} implies convergence of the integral \eqref{BGF_def} in $\Omega_R$.

It remains to check holomorphicity of \eqref{BGF_def} as a function of $x_1,\dots,x_N$. This readily follows from holomorphicity of $B_{(a_1,\dots,a_N)}(x_1,\dots,x_N; \th)$. Indeed, $G_\th(x_1, \dots, x_N; \mu)$ is continuous as a uniformly convergent integral of continuous functions. Thus, by Morera's theorem, the holomorphicity follows from vanishing of the integrals over closed contours. The latter vanishing can be deduced by swapping the integrations using the Fubini's theorem and using vanishing of the similar integrals for $B_{(a_1,\dots,a_N)}(x_1,\dots,x_N; \th)$.
\end{proof}

The BGFs have recently been used in connection to problems in random matrix theory, see \cite{C}, \cite{GS}.
However, the BGF is not a new invention. The formula \eqref{BGF_def} is essentially the definition of (a symmetric version of) the  \emph{Dunkl transform}, a one-parameter generalization of the Fourier transform; this is a rich and well-studied subject, see e.g. the survey \cite{A} and references therein.

The next two propositions will be important in our developments.
\begin{prop}\label{ops_1}
Let $k\in\Z_{\geq 1}$ and let $\mu\in\M_N$ be an exponentially decaying measure. Then
$$\bigl[\P_k\, G_\th(x_1, \dots, x_N; \mu)\bigr]_{\setzeroes} = \E_{\mu}\!\left[\sum_{j=1}^N (a_j)^k\right],$$
where $(a_1, \dots, a_N)\in\R^N$ is random and $\mu$-distributed on the right-hand side.
\end{prop}
\begin{proof} We apply $\P_k$ to \eqref{BGF_def} under the sign of the integral, use the eigenrelation of Theorem \ref{thm:opdam} and the normalization $B_{(a_1,\dots,a_N)}(0,\dots,0; \theta)=1$. We can exchange the order between the operator $\P_k$ and the integral because $\mu$ is exponentially decaying.
\end{proof}

The following generalization of Proposition \ref{ops_1} is proved in the same way.

\begin{prop}\label{proposition_moments_through_operators}
Let $k_1, \dots, k_s\in\Z_{\geq 1}$ and let $\mu\in\M_N$ be an exponentially decaying measure. Then
\begin{equation}\label{ops_2_eqn}
\left( \prod_{i=1}^s{\P_{k_i}}\!\right) G_\th(x_1, \dots, x_N; \mu)\Bigr|_{\setzeroes} = \E_{\mu}\!\left[ \prod_{i=1}^s \left(\sum_{j=1}^N (a_j)^{k_i}\right) \right].
\end{equation}
\end{prop}

\medskip

Observe that the pairwise commutativity of the Dunkl operators implies the pairwise commutativity of the operators $\P_k$, $k\in\Z_{\geq 1}$.
As a result, the order of application of the operators $\P_{k_i}$ in the left-hand side of \eqref{ops_2_eqn} does not matter.

\subsection{Extension to distributions}

Ultimately, we treat Bessel Generating Functions as a tool for studying symmetric probability measures on $\mathbb R^N$ (which can be identified with probability measures on ordered $N$--tuples $a_1\le a_2\le\dots\le a_N$). One of the applications that we have in mind is to use them for the study of addition of independent general $\beta$ random matrices. While it is conjectured that the spectrum of such sum should be described by a probability measure, it is not proven yet: we only rigorously know that the spectrum can be described as a generalized function or distribution (the technical problem is in proving positivity;  see \cite{Tri}, \cite[Section 3.6]{A}). In order to avoid the necessity to rely on the positivity conjectures, we explain in this section that the framework of Bessel generating functions can be extended to objects more general than probability measures.

Let $\mu$ be a \emph{distribution} on $\mathbb R^N$ with coordinates $(a_1,\dots,a_N)$, i.e.\ $\mu$ is an element of the dual space to the space of compactly supported infinitely--differentiable \emph{test-functions}.\footnote{The space of test-functions $f$ is equipped with a topology: $f^{n}$ converge to $0$ as $n\to\infty$, if the supports of all these functions belong to the same compact set and all partial derivatives of $f^{n}$ converge to $0$ uniformly.} $\mu$ is said to be symmetric if for any test-function $f$ and any permutation $\sigma$:
$$
 \langle \mu, f(a_1,\dots,a_N)\rangle = \langle \mu, f(a_{\sigma(1)},\dots,a_{\sigma(N)})\rangle,
$$
where we use the notation $\langle \mu,f\rangle$ for the value of the functional $\mu$ on the test-function $f$.

\begin{df}\label{bgf_dist}
For a symmetric distribution (generalized function) $\mu$ on $\mathbb R^N$, its \emph{Bessel generating function} (or BGF) is a function of $(x_1,\dots,x_N)$ given by
\begin{equation}\label{BGF_def_gen}
G_\theta(x_1, \dots, x_N; \mu) :=\frac{1}{N!} \left\langle \mu, B_{(a_1, \dots, a_N)}(x_1, \dots, x_N; \theta)\right\rangle,
\end{equation}
where in the right-hand side $B_{(a_1, \dots, a_N)}(x_1, \dots, x_N; \theta)$ is treated as a test-function in $(a_1,\dots,a_N)$ variables with parameters $(x_1,\dots,x_N)$.
\end{df}

There are two tricky points in this definition. First, the $N$--tuple $(a_1,\dots,a_N)$ was ordered in the original definition of the multivariate Bessel function, whereas $\mu$ is a distribution on $\mathbb R^N$. However, multivariate Bessel functions can be extended to $\mathbb R^N$ in a symmetric way. The $\frac{1}{N!}$ prefactor is introduced to match the integral over \emph{ordered} $N$--tuples in \eqref{BGF_def} with distribution on whole $\R^N$ in \eqref{BGF_def_gen}.

More importantly, for general distributions $\mu$ \eqref{BGF_def_gen} is not defined, since $B_{(a_1, \dots, a_N)}(x_1, \dots, x_N; \theta)$ is not compactly supported and therefore not a valid test function. Hence, one needs to impose some growth conditions similar to Definition \ref{df_decaying} on $\mu$, in order to make \eqref{BGF_def_gen} meaningful. Rather than exploring the full generality, let us only consider the case of compactly supported $\mu$ (which means that $\mu$ vanishes on any test
function whose support does not intersect a certain compact set), which is all we need for our application.
For compactly supported distributions $\mu$, Definition \ref{bgf_dist} is well-posed and in fact the pairing $\langle \mu, f\rangle$ makes sense for any infinitely--differentiable function $f$. In this text, we will be interested in compactly supported distributions $\mu$ of total mass equal to $N!$, meaning that $\langle \mu, \mathbf{1}\rangle = N!$, where $\mathbf{1}$ is the test function on $\R^N$ that is identically equal to $1$; in this case, $G_\th(0, \cdots, 0; \mu) = 1$.

\begin{proposition} \label{Proposition_BGF_dist}
 Suppose that $\mu$ is a symmetric compactly supported distribution on $\mathbb R^N$. Then its BGF $G_\theta(x_1, \dots, x_N; \mu)$ is an entire function. We also have
 \begin{equation}\label{eq_expectation_operator_general}
\left( \prod_{i=1}^s{\P_{k_i}}\right)G_\theta(x_1, \dots, x_N; \mu)\Bigr|_{\setzeroes} =\frac{1}{N!}\left\langle \mu,\, \prod_{i=1}^s \left(\sum_{j=1}^N (a_j)^{k_i}\right) \right\rangle.
\end{equation}
\end{proposition}
\begin{proof}
 Each compactly supported distribution can be identified with a (higher order) derivative of a compactly supported continuous function  (see, e.g., \cite[Section 6]{Rudin}). Hence, we have
 $$
  G_\theta(x_1, \dots, x_N; \mu)=\int_{\mathbb R^N} \frac{\partial^{|\alpha|}}{\partial (a_i)^\alpha} \left[B_{(a_1,\dots,a_N)}(x_1,\dots,x_N; \theta)\right] f(a_1,\dots,a_N)\, \d a_1\cdots \d a_N,
 $$
 where $f$ is a compactly supported continuous function and $\frac{\partial^{|\alpha|}}{\partial (a_i)^\alpha}$ is a partial derivative of multi-index $\alpha$ in variables $(a_1,\dots,a_N)$. It remains to repeat the arguments of Section \ref{Section_BGF}.
We remark that the condition of $\mu$ being exponentially decaying, required by Proposition \ref{proposition_moments_through_operators}, has been substituted by the condition of $\mu$ being compactly supported.
\end{proof}

\section{Statements of the Main Results}
\label{Section_main_results}

Throughout this section, we fix a real parameter $\ga > 0$.

\subsection{Law of Large Numbers at high temperature}

Let $\{\mu_N\}_{N \geq 1}$ be a sequence of exponentially decaying probability measures, such that $\mu_N\in\M_N$ for each $N$, that is, $\mu_N$ is a probability measure on $N$--tuples $a_1\le a_2\le \dots \le a_N$.
Alternatively, we can assume that each $\mu_N$ is a compactly supported symmetric distribution on $\mathbb R^N$ of total mass (i.e., the pairing against test function $1$) equal to $N!$. Denote their Bessel generating functions by
$$G_{N; \th}(x_1, \dots, x_N) := G_\th(x_1, \dots, x_N; \mu_N).$$
By the results from the previous section, each $G_{N; \th}(x_1, \dots, x_N)$ is holomorphic in a neighborhood of the origin and satisfies $G_{N; \th}(0, \dots, 0) = 1$. Thus, for each $N$, the logarithm $\ln(G_{N; \th})$ is a well-defined holomorphic function in a neighborhood of $(0, \cdots, 0)\in\C^N$, and
\begin{equation*}
\ln(G_{N; \th}) \bigr|_{\setzeroes} = 0.
\end{equation*}

We are interested in the interplay between the partial derivatives of $\ln(G_{N; \th})$ at the origin and asymptotic properties of random $\mu_N$--distributed\footnote{In our wordings we stick to the situation when $\mu_N$ are bona fide probability measures. If they are distributions (i.e.\ generalized functions possibly without any positivity), then all the random variables produced from them should be interpreted in formal sense: the laws of such random variables can be identified with expectations of various smooth functions of them, which are readily computed as pairings of $\mu_N$ with appropriate test functions. (One also should divide by $N!$ to adjust for differences between ordered and arbitrary $N$--tuples.)} $N$--tuples $(a_1,\dots,a_N)$. We deal with the latter through the random variables
$$
p_k^N := \frac{1}{N}\sum_{i=1}^N (a_i)^k, \qquad (a_1,\dots,a_N)\text{ is }\mu_N\text{--distributed.}
$$

\begin{df}[LLN--satisfaction]\label{Definition_LLN_sat_ht}
We say that a sequence $\{\mu_N\}_{N \geq 1}$ \emph{satisfies a Law of Large Numbers} if there exist real numbers $\{m_{k}\}_{k\geq 1}$ such that for any
$s=1,2,\dots$ and any $k_1, \dots, k_s\in\Z_{\geq 1}$, we have
$$\lim_{N\to\infty} \E_{\mu_N} \prod_{i=1}^s  p_{k_i}^N= \prod_{i=1}^{s}  m_{k_i}.$$
\end{df}

\begin{remark}\label{rem_uniqueness}
Consider the empirical measure of $(a_1,\dots,a_N)$ given by $\frac{1}{N} \sum_{i=1}^N \delta_{a_i},$ where $\delta_x$ is the Dirac delta mass at $x\in\mathbb R$.
Since the $N$-tuples $(a_1,\dots,a_N)$ are random, their empirical measures are random probability measures on $\mathbb R$.
Under mild technical conditions (uniqueness of a solution to the moments problem, which holds whenever the numbers $m_k$ do not grow too fast, see, e.g., \cite[Section VII.3]{Feller}), LLN--satisfaction implies that these measures converge weakly, in probability, to a non-random measure whose moments are $m_1,m_2,\cdots$.
\end{remark}

\begin{df}[$\ga$-LLN--appropriateness]\label{Definition_LLN_appr_ht}
We say that the sequence $\{\mu_N\}_{N \geq 1}$ is \emph{$\gamma$-LLN--appropriate} if there exists a sequence of real numbers $\{\ka_l\}_{l\geq 1}$ such that
\begin{enumerate}[label=(\alph*)]
\item $\displaystyle \lim_{\begin{smallmatrix} N\to\infty,\, \theta \to 0\\ \theta N\to \gamma \end{smallmatrix}}
\frac{\pa^l}{\pa x_i^l} \ln{(G_{N; \th})}\Bigr|_{\setzeroes} =  (l-1)!\cdot \ka_l$,\quad for all $l, i\in\Z_{\geq 1}$.

\item $\displaystyle \lim_{\begin{smallmatrix} N\to\infty,\, \theta \to 0\\ \theta N\to \gamma \end{smallmatrix}}\left.\frac{\partial}{\partial x_{i_1}}\cdots\frac{\partial}{\partial x_{i_r}}\ln{(G_{N; \th})}\right|_{\setzeroes} = 0$,\quad for all $r\ge 2$, and $i_1, \dots, i_r\in\Z_{\geq 1}$ such that the set $\{i_1, \dots, i_r\}$ is of cardinality at least two.
\end{enumerate}
\end{df}
\begin{remark}
Because the BGF $G_{N; \th}(x_1, \cdots, x_N)$ is symmetric on the variables $x_1, \cdots, x_N$, the condition (a) is equivalent to:

(a') $\displaystyle \lim_{\begin{smallmatrix} N\to\infty,\, \theta \to 0\\ \theta N\to \gamma \end{smallmatrix}}
\frac{\pa^l}{\pa x_1^l} \ln{(G_{N; \th})}\Bigr|_{\setzeroes} =  (l-1)!\cdot \kappa_l$,\quad for all $l\in\Z_{\geq 1}$.

\noindent Likewise, we could also simplify condition (b).
\end{remark}
\begin{remark}
 Suppose that
 $$
\displaystyle \lim_{\begin{smallmatrix} N\to\infty,\, \theta \to 0\\ \theta N\to \gamma \end{smallmatrix}}
{\frac{\partial}{\partial z} \ln(G_{N; \th}(z,0,\dots,0))} = g(z),
 $$
 uniformly over $z$ in a complex neighborhood of $0$. Then $\ka_l$ are the Taylor coefficients of $g$, that is, $$g(z)=\sum_{l=1}^{\infty} {\ka_l z^{l-1}}.$$
\end{remark}

To state the main theorem, we use the language of formal power series in a formal variable $z$, namely series of the form
\begin{equation*}
 h_0 + h_1 z + h_2 z^2 +\cdots.
\end{equation*}

\begin{df}\label{df:ops}
Let $\R[[z]]$ be the space of formal power series in $z$ with real coefficients.
Let $a(z)=a_0+a_1 z+a_2z^2 + \cdots$ be any power series in $\R[[z]]$.
We define three operators in $\R[[z]]$ by their action on a generic element $h(z) = h_0 + h_1 z + h_2 z^2+\cdots\in\R[[z]]$, as follows.
\begin{itemize}
\item Derivation operator $\partial$:
$$
 \partial(h_0 + h_1 z + h_2 z^2+\cdots) := h_1 + 2 h_2 z + 3 h_3 z^2 + \cdots.
$$
\item Lowering operator $d$:
$$
 d(h_0 + h_1 z + h_2 z^2+\cdots) := h_1 + h_2 z + h_3 z^2 + \cdots.
$$
\item Multiplication operator $*_a$:
$$
 *_a(h(z)) := h(z) a(z).
$$
\end{itemize}
\end{df}

\begin{df}\label{def_R_map}
Define the map $\Tcm : \R^{\infty} \to \R^{\infty}$ that takes as input a countable real sequence $\{\ka_l\}_{l\ge 1}$ and outputs the countable real sequence $\{m_k\}_{k\ge 1}$ by means of the relations
\begin{equation}\label{eq_moments_through_f_cumulants}
m_k = [z^0] (\partial + \gamma d + *_g)^{k-1}(g(z)),\quad k=1, 2, \cdots,
\end{equation}
where $[z^0]$ is the constant term of the expression following it and
\begin{equation*}
g(z) = \sum_{l=1}^{\infty} {\ka_l z^{l-1}}.
\end{equation*}
\end{df}

\medskip

For notation purposes, in the remainder of the paper the input of the map $\Tcm$ is denoted by $\{\ka_l\}_{l\ge 1}$ and the output is denoted by $\{m_k\}_{k\ge 1}$.
Whenever $\Tcm(\{\ka_l\}_{l\geq 1}) = \{m_k\}_{k\ge 1}$, the quantities $\ka_l$ are called \emph{$\ga$-cumulants} and the $m_k$'s are called \emph{moments}.
This is meant to draw an analogy with the sequences of classical cumulants and moments of a probability measure.
The motivation for this terminology is explained by the results in Section \ref{Section_semifree}.
Roughly speaking, the map $\Tcm$ degenerates to the relation between cumulants and moments when $\ga\to 0$, and to the relation between free cumulants and moments when $\ga\to \infty$.

\begin{thm}[Law of Large Numbers for high temperature]\label{thm_small_th}
The sequence $\{\mu_N\}_{N \geq 1}$ is $\gamma$-LLN--appropriate if and only if it satisfies a LLN.
In case this occurs, the sequences $\{\ka_l\}_{l\geq 1}$ and $\{m_k\}_{k\geq 1}$ are related by
\begin{equation}
\label{eq_x28}
\{m_k\}_{k\ge 1} = \Tcm(\{\ka_l\}_{l\ge 1}).
\end{equation}
\end{thm}

\medskip

The proof of this theorem is given later in Section \ref{sec_proof_LLN} below.

Our next results describe in more detail the map $\Tcm$ from Definition \ref{def_R_map}.

\subsection{Combinatorial formula for the map $\Tcm$}\label{sec_Tcm}

From Definition \ref{def_R_map}, we are able to obtain the values of $m_k$ by doing  calculations with formal power series and isolating the constant term of the resulting expansion.
For example, for $k=1, 2, 3,4$, the resulting formulas are the following:
\begin{equation}\label{ex_c_to_m}
\begin{aligned}
m_1 &= \ka_1,\\
m_2 &= (\ga+1)\ka_2 + \ka_1^2,\\
m_3 &= (\ga+1)(\ga+2)\ka_3 + 3(\ga+1)\ka_2\ka_1 + \ka_1^3,\\
m_4&=(\ga+1)(\ga+2)(\ga+3)\ka_4+4(\ga+1)(\ga+2)\ka_3 \ka_1+(\ga+1)(2\ga+3)\ka_2^2+6(\ga+1)\ka_2 \ka_1^2+\ka_1^4.
\end{aligned}
\end{equation}

However, the defining formula \eqref{eq_moments_through_f_cumulants} is not explicit enough and becomes  complicated when $k$ is large.
Our next main theorem is a simpler combinatorial formula that expresses $m_k$ as a polynomial of the variables $\ka_1, \ka_2, \cdots, \ka_k$.
To state it, we need some terminology.

For any $k\in\Z_{\geq 1}$, denote $[k] := \{1, 2, \dots, k\}$.
A \emph{set partition} $\pi$ of $[k]$ is an (unordered) collection of  pairwise disjoint nonempty subsets of $[k]$ such that $[k] = B_1\cup\dots\cup B_m$.
The subsets $B_1, \dots, B_m$ are called the \emph{blocks} of the set partition $\pi$ and we use the notation $\pi = B_1\sqcup \dots\sqcup B_m$.
The cardinalities of the blocks are denoted $|B_1|, \dots, |B_m|$.
We denote the collection of all set partitions of $[k]$ by $\PP(k)$.
Given a set partition $\pi$, we denote by $\#(\pi)$ its number of blocks. For example, there are seven set partitions of $[4]$ with two blocks; they are:
$$
 \{1\}\sqcup \{2,3,4\},\quad \{1,3,4\}\sqcup \{2\},\quad \{1,2,4\}\sqcup\{3\},\quad \{1,2,3\}\sqcup\{4\},
$$
$$
 \{1,2\}\sqcup \{3,4\},\quad \{1,3\}\sqcup \{2,4\},\quad \{1,4\}\sqcup \{2,3\}.
$$
We also use the Pochhammer symbol notation:
\begin{equation*}
(x)_q := \begin{cases}
x(x+1)\cdots (x+q-1), &\text{ if }q\in\Z_{\ge 1},\\
1, &\text{ if }q=0. \end{cases}
\end{equation*}

\begin{df}\label{W_def}
For any $\pi\in\PP(k)$ and $\gamma\in\mathbb R$, define the quantity $W(\pi)$, that will be called the \emph{$\ga$-weight of $\pi$}, as follows\footnote{We omit the dependence on $\gamma$ from the notation $W(\pi)$.}.
Let $m = \#(\pi)$ and label the blocks of $\pi$ by $B_1, \cdots, B_m$ in such a way that the smallest element from $B_i$ is smaller than all elements from $B_j$, whenever $i<j$.
That is, if the blocks are $B_i = \{b^i_1 < \dots < b^i_{|B_i|}\}$, then $b^1_1 < b^2_1 < \cdots < b^m_1$.
For each $i\in\{1, \cdots, m\}$, define $p(i)$ as the number of indices $j\in\{1, \dots, |B_i| - 1\}$ such that $\{b^i_j + 1, \dots, b^i_{j+1} - 1\}\cap B_t \neq\emptyset$ for some block $B_t$ with $t < i$, and set $q(i) := |B_i|-1-p(i)$. Then define
\begin{equation}\label{W_formula}
W(\pi) := \prod_{i=1}^m\Bigl( p(i)!\cdot (\ga+p(i)+1)_{q(i)} \Bigr).
\end{equation}
\end{df}

For a set partition $\pi = B_1\sqcup \dots\sqcup B_m$ of $[k]$, we can think of the quantity $p(i)!\cdot (\ga+p(i)+1)_{q(i)}$ as a weight associated to the block $B_i$.
Therefore the $\ga$-weight $W(\pi)$ is the product of all weights of the blocks of $\pi$.
The weight of a block $B_i$ depends on the integer $p(i)$, whose computation can be visualized through a geometric procedure involving arc diagrams:
\begin{itemize}
 \item Draw each block $B_j=(b_1^j<b_2^j<\dots<b_{r}^j)$ as an arc with $r$ vertical legs at positions $b_a^j$, $a=1,\dots,r$ and with $r-1$ horizontal roofs joining adjacent legs at height $m+1-j$.
 \item $p(i)$ is the number of roofs in $B_i$, which intersect legs of other blocks. Note that each roof is counted only once, no matter how many legs it intersects.
\end{itemize}

Let us provide several examples. First, consider set partition $\{1, 2, 5, 7\}\sqcup \{3, 4, 6\}\in\PP(7)$ corresponding to the following arc diagram:
\begin{figure}[H]
\centering
\includegraphics[width=0.3\linewidth]{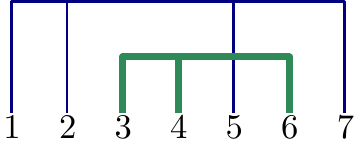}
\caption{Set partition $\{1, 2, 5, 7\}\sqcup \{3, 4, 6\}$}
\label{fig_1}
\end{figure}
\noindent The blocks are labeled $B_1 = \{1, 2, 5, 7\}$ and $B_2 = \{3, 4, 6\}$. We have $p(1) = 0$, $q(1) = 3$, $p(2) = 1$, $q(2) = 1$, and therefore
$$W(\pi) = 0!\cdot (\ga+1)_3\cdot 1!\cdot (\ga+2)_1 = (\ga+1)(\ga+2)^2(\ga+3).$$
For a different example, consider set partition $\{1, 4\}\sqcup\{2, 6\}\sqcup\{3, 5, 7\}\in\PP(7)$ corresponding to the following arc diagram:
\begin{figure}[H]
\centering
\includegraphics[width=0.3\linewidth]{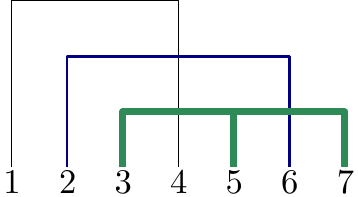}
\caption{Set partition $\{1, 4\}\sqcup\{2, 6\}\sqcup\{3, 5, 7\}$}
\label{fig_2}
\end{figure}
\noindent The blocks are labeled $B_1 = \{1, 4\}$, $B_2 = \{2, 6\}$, and $B_3 = \{3, 5, 7\}$.
We have $p(1) = 0$, $q(1) = 1$, $p(2) = 1$, $q(2) = 0$, $p(3) = 2$, $q(3) = 0$, and therefore
$$W(\pi) = 0!\cdot (\ga+1)_1\cdot 1!\cdot (\ga+2)_0\cdot 2!\cdot (\ga+3)_0 = 2(\ga+1).$$
For the final example, consider set partition  $\{1,3,4,5,6\}\sqcup\{2,7\}\in\PP(7)$ corresponding to the following arc diagram:
\begin{figure}[H]
\centering
\includegraphics[width=0.3\linewidth]{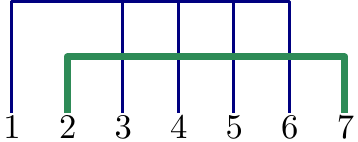}
\caption{Set partition $\{1, 3, 4, 5, 6\}\sqcup\{2, 7\}$}
\label{fig_3}
\end{figure}
\noindent The blocks are labeled $B_1=\{1,3,4,5,6\}$ and $B_2=\{2,7\}$. We have $p(1)=0$, $q(1)=4$, $p(2)=1$, $q(2)=0$, and therefore
$$
 W(\pi)=0! \cdot (\gamma+1)_{4} \cdot 1!\cdot (\gamma+1)_0= (\gamma+1)(\gamma+2)(\gamma+3)(\gamma+4).
$$
Let us also mention two useful properties which directly follow from the definition of $p(i)$:
\begin{itemize}
\item $p(1)=0$.
\item If $|B_i|=1$, then $p(i)=q(i)=0$.
\end{itemize}

\medskip

We have introduced all notations and can now state the main theorem of this section:

\begin{thm}[$\ga$--cumulants to moments formula]\label{theorem_cumuls_moms}
Let $\{m_k\}_{k\ge 1}$ and $\{\ka_l\}_{l\ge 1}$ be real sequences that are related by $\{m_k\}_{k\ge 1} = \Tcm(\{\ka_l\}_{l\ge 1})$.
Let $k\in\Z_{\ge 1}$ be arbitrary. Then
\begin{equation*}
m_k = \sum_{\pi\in\PP(k)}{W(\pi)\prod_{B\in\pi}{\ka_{|B|}}}.
\end{equation*}
\end{thm}

\medskip

The proof is presented in Section \ref{mom_cum_sec} below. In Section \ref{Section_semifree} we explain how in the limits $\gamma\to 0$ and $\gamma\to\infty$, Theorem \ref{theorem_cumuls_moms} turns into the expression of moments through classical cumulants and through free cumulants, respectively.

\subsection{ $\Tcm$ and its inverse $\Tmc$ through generating functions}
\label{Section_gen_functions}

The map $\Tcm: \{\ka_l\}_{l\ge 1}\mapsto\{m_k\}_{k\ge 1}$ is equivalent to relations of the form
\begin{equation}
\label{eq_mom_through_cum}
m_k = (\ga+1)_{k-1}\cdot \ka_k + \text{ certain polynomial in the variables }\ka_1,\dots, \ka_{k-1},
\end{equation}
Recursively using \eqref{eq_mom_through_cum}, each $\ka_l$ can be expressed as a polynomial in the variables $m_1, \cdots, m_l$.
In other words, the map $\Tcm$ has an inverse denoted by
$$
\Tmc := (\Tcm)^{-1}: \{m_k\}_{k\ge 1}\mapsto\{\ka_l\}_{l\ge 1}.
$$

For example, inverting the formulas in \eqref{ex_c_to_m} we get
\begin{equation}\label{ex_m_to_c}
\begin{aligned}
\ka_1 &= m_1,\\
\ka_2 &= \frac{1}{\ga+1}\bigl(m_2 - m_1^2\bigr),\\
\ka_3 &= \frac{1}{(\ga+1)_2}\bigl(m_3 - 3 m_2 m_1 + 2 m_1^3\bigr),\\
\ka_4 &= \frac{1}{(\ga+1)_{3}}\Biggl(m_4-4 m_3 m_1-\left[2+\frac{1}{\ga+1}\right] m_2^2 + \left[10+\frac{2}{\ga+1}\right] m_2 m_1^2 - \left[5+\frac{1}{\ga+1}\right] m_1^4\Biggr).
\end{aligned}
\end{equation}



One way to write the formulas connecting moments and cumulants in a compact form is through generating function:

\begin{thm}\label{thm:mom_cums2}
Let $\{m_k\}_{k\ge 1}$ and $\{\ka_l\}_{l\ge 1}$ be real sequences related by $\{\ka_l\}_{l\ge 1} = \Tmc(\{m_k\}_{k\ge 1})$.
Then
\begin{equation}\label{eq_cums_moments}
\exp\left( \sum_{l = 1}^{\infty}\frac{\ka_l y^l}{l} \right) =
[z^0]\left\{ \sum_{n=0}^{\infty}\frac{(yz)^n}{(\ga)_n} \cdot\exp\left( \ga\sum_{k=1}^{\infty}\frac{m_k}{k} z^{-k} \right)\right\}.
\end{equation}
Equivalently, \eqref{eq_cums_moments} can be rewritten as a combination of two identities involving an auxiliary sequence $\{c_n\}_{n\ge 0}$ through:
\begin{equation} \label{eq_cums_moments_2}
\begin{dcases}
 \exp\left(\sum_{l=1}^{\infty} \frac{\kappa_l}{l} z^l\right)=\sum_{n=0}^{\infty} \frac{c_n}{(\gamma)_n} z^n,\\
  \exp\left( \gamma \sum_{k=1}^{\infty} \frac{m_k z^k}{k}\right)=\sum_{n=0}^{\infty} c_n z^n.
\end{dcases}
\end{equation}

\end{thm}

As we explain in Section \ref{Section_semifree}, in the limit $\gamma\to 0$, the statement of Theorem \ref{thm:mom_cums2} turns into the well-known formula expressing the generating function of (classical) cumulants as a logarithm of the generating function of moments (equivalently, of the characteristic function of a random variable). On the other hand, in the limit $\gamma\to\infty$, Theorem \ref{thm:mom_cums2} can be converted into the identification of the free cumulants with Taylor series coefficients of the Voliculescu $R$--transform of a probability measure.

A close examination of \eqref{eq_cums_moments_2} reveals an unexpected connection to the $d$--cumulants for the (additive) finite free convolution. We recall that the latter is a deterministic binary operation on $d$--tuples of real numbers, which was shown in \cite[Theorem 1.2]{GM} to be the $\theta\to\infty$ limit of the operation $(\mathbf a,\mathbf b)\mapsto \mathbf a +_\theta \mathbf b$ for fixed $N=d$. Generating functions and certain combinatorial formulas for $d$--cumulants were developed in \cite{Marcus}, \cite{AP}. Comparing with \cite{AP}, we observe a match under the following change in notations, where in the left column we use notations from \cite{AP} and in the right column we use notations from our work:
\begin{equation} \label{eq_match_in_notations}
\begin{split}
d &\longleftrightarrow -\gamma,\\
m_n &\longleftrightarrow  m_n,\\
\kappa_n &\longleftrightarrow \gamma^{n-1}\kappa_n,\\
a_n &\longleftrightarrow (-1)^{n} c_n.
\end{split}
\end{equation}
Indeed, under \eqref{eq_match_in_notations} the first formula of \eqref{eq_cums_moments_2} becomes \cite[(3.1) or (3.3)]{AP} and the second formula of \eqref{eq_cums_moments_2} becomes \cite[(4.2)]{AP}. Note that the symbol $(x)_n$ has the meaning $x(x-1)\dots(x-n+1)$ in \cite{AP}, which is different from the convention that we use.

It is important to emphasize that in our work $\gamma>0$, while in \cite{Marcus}, \cite{AP}, $d$ is a positive integer. Hence, using \eqref{eq_match_in_notations} we see that there are no values of parameters under which finite free cumulants coincide with our $\gamma$--cumulants. Instead, one family of cumulants should be treated as an analytic continuation of another. There are two consequences of this correspondence. First, Theorem \ref{theorem_cumuls_moms} translates into a new combinatorial formula for finite free cumulants. Second, \cite[Theorem 4.2]{AP} explains how the generating function identity equivalent to \eqref{eq_cums_moments_2} leads to transition formulas (involving double sums over set partitions) between moments and finite free cumulants and vice versa. Hence, substituting \eqref{eq_match_in_notations} we can obtain similar formulas between moments and our $\gamma$--cumulants.

\subsection{Generalized Markov-Krein transform}\label{MK_transform}

There is a way to recast the formulas of Theorem \ref{thm:mom_cums2} connecting them to a remarkable non-linear transformation of measures discussed in \cite{FaFo} (see also \cite{MP} and \cite{K}). We take the numbers $c_n$ from \eqref{eq_cums_moments_2} and replace them with
$$
 \tilde c_n = \frac{n!}{(\gamma)_n} c_n.
$$
Further, suppose that $m_k$ are moments of a compactly supported probability measure $\nu$ and $\tilde c_n$ are moments of a compactly supported probability measure $\mu$:
$$
 m_k=\int_{\mathbb R} x^k \nu(dx),\qquad \tilde c_n = \int_{\mathbb R} x^n \mu(dx).
$$
Then the second identity of \eqref{eq_cums_moments_2} can be recast as
\begin{equation}
\label{eq_cums_moments_recast}
  \exp\left( -\gamma \int_{\mathbb R} \ln(z-x) \nu(dx) \right)= \int_{\mathbb R} \frac{1}{(z-t)^\gamma} \mu(dt), 
\end{equation}
where the equivalence of \eqref{eq_cums_moments_recast} with \eqref{eq_cums_moments_2} can be seen by assuming $z$ to be large and expanding the integrals into $1/z$ power series. It is proven in \cite{FaFo} that for any probability measure $\nu$ with $\int_{\mathbb R} \ln(1+|x|) \nu(dx) <\infty$ (in particular, compact support is not necessary), there exists another probability measure $\mu$, such that the identity \eqref{eq_cums_moments_recast} holds. For $\gamma=1$ the correspondence \eqref{eq_cums_moments_recast} and its relatives were popularized in the context of asymptotic problems by Kerov (see \cite{K}, \cite[Chapter VI]{Kerov_book}) under the name Markov--Krein transform; its origins go back to the studies of the solutions to the moment problems in the middle of the XX century. General $\gamma>0$ case was mentioned in \cite[Section 3.7 and 4.1]{K} and further discussed in \cite{FaFo}. In our setting the correspondence \eqref{eq_cums_moments_2} is useful because the first identity of \eqref{eq_cums_moments_2} is recast in terms of the measure $\mu$ as
$$
\sum_{l=1}^{\infty} \frac{\kappa_l}{l} z^l= \ln\left(\int_{\mathbb R} \exp(tz) \mu(dt)\right).
$$
Therefore, up to multiplication by $(l-1)!$, the $\gamma$--cumulants $\kappa_l$ of the measure $\nu$ are classical cumulants of the measure $\mu$ (we recall the definition of the classical cumulants in Section \ref{Section_limit_to_0}).

In a sense, the correspondence $\nu\leftrightarrow \mu$ of \eqref{eq_cums_moments_recast} reduces $\gamma$--cumulants (and all operations based on them) to classical cumulants. However, a difficulty in efficiently using this point of view is that the correspondence is highly non-linear and its properties are mostly unknown. For instance, describing all measures $\mu$, which can appear in the right-hand side of \eqref{eq_cums_moments_recast} is an open question.

\section{Applications}

\label{Section_applications}

In this section we list several applications of the general theorems from Section \ref{Section_main_results}.

\subsection{Law of Large Numbers for Gaussian $\beta$ ensembles}
\label{Section_GbE}

For each $N\ge 1$, let $\mu_{N,\,\th}$ be the $N$--particle \emph{Gaussian $\beta$ ensemble} with parameter $\beta = 2\th$ --- this is a probability distribution on $N$-tuples of real numbers $a_1\le a_2\le\cdots\le a_N$ with density proportional to
\begin{equation}
\label{eq_x29}
  \prod_{1\le i < j\le N}{\!\!\!\!(a_j - a_i)^{2\theta}}\, \prod_{i=1}^N{e^{-a_i^2/2}}.
\end{equation}
The eigenvalue distributions of the celebrated Gaussian Orthogonal/Unitary/Symplectic ensembles of random matrices are given by \eqref{eq_x29} at $\theta=\tfrac{1}{2}/1/2$, respectively.

To state the result of this subsection, we need a few definitions.
Denote by $\mathscr{M}(k)$ the collection of all \emph{perfect matchings} of $[k]$, that is, the collection of set partitions of $[k]$ where each block has size $2$.
$\mathscr{M}(k)$ is empty if $k$ is odd, and if $k = 2m$ is even, then $\mathscr{M}(k)$ has cardinality $(2m-1)!! = (2m-1)(2m-3)\cdots 3\cdot 1$.
Any perfect matching $\pi = \{B_1, \cdots, B_m\}$ of $[2m]$ is also a set partition of $[2m]$, so we can draw its arc diagram, as described in Section \ref{sec_Tcm}.
Denote by roof$(\pi)$ the number of roofs that do not intersect some leg.
Roof$(\pi)$ is an integer between $1$ and $m$, and roof$(\pi) = m$ if and only if the perfect matching $\pi$ is non-crossing, see Figure \ref{fig_roofs} for an illustration.

\begin{figure}[H]
\centering
\includegraphics[width=0.98\linewidth]{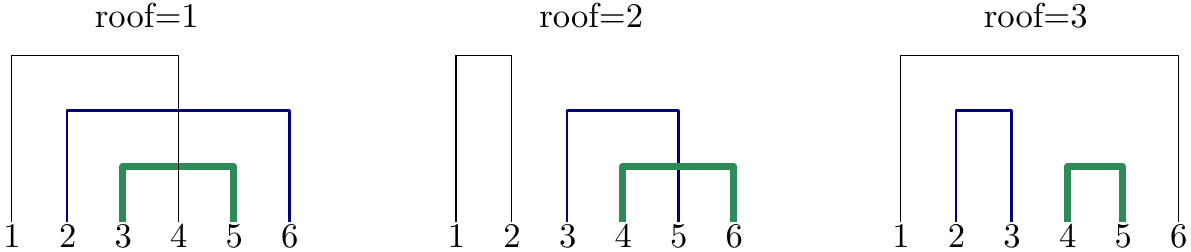}
\caption{Perfect matchings of $[6]$ with three possible values of roof$(\pi)$}
\label{fig_roofs}
\end{figure}

Finally, consider the empirical measures
$$
\rho_{N\!,\,\th} := \frac{1}{N}\sum_{i=1}^N{\delta_{a_i}},\qquad (a_1, \cdots, a_N) \text{ is $\mu_{N,\th}$--distributed}.
$$

\begin{thm}\label{thm_Gauss}
As $N\to\infty$, $\th\to 0$, $\th N\to\ga$, the (random) measures $\rho_{N\!,\,\th}$ converge weakly, in probability to a deterministic probability measure $\mu_\ga$ which is uniquely determined by its moments:
\begin{equation}\label{moms_Gaussian}
\int_{-\infty}^{\infty}{x^k \mu_\ga(dx)} = \sum_{\pi\in\mathscr{M}(k)}{\!\!\!(\ga+1)^{\mathrm{roof}(\pi)}},
\end{equation}
which is set to be $0$ for odd $k$.
\end{thm}
\begin{rem}
In our Theorem \ref{thm_Gauss}, the limiting measure $\mu_\ga$ is an analogue of Wigner's semicircle law from free probability theory and of the Gaussian distribution from classical probability, because the only nonzero $\ga$-cumulant is the second one. Similarly to these measures, $\mu_\ga$ is also present in a Central Limit Theorem with respect to the operation of $\ga$-convolution discussed in the next subsection, see \cite[Section 5.3]{MP}. In fact, $\mu_\gamma$ degenerates into these measures at special values of $\gamma$. Indeed, when $\ga = 0$ and $k = 2m$, the right-hand side of \eqref{moms_Gaussian} is equal to $|\mathscr{M}(2m)| = (2m-1)!!$, which coincides with the $(2m)$-th moment of the standard normal distribution.
When $\ga\to\infty$ and $k = 2m$, the right-hand side of \eqref{moms_Gaussian} (when divided by $\ga^m$) becomes the number of non-crossing perfect matchings of $[2m]$, which is the $m$-th Catalan number $C_m = \frac{1}{m+1}{2m \choose m}$. This is the $(2m)$-th moment of the standard Wigner's semicircle law.
\end{rem}

\begin{rem}
 While identification of $\mu_\ga$ through \eqref{moms_Gaussian} was not stated explicitly in the literature before, the LLN itself, i.e.\ existence of the limiting measures $\mu_\ga$ is known at least from \cite{ABG}.  \cite{BGP} provides other (more complicated) formulas for the moments of $\mu_\ga$, and in \cite{DS} the measure $\mu_\ga$ is identified with the mean spectral measure of a certain random Jacobi matrix.
\end{rem}
\begin{remark} \label{Remark_aHerm}
The measures $\mu_\ga$ were also previously studied by other authors without knowing about their connections to Gaussian $\beta$ ensembles.
Askey and Wimp \cite{AW} studied $\mu_\ga$ as an orthogonality measure for the \emph{associated Hermite polynomials} (\cite{D} obtained a formula for the moments that is equivalent to ours, and \cite{SZ}, \cite{BDEG} contain generalizations: see (4.7) in the former paper and Proposition 4.19 in the latter). Interestingly, the same polynomials also play a role in studying $\beta\to\infty$ limits of Gaussian $\beta$ Ensembles, see \cite[Section 4.3]{Gorin_Klept}. From another direction, Kerov \cite{K} studied $\mu_\ga$ in connection to the Markov-Krein transform and noticed that these measures interpolate between Gaussian and semicircle laws.
\end{remark}
\begin{proof}[Proof of Theorem \ref{thm_Gauss}]
By Lemma \ref{hermite_bgf} below, the Bessel generating function of $\mu_{N\!,\, \th}$ is
\begin{equation*}
G_N(x_1, \cdots, x_N; \th) = \exp\left( \frac{x_1^2 + \cdots + x_N^2}{2} \right).
\end{equation*}
It follows that $\{\mu_{N\!,\,\th}\}$ is $\ga$-LLN--appropriate with $\{\kappa_l\}_{l\ge 1}$ given by
\begin{equation}\label{gamma_cums_1}
\ka_l = \begin{cases}
1, & \text{if }l = 2,\\
0,& \text{otherwise}.
\end{cases}
\end{equation}

The corresponding sequence $\{m_k\}_{k\ge 1} = \Tcm(\{\ka_l\}_{l\ge 1})$ is given by the formula in Theorem \ref{theorem_cumuls_moms}.
Because the only nonzero $\ga$--cumulant $\ka_l$ is the one with $l = 2$, the summation for $m_k$ reduces from all set partitions of $[k]$ to all perfect matchings of $[k]$. In particular, $m_k = 0$ if $k$ is odd.
In the case that $k$ is even, say $k = 2m$, consider any perfect matching $\pi = \{B_1, \cdots, B_m\}$ of $[2m]$; each block $B_i$ has cardinality $2$, so $p(i)$ is $1$ if the roof of the arc $B_i$ intersects some leg in the arc-diagram of $\pi$, and otherwise $p(i)$ is $0$. As a result, the weight $W(\pi)$ in \eqref{W_formula} is equal to $(\ga+1)^{\text{roof}(\pi)}$.

Then Theorem \ref{thm_small_th} shows that the sequence $\{\mu_{N, \th}\}$ satisfies a LLN, and this proves the desired convergence in the statement of the theorem, see Remark \ref{rem_uniqueness}.
It remains to show that the right-hand sides of \eqref{moms_Gaussian} are the moments of a \emph{unique} probability measure. For this, we check the Carleman's condition: the moments problem for a sequence $\{\alpha_k\}_{k\ge 1}$ determines a unique probability measure if
\begin{equation} \label{eq_x30}
\sum_{m=1}^{\infty}{(\alpha_{2m})^{-1/(2m)}} = +\infty.
\end{equation}
Indeed, elementary bounds show
\begin{gather*}
\alpha_{2m} = \!\!\!\sum_{\pi\in\mathscr{M}(2m)}{\!\!\!(\ga+1)^{\text{roof}(\pi)}}
\le (\ga+1)^m \cdot |\mathscr{M}(2m)| = (\ga+1)^m \cdot (2m-1)!! \le (\ga+1)^m \cdot (2m)^m
\end{gather*}
and so $(\alpha_{2m})^{-1/(2m)} \ge \text{const}\cdot\frac{1}{\sqrt{m}}$, thus proving \eqref{eq_x30}.
\end{proof}

\begin{lemma}\label{hermite_bgf}
The Bessel generating function of the $N$--particle Gaussian $\beta$ ensemble with parameter $\beta = 2\th$
is
\begin{equation}
\label{eq_Hermite_BGF}
G_{\theta}(x_1, \cdots, x_N; \mu_{N\!,\,\th}) = \exp\left( \frac{x_1^2 + \cdots + x_N^2}{2} \right).
\end{equation}
\end{lemma}

Identity \eqref{eq_Hermite_BGF} is a folklore and we do not claim any novelty. One can prove it by taking an appropriate limit of the Cauchy identity for Jack polynomials. Alternatively, the computation of the Bessel generating function for the Laguerre $\beta$ Ensemble is equivalent to \cite[Eqn. (3.1)]{Ner}, and \eqref{eq_Hermite_BGF} is then obtained by a limit transition. As yet another approach, Gaussian $\beta$ Ensembles can be identified with measures from \cite[Thm. 1.13 and Sec. 6]{AN} with a single non-zero parameter $\gamma_2$, and for those, the Bessel generating function is computed in \cite{AN}, see also \cite[Cor. 2.6]{OV} for the $\theta=1$ case. For completeness, let us sketch the first argument.

\begin{proof}[Sketch of the proof of Lemma \ref{hermite_bgf}]
We rely on the theory of symmetric functions, as presented in \cite[Ch. I and VI]{M}, see also \cite{Stanley_Jack}.
We denote by $\Lambda$ the real algebra of symmetric functions in infinitely many variables $x_1,x_2,\dots$, which is generated by the (algebraically independent) power sums $p_k=\sum_i (x_i)^k$, $k=1,2,\dots$. A distinguished linear basis of $\Lambda$ is given by the Jack symmetric functions $P_{\la}(\cdot; \theta)$, where $\la = (\la_1 \ge \la_2\ge \cdots \ge 0)$ ranges over the set of all partitions and $P_{\la}(\cdot; \theta)$ is homogeneous of degree $|\lambda|=\sum_i \lambda_i$.
The Jack symmetric functions exhibit the following Cauchy(-Littlewood) summation identity:
\begin{equation} \label{eq_Cauchy}
\sum_{\la}\frac{ P_{\la}(x_1,x_2,\dots; \theta) P_{\la}(y_1,y_2,\dots; \theta) }{\langle P_\lambda,P_\lambda\rangle} = \exp\left( \theta \sum_{k=1}^{\infty}{\frac{p_k(x_1,x_2,\dots)p_k(y_1,y_2,\dots)}{k}} \right),
\end{equation}
where the sum is over all partitions $\la$ and $\langle P_\lambda,P_\lambda\rangle$ are explicitly known $\theta$--dependent normalization constants. \eqref{eq_Cauchy} is a formal identity of power series. We can turn it into a numeric identity by applying a \emph{specialization}: an algebra homomorphism from $\Lambda$ to the real numbers. We need one particular homomorphism, which is known as the Plancherel specialization in the literature:
$$
 \tau_s: \Lambda \to \mathbb R,\qquad \tau_k(p_k)=\begin{cases} s, & k=1,\\ 0, & k \ge 1. \end{cases}
$$
Applying $\tau_s$ with $s>0$ to the $x$--variables and setting $y_{N+1}=y_{N+2}=\dots=0$, we transform\eqref{eq_Cauchy} into
\begin{equation} \label{eq_Cauchy_2}
\sum_{\la=(\lambda_1\ge \lambda_2\ge \dots \ge \lambda_N\ge 0)}\frac{ P_{\la}(\tau_s; \theta) P_{\la}(1^N; \theta) }{\langle P_\lambda,P_\lambda\rangle} \cdot \frac{P_{\la}(y_1,\dots, y_N; \theta)}{P_{\la}(1^N; \theta)} = \exp\left( \theta s \sum_{i=1}^N y_i \right),
\end{equation}
where $1^N$ stays for the $N$ variables equal to $1$ and the summation got restricted to the partitions with (at most) $N$ parts, because  $P_{\la}(y_1,\dots,y_N; \theta)$ vanishes for others. Let us now introduce a probability measure on $N$--tuples of integers $\lambda=(\lambda_1\ge \dots\lambda_N\ge 0)$ through the formula:
\begin{equation}\label{eq_Jack_measure}{\rm Prob}(\lambda_1,\dots,\lambda_N)=\mathcal{J}_{1^N; \tau_s}(\la) := e^{-\theta sN}\frac{P_{\la}(\tau_s; \theta)P_{\la}(1^N; \theta)}{\langle P_\lambda,P_\lambda\rangle}.
\end{equation}
$\mathcal{J}_{1^N; \tau_s}(\la)$ is an instance of the Jack measure on partitions and it is discussed in details in \cite[Sec. 2.2]{GSh}.
With this definition, we can rewrite \eqref{eq_Cauchy_2} as
\begin{equation}\label{eq_Cauchy_3}
\mathbb{E}_{\mathcal{J}_{1^N; \tau_s}} \left[ \frac{P_{\la}(y_1, \cdots, y_N; \theta)}{P_{\la}(1^N; \theta)} \right] =
 \exp\left(\theta s  \sum_{i=1}^N(y_i -1) \right).
\end{equation}
Next, we want to set
\begin{equation}\label{params_spec}
s = \eps^{-1}\theta^{-1},\qquad \la_i = \eps^{-1} + \eps^{-1/2} a_{N+1-i},
\quad y_i = \exp(\eps^{1/2}x_i),\quad i=1,2,\cdots, N,
\end{equation}
for $\eps > 0$ and then send $\eps\to 0$. \cite[Proposition 2.10 with $t = 1$]{GSh} proves that the measures \eqref{eq_Jack_measure} converge weakly to the Gaussian $\beta$ ensemble of \eqref{eq_x29}:
$$\lim_{\eps\to 0+}\mathcal{J}_{1^N;\, \tau_{\eps^{-1}\theta^{-1}}}=\mu_{N,\, \theta}.$$
In this way \eqref{eq_Hermite_BGF} is obtained as $\eps\to 0$ limit of \eqref{eq_Cauchy_3}. For the right-hand side of \eqref{eq_Cauchy_3} we have:
\begin{equation}\label{limit_1}
\exp\left(\theta s \sum_{i=1}^N(y_i- 1)\right) =  \exp\left(\eps^{-1/2}\sum_{i=1}^N{x_i}+ \sum_{i=1}^N{\frac{x_i^2}{2}} + O(\eps^{1/2})\right).
\end{equation}
For the left hand side of \eqref{eq_Cauchy_2}, we use the identity (valid for any $c\in\mathbb Z$) which is a discrete version of \eqref{const_seq}:
$$P_{(c+\mu_1, \dots,\, c+\mu_N)}(y_1, \cdots, y_N) = (y_1\cdots y_N)^c\cdot P_{(\mu_1, \dots,\, \mu_N)}(y_1, \cdots, y_N)$$
and limit relation between the symmetric Jack polynomials and multivariate Bessel functions:
\begin{equation}\label{limit_to_bessel}
\lim_{\eta\to 0^+}{ \frac{P_{(\eta^{-1} a_N, \dots, \eta^{-1} a_1)}(\exp(\eta x_1), \cdots, \exp(\eta x_N); \theta)}{P_{(\eta^{-1} a_N, \dots, \eta^{-1} a_1)}(1^N; \theta)} } = B_{(a_1, \cdots, a_N)}(x_1, \cdots, x_N; \theta).
\end{equation}
The last identity can be found in \cite[Sec. 4]{Ok_Olsh_shifted_Jack} or \cite[Thm. 7.5]{C}; in this limit transition the combinatorial formula for Jack polynomials (expressing the expansion of the polynomial into monomials as a sum over semi-standard Young tableaux or Gelfand-Tsetlin patterns) turns into \eqref{eq_Bessel_combinatorial} of Definition \ref{Definition_Bessel_function}.

Thus, dividing both sides of \eqref{eq_Cauchy_2} by  $\exp\left(\eps^{-1/2}\sum_{i=1}^N{x_i}\right)$ and sending $\eps\to 0$ (we omit a standard tail bound justifying the validity of the limit transition under the expectation sign), we get
\begin{equation*}
\mathbb{E}_{\mu_{N, \theta}}\!\left[ B_{(a_1, \cdots, a_N)}(x_1, \cdots, x_N; \theta) \right] = \exp\left(\sum_{i=1}^N\frac{x_i^2}{2}\right). \qedhere
\end{equation*}
\end{proof}

\subsection{$\gamma$--convolution}

\begin{proof}[Proof of Theorem \ref{Theorem_gamma_convolution}]
 Let $G_N^{\mathbf a}$ be the BGF of the distribution\footnote{In this section the word ``distribution'' is used in probabilistic meaning, as in ``distribution of a random variable'', rather than in functional-analytic meaning, where a distribution is a synonym of a generalized function.} of $\mathbf a(N)$  and let $G_N^{\mathbf b}$ be the BGF of the distribution of $\mathbf b(N)$. Since Definition \ref{Def_mom_convergence} is the same as Definition \ref{Definition_LLN_sat_ht}, Theorem \ref{thm_small_th} implies that the distributions of $\mathbf a(N)$ and $\mathbf b(N)$ are $\gamma$--LLN appropriate. Let us denote the corresponding $\gamma$--cumulants (right-hand sides in (a) of Definition \ref{Definition_LLN_appr_ht}) through $\kappa_l^{\mathbf a}$ and $\kappa_l^{\mathbf b}$, respectively.

Further, let $G_N^{\mathbf a+_\theta \mathbf b}$ be the BGF of $\mathbf a(N)+_\theta \mathbf b(N)$. By Definition \ref{Def_theta_addition} and the independence of the distributions of $\mathbf a(N)$ and $\mathbf b(N)$, we have
 $$
  G_N^{\mathbf a+_\theta \mathbf b}(x_1,\dots,x_N;\, \theta)=G_N^{\mathbf a}(x_1,\dots,x_N;\, \theta) \cdot G_N^{\mathbf b}(x_1,\dots,x_N;\, \theta).
 $$
 Hence, partial derivatives of $\ln(G_N^{\mathbf a+_\theta \mathbf b})$ are sums of those of $\ln(G_N^{\mathbf a})$ and those of $\ln(G_N^{\mathbf a})$. Therefore, the sequence of distributions of $\mathbf a(N)+_\theta \mathbf b(N)$ is $\gamma$--LLN appropriate with $\gamma$--cumulants given by the sums
 $\kappa_l^{\mathbf a}+\kappa_l^{\mathbf b}$, $l=1,2,\dots$. Applying Theorem \ref{thm_small_th} again, we conclude that $\mathbf a(N)+_\theta \mathbf b(N)$ converges in the sense of moments, which concludes the proof of the theorem. Observe that  the $l$th $\ga$-cumulant of the limit measure is $\kappa_l^{\mathbf a}+\kappa_l^{\mathbf b}$, for all $l=1,2,\cdots$, which then verifies formula \eqref{eq_convolution_cumulants}.
\end{proof}

\medskip

One remark is in order. We never prove (or claim) that the $N\to\infty$ limit of empirical distributions of $\mathbf a(N) +_\theta \mathbf b(N)$ is given by a probability measure; we only show that the moments converge to some limiting values. Of course, if we knew that $\mathbf a(N)+_\theta \mathbf b(N)$ is a bona fide random $N$--tuple of integers (which is widely believed to be true, see Conjecture \ref{conj_pos}), then we could say that the deterministic limit of random empirical distributions is necessarily given by a probability measure. In this case, the binary operation $\boxplus_\gamma$ would turn into an operation on probability measures.

\begin{definition}\label{semifree_conv}
Let $\ga>0$ be fixed. Let $\mu$ be a probability measure with finite moments $\{m_k^{\mu}\}_{k\ge 1}$.
If $\{\ka_l^{\mu}\}_{l\ge 1} := \Tcm(\{m_k^{\mu}\}_{k\ge 1})$, then the quantities $\ka_l^{\mu}$ are said to be the \emph{$\ga$-cumulants of $\mu$}. Next, let $\nu$ be another probability measure with all $\ga$-cumulants $\{\ka_l^{\nu}\}_{l\ge 1}$. If there exists a unique probability measure, to be denoted $\mu \boxplus_\gamma \nu$, such that all of its $\ga$-cumulants are finite and given by
$$\ka_l^{\mu \boxplus_\gamma \nu} = \ka_l^{\mu} + \ka_l^{\nu},\text{ for all }l\ge 1,$$
then $\mu \boxplus_\gamma \nu$ is said to be the \emph{$\ga$-convolution} of $\mu$ and $\nu$.
\end{definition}
We note that if Conjecture \ref{conj_pos} is true, then the words ``If there exists a probability measure'' can be removed from the above definition, because such a measure would always exists. However, one should still require the uniqueness because not every probability measure is uniquely determined by its $\gamma$--cumulants or moments; one needs to impose a slow growth condition on the moments to guarantee the uniqueness.

\medskip

This definition agrees with the definition embedded in Theorem \ref{Theorem_gamma_convolution} in the following sense.
If $\tau$ is a probability measure which is determined by its moments $\{m_k^{\tau}\}_{k\ge 1}$, then $\tau = \mu \boxplus_\gamma \nu$ (in the sense of Definition \ref{semifree_conv}) if and only if $\{m_k^{\tau}\}_{k\ge 1} = \{m_k^{\mu}\}_{k\ge 1} \boxplus_\gamma \{m_k^{\nu}\}_{k\ge 1}$ (in the sense of Theorem \ref{Theorem_gamma_convolution}).

\subsection{Examples of $\ga$--convolutions}

\begin{exam}
Let $\{m_k\}_{k\ge 1}$ be the sequence of moments of a probability measure $\mu$ on $\R$.
Also consider the sequence $\{a^k\}_{k\ge 1}$ of powers of a real number $a\in\R$; evidently this is the sequence of moments of the Dirac delta mass at point $a$.
Let $\{\wt m_k\}_{k\ge 1}$ be the sequence of moments of the conventional convolution $\delta_a * \mu$, in other words, if we set $m_0 := 1$ then
\begin{equation}\label{conv_delta}
\wt m_k = \sum_{i=0}^k{{k \choose i}a^i m_{k-i}},\qquad k=1, 2, \cdots.
\end{equation}
For any $\ga>0$, we claim
\begin{equation*}
\{\wt m_k\}_{k\ge 1} = \{a^k\}_{k\ge 1} \boxplus_\gamma \{m_k\}_{k\ge 1}.
\end{equation*}

This is equivalent to $\delta_a * \mu = \delta_a \boxplus_{\gamma} \mu$, i.e. it would mean that $\gamma$--convolution with a Dirac delta mass at $a$ is identified with shift by $a$. Indeed, we can verify this claim by using Theorem \ref{Theorem_gamma_convolution}.
Note that for the constant sequences $\{ \mathbf{a}(N) = (a, \cdots, a) \}_{N\ge 1}$, we have $\lim_{N\to\infty}\mathbf{a}(N) \stackrel{m}{=} \{ a^k \}_{k\ge 1}$. Moreover, for any $\mathbf{b}(N) = (b_1(N)\le \cdots \le b_N(N))$, we have that $\mathbf{a}(N) +_{\theta} \mathbf{b}(N)$ is the deterministic $N$-tuple $(b_1(N)+a\le \cdots \le b_N(N)+a)$, as mentioned right after Definition \ref{Def_theta_addition}. Hence if $\lim_{N\to\infty}\mathbf{b}(N) \stackrel{m}{=} \{ m_k \}_{k\ge 1}$, then $\lim_{N\to\infty}{\mathbf{a}(N) +_{\theta} \mathbf{b}(N)} \stackrel{m}{=} \{ \wt{m}_k \}_{k\ge 1}$. Then the claim follows from Theorem \ref{Theorem_gamma_convolution}.

\end{exam}

\begin{exam}\label{gauss_exam}
Let $\sigma^2 > 0$ be any positive number and consider the sequence of $\ga$--cumulants:
\begin{equation}\label{gamma_cums_ex}
\ka_{l}^{\sigma^2} := \begin{cases}
\sigma^2, & \text{if }l = 2,\\
0,& \text{otherwise}.
\end{cases}
\end{equation}
Denote the corresponding sequence of moments as $\{m_k^{\sigma^2}\}_{k\ge 1} := \Tcm(\{\ka_l^{\sigma^2}\}_{l\ge 1})$.
Observe that $\{m_k^{\sigma^2}\}_{k\ge 1}$ is the sequence of moments of a rescaled version of the distribution of Theorem \ref{thm_Gauss}, that we denote $\mu_{\ga, \sigma^2}$. From \eqref{gamma_cums_ex} and the definition of $\ga$--convolution, it follows that for any $\sigma_1^2, \sigma_2^2 >0$ we have $\mu_{\ga,\, \sigma_1^2 + \sigma_2^1} = \mu_{\ga, \sigma_1^2} \boxplus_\gamma \mu_{\ga, \sigma_2^2}$, or equivalently
$$
\{m_k^{\sigma_1^2 + \sigma_2^2}\}_{k\ge 1} = \{m_k^{\sigma_1^2}\}_{k\ge 1} \boxplus_\gamma \{m_k^{\sigma_2^2}\}_{k\ge 1}.
$$
\end{exam}

\begin{exam}\label{laguerre_exam}
Let $\la > 0$ be arbitrary and consider the constant sequence of $\ga$--cumulants: $\ka_{l}^{\la} := \la$, for all $l = 1, 2, \cdots$.
Denote the corresponding sequence of moments as $\{ m_k^\la \}_{k\ge 1} := \Tcm(\{ \ka_l^\la \}_{l\ge 1})$.
It is known that $\{m_k^\la\}_{k\ge 1}$ is the sequence of moments of a probability measure $\nu_\ga^\la$.
It follows that for any $\la_1, \la_2>0$, we have $\nu_\ga^{\la_1+\la_2} = \nu_\ga^{\la_1}\boxplus_\gamma\nu_\ga^{\la_2}$, or equivalently
$$
\{m_k^{\la_1 + \la_2}\}_{k\ge 1} = \{m_k^{\la_1}\}_{k\ge 1} \boxplus_\gamma \{m_k^{\la_2}\}_{k\ge 1}.
$$
The measure $\nu_\ga^\la$ was studied in \cite{ABMV,TT_Laguerre}, where it was shown to be the limit of the empirical measures of beta Laguerre ensembles in the limit $N\to\infty$, $\beta N\to 2\gamma$, $N/M\to\la$. The density of $\nu_\ga^\la$ can be obtained from \cite[Lemma 2.1]{TT_Laguerre}; note that in that paper our parameters $\ga, \la$ are denoted by $c, \alpha$, respectively. We also refer to \cite[Section 5.4 and Figure 5]{MP} for additional details and plots of the densities.
Since all the $\gamma$--cumulants of $\nu_\ga^{\la}$ are equal to each other,  this measure is similar to the Poisson and the Marchenko-Pastur distributions whose cumulants, respectively, free cumulants, are all the same.
\end{exam}

\subsection{Law of Large Numbers for ergodic measures}\label{sec_ergodic}

We start this section by providing some context in the complex case $\theta=1$ (or $\beta=2$). The infinite-dimensional unitary group $U(\infty)$ is defined as the union of the groups of $N\times N$ unitary matrices, $\bigcup_{N=1}^{\infty} U(N)$, where we embed $U(N)$ into $U(N+1)$ as the subgroup of operators fixing the $(N+1)$st basis vector. Each element of $U(\infty)$ is an infinite matrix, such that for some $N=1,2,\dots$, its top $N\times N$ corner is unitary and outside this corner we have $1$s on the diagonal and $0$s everywhere else. Consider the space $\mathcal H$ of infinite complex Hermitian matrices with rows and columns parameterized by positive integers $i$ and $j$.  $U(\infty)$ acts on $\mathcal H$ by conjugations and one can ask about random matrices in $\mathcal H$ whose laws are invariant under such action. Their probability distributions form a simplex and much of the work on conjugation-invariant matrices comes down to study extreme points of this simplex --- ergodic conjugation-invariant random matrices in $\mathcal H$. In \cite{Pickrell,OV} these matrices were completely classified: they depend on a sequence of real parameters $\{\alpha_i\}_{i=1}^{\infty}$ with $\sum_{i=1}^{\infty} (\alpha_i)^2\le \infty$ and two reals $\delta_1\in\mathbb R$, $\delta_2\ge 0$ and are given by an infinite sum:
\begin{equation}
\label{eq_OV_expansion}
 \delta_1 \mathcal I + \sqrt{\delta_2}\, \frac{X+X^*}{2} + \sum_{i=1}^\infty \alpha_i \left(\tfrac{1}{2}V_i V_i^*-\mathcal I\right),
\end{equation}
where $\mathcal I$ is the identical matrix (with $1$s on the diagonal and $0$s everywhere else), $X$ is a matrix with i.i.d.\ Gaussian $\mathcal N(0,1)+\ii \mathcal N(0,1)$ elements, $V_i$ is an infinite (column-)vector with i.i.d.\  Gaussian $\mathcal N(0,1)+\ii \mathcal N(0,1)$ components and all the involved matrices are independent. Note that if the only non-zero parameter is $\delta_2$, then \eqref{eq_OV_expansion} gives the Gaussian Unitary Ensemble (a particular case of Wigner matrices). If the only non-zero parameters are $\alpha_1=\alpha_2=\dots=\alpha_K=1$ and $\delta_1=K$, then \eqref{eq_OV_expansion} gives the Laguerre Unitary Ensemble (a particular case of Wishart or sample-covariance matrices).

As was first mentioned in \cite[Remark 8.3]{OV} and recently studied in details in \cite{AN}, the problem of classification of conjugation-invariant infinite complex Hermitian matrices has a general $\theta$--version related to the $\theta$-corners processes of Definition \ref{def_betacorner}.
Roughly speaking, while there are no infinite self-adjoint matrices in the general $\th$--version, one can make sense of the \emph{distribution of eigenvalues} of the top-left principal submatrices (corners) of the infinite self-adjoint matrices.
One of the problems addressed in \cite{AN} (see Theorem 1.13 there) is the classification of ergodic random matrices at general values of $\th$. Since there are no bona fide matrices, this problem actually asks for distributions of $N$-tuples, for $N=1, 2, \cdots$, satisfying certain coherence relations --- the distributions should be regarded as the eigenvalue distributions of the $N\times N$ corners of an ergodic matrix.
It turns out that the set of parameters remains the same as in the $\theta=1$ case. The law of the top-left $1\times 1$ corner $\eta$ of an ergodic matrix at general values of $\theta$ has characteristic function
\begin{equation}
\label{eq_Fourier_theta}
 \E e^{\ii t \eta}= \mathcal F_{\theta; \{\alpha_i\}, \delta_1, \delta_2} (\ii t),\quad \text{ where }\quad \mathcal F_{\theta; \{\alpha_i\}, \delta_1, \delta_2} (z):= \exp\Bigl(\delta_1 z+\tfrac{\delta_2}{2\theta} z^2\Bigr) \cdot \prod_{i=1}^{\infty} \frac{\exp(-\alpha_i z)}{\left(1-\tfrac{\alpha_i}{\theta} z\right)^\theta}.
\end{equation}
More generally, there is a formula that uniquely determines the eigenvalue distribution of the corners, namely if $\eta_1\le \eta_2\le\dots\le \eta_N$ are the random eigenvalues of the $N\times N$ corner, then their Bessel generating function is explicit:
\begin{equation}
\label{eq_BGF_ergodic}
\E\!\left[ B_{(\eta_1,\dots,\eta_N)}(x_1,\dots,x_N;\, \theta) \right] = \prod_{j=1}^N \mathcal F_{\theta; \{\alpha_i\}, \delta_1, \delta_2} (x_j).
\end{equation}

We will take \eqref{eq_BGF_ergodic} as our definition of the distributions on $N$-tuples $(\eta_1\le\dots\le\eta_N)$; these distributions are the ergodic measures of \cite{OV, AN}. For them, we prove the following Law of Large Numbers in the regime $N\to\infty,\,\theta N\to\gamma$.

\begin{thm} \label{Theorem_ergodic}
 Suppose that $\theta$ and $\{\alpha_i\}_{i=1}^{\infty}$, $\delta_1$, $\delta_2$ vary with $N$ in such a way that $N\to\infty$, $\th\to 0$, $\theta N\to \gamma$ and
 \begin{equation}
 \label{eq_ergodic_convergence}
   \ln \left( \mathcal F_{\theta; \{\alpha_i\}, \delta_1, \delta_2} (z) \right) \longrightarrow F(z)= \sum_{l=1}^{\infty} \frac{\kappa_l}{l} z^l,
 \end{equation}
 uniformly over a complex neighborhood of $0$. Then the eigenvalues $(\eta_1,\dots,\eta_N)$ of the $N\times N$ corners of the corresponding general $\theta$ ergodic random matrix converge in the sense of moments (as in Definitions \ref{Def_mom_convergence} or \ref{Definition_LLN_sat_ht}) to a probability distribution with $\gamma$--cumulants $\kappa_l$, i.e.\ its moments are found by the expression of Theorem \ref{theorem_cumuls_moms}.
\end{thm}
\begin{remark}\label{Remark_Gauss_Laguerre}
Choosing $\delta_2 = \theta$, so that $\mathcal F_{\theta; \{\alpha_i\}, \delta_1, \delta_2} (z)=\exp( z^2/2 )$, we recover the LLN for the Gaussian $\beta$--ensembles as $N\to\infty$, $\beta N\to 2\gamma$, as in Section \ref{Section_GbE}.
\noindent Choosing $\alpha_1=\alpha_2=\dots=\alpha_M=\theta$, $\delta_1=M\theta$ with $M=\lfloor \la N/\ga \rfloor$, so that $\mathcal F_{\theta; \{\alpha_i\}, \delta_1, \delta_2} (z)=(1-z)^{-\theta \lfloor \la N / \ga \rfloor}\to (1-z)^{-\la}$, we recover the LLN for the Laguerre $\beta$--ensembles as $N\to\infty$, $\beta N\to 2\gamma$ and $M/N\to \la/\ga$.
The limiting probability measure is $\nu_\ga^\la$, as described in Example \ref{laguerre_exam}.
\end{remark}
\begin{remark} The formula \eqref{eq_Fourier_theta} has a multiplicative structure: a product of $\mathcal F_{\theta; \{\alpha_i\}, \delta_1, \delta_2} (z)$ functions is again a function of the same type. This property leads to the limits in Theorem \ref{Theorem_ergodic} being infinitely-divisible with respect to $\gamma$--convolution $\boxplus_\gamma$.
This is in agreement with Examples \ref{gauss_exam} and \ref{laguerre_exam}, which show that the measures $\mu_\ga^{\sigma^2}$ and $\nu_\ga^\la$ are $\gamma$--infinitely-divisible.
\end{remark}
\begin{proof}[Proof of Theorem \ref{Theorem_ergodic}]
 Combining \eqref{eq_BGF_ergodic} with \eqref{eq_ergodic_convergence}, we conclude that the BGF of $(\eta_1,\dots,\eta_N)$ is $\gamma$-LLN appropriate in the sense of Definition \ref{Definition_LLN_appr_ht}. Hence, by Theorem \ref{thm_small_th}, $(\eta_1,\dots,\eta_N)$ converge in the sense of moments and the asymptotic moments are recovered from the $\gamma$--cumulants $\kappa_l$ by using the map $\Tcm$.
\end{proof}

\subsection{Limit of projections}

We again start from the complex case $\theta=1$. This time we fix $N=1,2,\dots$ and a deterministic $N$--tuple of reals $a_1\le \dots\le a_N$. Let $A_N$ be a uniformly random $N\times N$ complex Hermitian matrix with eigenvalues $a_1,\dots,a_N$ and let $A_m$ be the $m\times m$ top-left submatrix of $A_N$. We now fix $\tau>1$, set $m=\lfloor N/\tau\rfloor$ and send $N\to\infty$. If we assume that the empirical measures of eigenvalues of $A_N$, $\frac{1}{N}\sum_{i=1}^N \delta_{a_i}$, converge to a limiting probability measure $\mu$, then the (random) empirical measures of eigenvalues of $A_m$ converge to a (deterministic) measure $\mu^{\boxplus \tau}$. For integer $\tau$ this measure is the same as the free convolution of $\tau$ copies of $\mu$, hence, $\mu^{\boxplus \tau}$ can be called a fractional convolution power, see \cite{STJ} for a recent study and references.

We now present an analogue of the operation $\mu\mapsto \mu^{\boxplus \tau}$ in our $\theta\to 0$ asymptotic framework.

\begin{thm} \label{Theorem_projections} Fix real numbers $\gamma>0$ and $\tau>1$. Suppose that for each $N=1,2,\dots,$ we are given an $N$--tuple of reals $a_1\le\dots\le a_N$, and let $\{y_i^k\}_{1\le i \le k \le N}$ be the $\theta$--corners process with top row $a_1,\dots,a_N$, as in Definition \ref{def_betacorner}. (In particular, this means $y_1^N=a_1$,\dots, $y_N^N=a_N$.) Define the empirical measures
$$
 \rho_N=\frac{1}{N}\sum_{i=1}^N \delta_{a_i},\qquad \rho_{N}^\tau=\frac{1}{\lfloor N/\tau\rfloor} \sum_{i=1}^{\lfloor N/\tau\rfloor }\delta_{y^{\lfloor N/\tau\rfloor}_i},
$$
and suppose that all measures $\rho_N$ are supported inside a segment $[-C,C]$ and as $N\to\infty$, $\rho_N$ weakly converge to a probability measure $\mu$ (supported inside the same segment). Then as $N\to\infty$, $\theta\to 0$ with $\theta N\to\gamma$, the (random) measures $\rho_N^\tau$ converge weakly, in probability to a deterministic measure $\mu^{\tau,\gamma}$. If $\{m_k\}_{k\ge 1}$ are the moments of $\mu$ and $\{m_k^{\tau}\}_{k\ge 1}$ are the moments of $\mu^{\tau,\gamma}$, then
\begin{equation}\label{same_cums}
 T^{\gamma}_{m\to \kappa}\bigl(\{m_k\}_{k\ge 1}\bigr)=T^{\gamma/\tau}_{m\to \kappa}\bigl(\{m_k^{\tau}\}_{k\ge 1}\bigr).
\end{equation}
In other words, $\gamma$--cumulants of $\mu$ coincide with $\tfrac{\gamma}{\tau}$--cumulants of $\mu^{\tau,\gamma}$.
\end{thm}
\begin{remark}
 The condition of support inside $[-C,C]$ is used to guarantee that all the involved measures are determined by their moments; it can be replaced by other uniqueness conditions for the moments problem.
\end{remark}
\begin{proof}[Proof of Theorem \ref{Theorem_projections}]
 Convergence $\rho_N\to \mu$ and the condition on the support of $\rho_N$ imply that the moments of $\rho_N$ converge to those of $\mu$. Hence, the sequence of delta-measures (unit masses) on $N$--tuples $(a_1,\dots,a_N)$ satisfies LLN in the sense of Definition \ref{Definition_LLN_sat_ht}. Thus, Theorem \ref{thm_small_th} yields that it is $\gamma$--LLN appropriate, i.e., its BGF
 $$
  G_{N;\theta}(x_1,\dots,x_N)= B_{(a_1,\dots,a_N)}(x_1,\dots, x_N;\, \theta)
 $$
 satisfies the conditions of Definition \ref{Definition_LLN_appr_ht}. Let $\tilde G_{N;\theta}$ denote the BGF of the $\lfloor N/\tau\rfloor$--tuple of reals $y^{\lfloor N/\tau\rfloor}_1\le \dots \le y^{\lfloor N/\tau\rfloor}_{\lfloor N/\tau\rfloor}$. Then Definition \ref{Definition_Bessel_function} implies that
 $$
  \tilde G_{N;\theta}(x_1,\dots,x_{\lfloor N/\tau\rfloor})=G_{N;\theta}(x_1,\dots,x_{\lfloor N/\tau\rfloor}, 0,0,\dots,0),
 $$
 where there are $N-\lfloor N/\tau\rfloor$ in the right-hand side. Hence, the partial derivatives of $\ln(\tilde G_{N;\theta})$ coincide with partial derivatives of $\ln(G_{N;\theta})$ and, therefore, the former is $\gamma$--LLN appropriate. It is important to emphasize at this point that we use the same $\theta$ for $G_{N;\theta}$ and $\tilde G_{N;\theta}$, however, the number of variables for the latter is $\lfloor N/\tau\rfloor$ rather than $N$. This leads to $\gamma$ being divided by $\tau$. It remains to use Theorem \ref{thm_small_th} yet again to conclude that the random measures $\rho_{N}^\tau$ converge in the sense of moments and consequently also weakly, in probability.
\end{proof}
 In general, we do not know any simple criteria on when a given sequence of numbers is a sequence of $\gamma$--cumulants corresponding to a probability measure. Yet Theorem \ref{Theorem_projections} leads to an interesting comparison between different $\gamma$'s.
\begin{corollary}
 Take a sequence of real numbers $\kappa_1,\kappa_2,\dots$ and suppose that for some $\gamma_0>0$, these numbers are $\gamma_0$--cumulants of some probability measure $\nu$, i.e., $\{\kappa_l\}_{l\ge 1}=T^{\gamma_0}_{m\to\kappa}\bigl(\{m_k\}_{k\ge 1}\bigr)$ with $m_k=\int_{\mathbb R} x^k \nu(dx)$. Then for each $0<\gamma<\gamma_0$ the same numbers are also $\gamma$--cumulants of some probability measure. In particular, sending $\gamma\to 0$, we also have that the sequence $0! \kappa_1, 1!\kappa_2, 2!\kappa_3,\dots$ gives conventional cumulants of some probability measure.
\end{corollary}

In fact, $0! \kappa_1, 1!\kappa_2, 2!\kappa_3,\dots$ are the conventional cumulants of the probability measure $\mu$ that is related to $\nu$ by means of the generalized Markov-Krein transform \eqref{eq_cums_moments_recast}.

\begin{proof}
 We apply Theorem \ref{Theorem_projections} with $\tau=\gamma_0/\gamma$. The theorem was proven only for compactly supported measures, but we can approximate any measure by compactly supported ones. Finally, the convergence of $\gamma$--cumulants to conventional cumulants as $\gamma\to 0$ is discussed in Section \ref{Section_semifree}.
\end{proof}

\begin{remark}
Theorem \ref{Theorem_projections}, or just equation \eqref{same_cums}, defines for $\tau>1$ the \emph{$(\tau,\ga)$--projection} map
$$
\Pi_{\tau,\ga} : \mu \mapsto \mu^{\tau,\ga},
$$
which maps the space of probability measures of compact support to itself.
It would be interesting to study the possibility of an extension of $(\tau, \ga)$--projection map to all probability measures. In particular, the probability measures $\mu_\ga^{\sigma^2}$ and $\nu_\ga^\la$ from Examples \ref{gauss_exam} and \ref{laguerre_exam} should map to the measures of the same type:
$\Pi_{\tau, \ga}(\mu_\ga^{\sigma^2}) = \mu_{\ga/\tau}^{\sigma^2}$ and $\Pi_{\tau, \ga}(\nu_\ga^\la) = \nu_{\ga/\tau}^\la$.
\end{remark}

\section{Law of Large Numbers at high temperature}\label{sec_proof_LLN}

In this section, we prove Theorem \ref{thm_small_th}.
Recall that the real parameter $\ga>0$ is fixed, and we are interested in the limit regime  $N\to\infty$, $\th\to 0$, and $\th N\to\ga$.

Let us recall some terminology about partitions of numbers (rather than set partitions of Section \ref{Section_main_results}).
A partition $\la$ is a weakly decreasing sequence of nonnegative integers $\la = (\la_1\ge \la_2\ge \cdots\ge 0 )$, $\la_i\in\Z_{\ge 0}$, such that $\sum_{i=1}^{\infty}\lambda_i<\infty$. The latter sum is denoted $|\la|$ and is called the \emph{size} of the partition $\la$.
If $\la$ is a partition of size $k$, we write $\la\vdash k$.
The \emph{length} $\ell(\la)$ of $\la$ is defined as the number of strictly positive parts of $\la$.

The partitions are often identified with Young diagrams, in which $\lambda_i$ become the row lengths. We also need column lengths $\lambda'_1\ge \lambda'_2\ge\dots$ defined by $\lambda_{j}'=|\{i\ge 1 \mid \lambda_i\ge j\}|$. In particular, $\lambda'_1=\ell(\lambda)$.


\subsection{The asymptotic expansion of Dunkl operators}

If $F$ is a smooth symmetric function of the $N$ variables $x_1,\dots,x_N$, then its Taylor series expansion is also symmetric and we can write the $k$--th order approximation as
\begin{equation} \label{eq_symmetric_Taylor}
F(x_1,\dots,x_N)= \sum_{\lambda:\, |\lambda|\le k,\, \ell(\lambda) \le N} c_F^{\lambda}\cdot M_\lambda(\vx)+ O(\|x\|^{k+1}),
\end{equation}
where the sum is over partitions $\lambda$ of size at most $k$ and length at most $N$. Finally, $M_\lambda(\vx)$ is the monomial symmetric function:
$$
M_\lambda(\vx)=\sum_{\begin{smallmatrix}(d_1,\dots,d_N)\in\mathbb Z^N_{\ge 0}, \text{ such that}\\ \lambda\text{ is the rearrangement of } d_i \text{ in nonincreasing order}\end{smallmatrix}} x_1^{d_1} \cdot x_2^{d_2}\cdots x_N^{d_N}.
$$

\begin{thm} \label{theorem_operators_expansion}
Fix $k=1,2,\dots$ and a partition $\lambda$ with $|\lambda|=k$.  Let $F(x_1,\dots,x_N)$ be a symmetric function of $(x_1,\dots,x_N)\in\R^N$, which is $(k+1)$--times continuously differentiable in a neighborhood of $(0,\dots,0)$ and satisfies $F(0,\dots,0)=0$. Then we have:
\begin{multline}\label{eq_operator_small_th_expansion}
 N^{-\ell(\lambda)} \left[\prod_{i=1}^{\ell(\lambda)} \P_{\lambda_i}\right] \exp\bigl( F(x_1,\dots,x_N)\bigr)\Bigr|_{\setzeroes}= b^{\lambda}_{\lambda} \cdot c_F^{\lambda} + \sum_{\mu:\, |\mu|=k,\, \ell(\mu)>\ell(\lambda)} b^{\lambda}_{\mu} \cdot c_F^{\mu}\\ + L\Bigl(c_F^{(i)},\, 1\le i \le k-1\Bigr) +  R_1\Bigl(c^{\nu}_F, \, |\nu|< k\Bigr)  + N^{-1} R_2\Bigl(c^{\nu}_F, \, |\nu|\le k\Bigr),
\end{multline}
where $b^\lambda_\mu$ are coefficients, which are uniformly bounded in the regime $N\to\infty$, $\theta\to 0$, $\theta N\to \gamma$. In particular,
\begin{equation}
\label{eq_diagonal_asymptotics}
 \lim_{\begin{smallmatrix} N\to\infty,\, \theta\to 0\\ \theta N\to\gamma\end{smallmatrix}}\, b_{\lambda}^{\lambda}= \prod_{i=1}^{\ell(\lambda)}\lambda_i (1+\gamma)_{\lambda_i-1}.
\end{equation}
Further,
\begin{equation} \label{eq_one_row_part}
 L\Bigl(c_F^{(i)},\, 1\le i \le k-1\Bigr)=  \prod_{i=1}^{\ell(\lambda)} \left([z^0](\partial+\gamma d+*_g)^{\lambda_i-1} g(z)\right)-  k (1+\gamma)_{k-1} c^{(k)}_F {\mathbf 1}_{\ell(\lambda)=1},
\end{equation}
where $(m)$ is the one-row Young diagram of size $m$, the operators $\partial$, $d$, $*_g$ are the ones introduced in Definition \ref{df:ops}, and $g(z) := \sum_{n=1}^{\infty} n c_F^{(n)} z^{n-1}$.

Next,  $R_1\Bigl(c^{\nu}_F, \, |\nu|< k\Bigr)$ is a polynomial  in $c^{\nu}_F, \, |\nu|< k$, such that:
\begin{itemize}
 \item If we assign the degree $|\nu|$ to each $c^\nu_F$, then $R_1$ is homogeneous of degree $k$.
 \item The coefficients of the monomials in $R_1$ are uniformly bounded in the regime $N\to\infty$, $\theta\to 0$, $\theta N\to \gamma$.
 \item Each monomial in $R_1$ has at least one factor $c^\nu_F$ with $\ell(\nu)>1$.
\end{itemize}
Finally,  $R_2\Bigl(c^{\nu}_F, \, |\nu|\le k\Bigr)$ is a homogeneous polynomial in   $c^{\nu}_F, \, |\nu|\le  k$, of degree $k$ and with uniformly bounded  coefficients (in the same regime).
\end{thm}

Before proving Theorem \ref{theorem_operators_expansion}, let us use it to deliver the proof of Theorem \ref{thm_small_th}.

\begin{proof}[Proof of Theorem \ref{thm_small_th}]
First, take a LLN--appropriate sequence $\{\mu_N\}_N$ with associated sequence of real numbers $\{\ka_l\}_{l\ge 1}$.
Let $\{m_k\}_{k\ge 1}$ be the image of $\{\ka_l\}_{l\ge 1}$ under the map $\Tcm$, that is, each $m_k$ is the function of the $\ka_l$'s given by \eqref{eq_moments_through_f_cumulants}.
We aim to show that $\{\mu_N\}_N$ satisfies a LLN with associated sequence of real numbers $\{m_k\}_{k\ge 1}$.

Let us denote the BGF of $\mu_N$ by $G_{N; \theta}(x_1, \dots, x_N)$. Let $s = 1, 2, \dots$ and $k_1, \dots, k_s\in\Z_{\geq 1}$ be arbitrary.
By Proposition \ref{proposition_moments_through_operators} (or Proposition \ref{Proposition_BGF_dist} for distributions), we have
\begin{equation}
\E_{\mu_N} \!\!\left[\prod_{i=1}^s  p_{k_i}^N \right] =
N^{-s} \left( \prod_{i=1}^s{\P_{k_i}}\right)G_{N; \th}(x_1, \dots, x_N)\Bigr|_{\setzeroes}.
\end{equation}
Without loss of generality, we assume that $k_1\ge k_2\ge\dots\ge k_s$, so that the $k_i$'s form a partition.
Since $G_{N; \th}$ is holomorphic in a neighborhood of the origin and $G_{N; \th}(0,\dots,0)=1$, then there is a holomorphic function $F_{N; \th}(x_1, \dots, x_N)$ in a neighborhood of the origin such that $G_{N; \th} = \exp(F_{N; \th})$ and $F_{N; \th}(0,\dots,0)=0$. The functions $F_{N; \th}(x_1, \dots, x_N)$ are smooth and symmetric in the real variables $x_1, \dots, x_N$, so we can consider their Taylor expansions:
$$
F_{N; \th}(x_1,\dots,x_N)= \sum_{\la:\, |\la|\le k,\, \ell(\la) \le N} c_{F_{N; \th}}^{\lambda}\cdot M_\lambda(\vx)+ O(\|x\|^{k+1}).
$$
By LLN--appropriateness,
$$
\limtwo c^{(n)}_{F_{N; \th}} = \frac{\ka_n}{n},\qquad \limtwo c^{\mu}_{F_{N; \th}}=0,\quad \text{if }\ell(\mu)>1.
$$
Apply Theorem \ref{theorem_operators_expansion} to the function $F_{N; \th}$ and the partition $\la = (k_1 \ge \cdots \ge k_s)$.
Let us take the limit of each term in the resulting right-hand side of \eqref{eq_operator_small_th_expansion} in the limit regime $N\to\infty$, $\th\to 0$, $\th N\to\ga$:
\begin{itemize}
 \item In the first line, if $s>1$, then each term involves some $c^{\mu}_{F_{N; \th}}$ with $\ell(\mu)>1$, and therefore tends to $0$. Otherwise, if $s=1$, then there is a single asymptotically non-vanishing term, namely $b^{(k_1)}_{(k_1)} \cdot c^{(k_1)}_{F_{N; \th}}$, which converges to $(1+\gamma)_{k_1-1}\,\ka_{k_1}$.
  \item The polynomial $R_1$ converges to $0$, since each of its monomials involves some $c^{\mu}_{F_{N; \th}}$ with $\ell(\mu)>1$ and, therefore, vanishes asymptotically.
  \item The polynomial $\frac{1}{N} R_2$ converges to $0$ due to the $\frac{1}{N}$ prefactor.
  \item The polynomial $L$ converges to
  $$
  \prod_{i=1}^{s} \left([z^0](\partial+\gamma d+*_g)^{k_i-1} g(z)\right)-  (1+\gamma)_{k_1-1} \ka_{k_1} {\mathbf 1}_{s=1},
  $$
where $g(z) = \sum_{n=1}^{\infty} {\ka_n z^{n-1}}$, due to the fact that the power series $\sum_{n=1}^{\infty}nc_{F_{N; \th}}^{(n)}z^{n-1}$ converges coefficient-wise to $g(z)$.
\end{itemize}
Combining the terms coming from the above four items, we conclude that
$$
\lim_{N\to\infty} \E_{\mu_N} \!\!\left[\prod_{i=1}^s  p_{k_i}^N \right] = \prod_{i=1}^s \left([z^0](\partial+\gamma d+*_g)^{k_i-1} g(z)\right)\!.
$$
We have thus arrived at the Law of Large Numbers with $m_k$ given by \eqref{eq_moments_through_f_cumulants}.

\bigskip

In the opposite direction, take a sequence $\mu_N$ which satisfies the Law of Large Numbers with associated sequence $\{m_k\}_{k\ge 1}$.
Let $\{\ka_l\}_{l\ge 1}$ be the image of $\{m_k\}_{k\ge 1}$ under the map $\Tmc$.
Again let $G_{N; \th}(x_1, \dots, x_N)$ be the BGF of $\mu_N$. We show that $\mu_N$ is LLN--appropriate with corresponding sequence $\{\ka_l\}_{l\ge 1}$, that is, we are going to establish the conditions on partial derivatives of Definition \ref{Definition_LLN_appr_ht}. This will be done by induction on the total order of the derivative. For the inductive step, we assume that for all $s\le k-1$, the asymptotic behavior of all partial derivatives of order $s$ is already established, i.e.\ we assume that the limits
$$
\lim_{\begin{smallmatrix} N\to\infty,\, \theta \to 0\\ \theta N\to \gamma \end{smallmatrix}}\left.\frac{\pa}{\pa x_{i_1}}\cdots\frac{\pa}{\pa x_{i_s}}\ln{(G_{N; \th})}\right|_{\setzeroes},\quad i_1, \cdots, i_s\in\Z_{\ge 1},
$$
exist and are equal to zero unless $i_1 = \cdots = i_s$, in which case the limit is equal to $(s-1)!\cdot \kappa_s$.

Our task is to prove the two conditions of Definition \ref{Definition_LLN_appr_ht} for $\ell=k$ and for $r=k$.
Let $p(k)$ be the total number of partitions of $k$ and consider the $p(k)$ expressions \eqref{eq_operator_small_th_expansion} obtained by making $\lambda$ run over all partitions of $k$ and letting $F_N:=F_{N; \th}$ be determined through $G_{N; \th} = \exp(F_{N; \th})$.
We regard the left-hand sides and the coefficients $c^{\mu}_{F_{N; \th}}$ with $|\mu| < k$, as constants, while we regard the terms $c^{\la}_{F_{N; \th}}$ with $|\la| = k$, as variables; then we can treat these expressions as $p(k)$ linear equations for the $p(k)$ variables $c^{\la}_{F_{N; \th}}$ with $|\la|=k$.
The coefficients of these equations generally depend on $N$ and $\theta$, and moreover we know the $N\to\infty$, $\theta\to\infty$, $\theta N\to\infty$ asymptotic behavior of the left-hand sides of \eqref{eq_operator_small_th_expansion} as well as $L$ and $R_1$ in the right-hand side (by the inductive hypothesis).
The form of the first line of \eqref{eq_operator_small_th_expansion} implies that the matrix of coefficients of these equations becomes triangular as $N\to\infty$, $\th\to 0$, $\theta N\to \gamma$ in the lexicographic order $\leq$ on partitions of size $k$, viewed as vectors of column lengths $(\lambda'_1,\lambda'_2,\dots)$, because $\ell(\mu)>\ell(\la)$ implies $\mu>\la$. The diagonal elements have nonzero limits, because of \eqref{eq_diagonal_asymptotics}.

We can rewrite these linear equations in the matrix notation. Let $B^{N,\theta}$ be the $p(k)\times p(k)$ matrix with matrix elements
$$
B^{N,\theta}(\mu,\lambda)=b^{\lambda}_\mu.
$$
Further, let $\mathbf c^N$ denote the $p(k)$--dimensional column-vector with coordinates $c^{\la}_{F_{N; \th}}$, $|\lambda|=k$. Then the previous paragraph can be summarized as a matrix equation
\begin{equation}
\label{eq_x26}
 B^{N,\theta}\cdot \mathbf c^N = \mathbf r^N,
\end{equation}
where the vector $\mathbf c^N$ is unknown and the right-hand side $\mathbf r^N$ is known. The key property of \eqref{eq_x26} is that the entries of the inverse  matrix $(B^{N,\theta})^{-1}$ are bounded as $N\to\infty$, $\theta\to\infty$, $\theta N\to\gamma$; this follows from triangularity of $B^{N,\theta}$ and non-zero limits for its diagonal entries. Let $\mathbf c^{\infty}$ denote another $p(k)$--dimensional vector, in which the first coordinate (corresponding to the one-row partition $(k)$) is $\tfrac{\kappa_k}{k}$ (here $\kappa_k$ is found from \eqref{eq_x28}, in which the numbers $m_1,m_2,\dots$ are known us) and all other coordinates are zeros. The first part of the proof (where we showed that each LLN-appropriate sequence satisfies LLN) and the induction hypothesis imply that
\begin{equation}
\label{eq_x27}
 B^{N,\theta}\cdot \mathbf  c^\infty = \mathbf r^N +o(1),
\end{equation}
where $o(1)$ is a vanishing term as $N\to\infty$, $\theta\to\infty$, $\theta N\to\gamma$. Multiplying
\eqref{eq_x26} and \eqref{eq_x27} by $(B^{N,\theta})^{-1}$ and comparing the results, we conclude that
$$
 \lim_{\begin{smallmatrix}N\to\infty, \theta\to 0,\\ \theta N\to \gamma \end{smallmatrix}} \mathbf c^N=\mathbf c^\infty. \qedhere
$$
\end{proof}

\subsection{Proof of Theorem \ref{theorem_operators_expansion}}\label{sec_proof_expansion}

We start by reducing to the case of $F$ being a symmetric polynomial.

\begin{lemma} \label{Lemma_replace_by_polynomial}
Suppose that $F$ is a $(k+1)$--times continuously differentiable function in a neighborhood of $(0, \dots, 0)\in\C^N$, with Taylor expansion \eqref{eq_symmetric_Taylor}. Then for any $\lambda$ with $|\lambda|=k$, we have
$$
\left[\prod_{i=1}^{\ell(\lambda)} \P_{\lambda_i}\right] \exp\bigl( F(x_1,\dots,x_N)\bigr)\Bigr|_{\setzeroes}=  \left[\prod_{i=1}^{\ell(\lambda)} \P_{\lambda_i}\right] \exp\bigl( \tilde F(x_1,\dots,x_N)\bigr)\Bigr|_{\setzeroes},
$$
where
$$
\tilde F(x_1,\dots,x_N)= \sum_{\nu:\, |\nu|\le k} c_F^{\nu}\cdot M_\nu(\vx).
$$
\end{lemma}
\begin{proof}
 We have
 $$
  \exp\bigl( F(x_1,\dots,x_N)\bigr)=\exp\bigl( \tilde F(x_1,\dots,x_N)\bigr) +R(x_1,\dots,x_N),
 $$
 where $R$ is a $(k+1)$--times continuously differentiable function, satisfying $R=O(\|x\|^{k+1})$ as $(x_1,\dots,x_N)\to(0,\dots,0)$. It remains to show that after we apply $k$ operators of the form $\frac{\pa}{\pa x_i}$ or $\frac{1-s_{ij}}{x_i-x_j}$ to $R$, the resulting function $R^{(k)}$ is continuous and vanishes at $(0,\dots,0)$. For that we let $R^{(m)}$, $m=1,2,\dots,k$ be the result of application of $m$ such operators and prove by induction in $m$ that $R^{(m)}$ is $(k+1-m)$--times continuously differentiable and satisfies $R^{(m)}=O(\|x\|^{k+1-m})$. The induction step is proven by applying to $R^{(m)}$ the Taylor's theorem with remainder in the integral form.
\end{proof}

By virtue of Lemma \ref{Lemma_replace_by_polynomial}, we can (and will) assume for the remainder of this section that
$$F(x_1, \dots, x_N) = \sum_{\nu:\, |\nu|\le k}{c^{\nu}_F\cdot M_\la(\vx)}.$$
Next, consider any product of $k$ operators, each of which is either $\frac{\pa}{\pa x_i}$ for some $i$, or $\frac{1}{x_i-x_j}(1-s_{ij})$ for some $i$ and $j$. We apply these operators inductively to $\exp(F)$, using the following rules:
\begin{multline}
\label{eq_Leibnitz}
\frac{\pa}{\pa x_i} \bigl[ H(x_1,\dots,x_N) \cdot \exp (F(x_1,\dots,x_N)) \bigr]\\
= \left(\frac{\pa}{\pa x_i} H(x_1,\dots,x_N)+ H(x_1,\dots,x_N) \frac{\pa}{\pa x_i} F(x_1,\dots,x_N)\right) \cdot \exp (F(x_1,\dots,x_N)),
\end{multline}
\begin{multline}
\label{eq_Leibnitz_dif}
 \frac{1}{x_i-x_j}(1-s_{ij}) \bigl[ H(x_1,\dots,x_N) \cdot \exp (F(x_1,\dots,x_N)) \bigr]\\ = \left(\frac{1}{x_i-x_j}(1-s_{ij})  H(x_1,\dots,x_N) \right) \cdot  \exp(F(x_1,\dots,x_N)).
\end{multline}

Hence, taking into account that $F(0,\dots,0)=0$, the result of acting by such product on $\exp(F)$ and then setting all variables equal to $0$ is a finite linear combination of products of actions of $\frac{\pa}{\pa x_i}$ and $\frac{1}{x_i-x_j}(1-s_{ij})$ on the function $F$, and then picking up the constant term of the polynomial.
Since $F$ is a polynomial with coefficients $c_F^\la$ and the actions of $\frac{\pa}{\pa x_i}$ and $\frac{1}{x_i-x_j}(1-s_{ij})$ on monomials are clear, we conclude the following statement.

\begin{lemma} \label{Lemma_D_as_polynomial}
For any $k$ indices $1\le i_1,\dots,i_k\le N$, the expression
\begin{equation}\label{prod_Ds}
\left( \prod_{m=1}^k \D_{i_m} \right) \exp(F(x_1\dots,x_N)) \Bigr|_{\setzeroes}
\end{equation}
is a homogeneous polynomial of degree $k$ in $c^\la_F$ (if we regard each $c^\la_F$ as a degree $|\lambda|$ variable), whose coefficients are uniformly bounded as $N\to\infty$, $\theta\to 0$, $\theta N\to\gamma$.
\end{lemma}
\begin{proof}
By definition, each $\D_i$ is linear combination of $N$ terms, each of which is $\frac{\pa}{\pa x_i}$ or $\frac{1}{x_i-x_j}(1-s_{ij})$.
Observe that any of these two simple operators decreases the degree of a polynomial in the variables $x_1, \cdots, x_N$ by $1$.
Therefore, using the rules \eqref{eq_Leibnitz} and \eqref{eq_Leibnitz_dif}, the expression \eqref{prod_Ds} is a polynomial in the coefficients of the degree $k$ component of $\exp(\sum_{\la:\, |\la|\le k}{c^\la_F \cdot M_\la(\vx)})$.
Such polynomial is therefore in the variables $c^\la_F$ and is homogeneous of degree $k$, because of how we assigned the degrees to the $c^F_\la$'s.

In the formula \eqref{dunkl_ops} for $\D_i$, the term $\frac{\pa}{\pa x_i}$ comes with unit coefficient, and the remaining terms $\frac{1}{x_i-x_j}(1-s_{ij})$ come with a prefactor $\th$, which decays as $\gamma/N$ as $N\to\infty$. Hence, expanding $\prod_{m=1}^k \D_{i_m}$ as a linear combination of products of the operators $\frac{\pa}{\pa x_i}$ and $\frac{1}{x_i-x_j}(1-s_{ij})$, we see that the coefficients of the polynomial \eqref{prod_Ds} in the variables $c_\la^F$ are uniformly bounded in the regime of our interest.
\end{proof}

\begin{corollary} \label{Corollary_reduce_terms}
Take any partition $\la$ with $|\la|=k$.
As $N\to\infty$, $\theta\to 0$, $\theta N\to \gamma$, we have
\begin{multline}
N^{-\ell(\lambda)} \left[\prod_{i=1}^{\ell(\lambda)} \P_{\lambda_i}\right]\! \exp\bigl( F(x_1,\dots,x_N)\bigr)\Bigr|_{\setzeroes} =
\left[\prod_{i=1}^{\ell(\lambda)} (\D_i)^{\lambda_i}\right]\! \exp\bigl( F(x_1,\dots,x_N)\bigr)\Bigr|_{\setzeroes}\\+ N^{-1} R_3,
\end{multline}
where $R_3$ is a homogeneous polynomial of degree $k$ in the coefficients $c^{\nu}_F$ (if we regard each $c^\nu_F$ as a degree $|\nu|$ variable), and with uniformly bounded coefficients.
\end{corollary}
\begin{proof}
 Each $\P_{\lambda_i}$ is a sum of $N$ terms $(\D_j)^{\lambda_i}$, $j=1,\dots,N$. Hence,  $\left[\prod_{i=1}^{\ell(\lambda)} \P_{\lambda_i}\right]$ is a sum of $N^{\ell(\lambda)}$ terms, each of which is a finite (independent of $N$ and $\theta$) product of $(\D_j)^{\lambda_i}$. For all but $O(N^{\ell(\lambda) - 1})$ of these terms, the indices $j$ are all distinct. Hence, by symmetry of $F$, the result of the action of such product on $\exp(F)$ is the same as that of $\prod_{i=1}^{\ell(\lambda)} (\D_i)^{\lambda_i}$, after setting all variables $x_i$ equal to zero. Dividing by $N^{-\ell(\lambda)}$, we get the desired statement.
\end{proof}

For the rest of the section, we analyze $\left[\prod_{i=1}^{\ell(\lambda)} (\D_i)^{\lambda_i}\right] \exp\bigl( F(x_1,\dots,x_N)\bigr)\Bigr|_{\setzeroes}$. In view of Corollary \ref{Corollary_reduce_terms} we need to show that it has an expansion of the form of the right-hand side of \eqref{eq_operator_small_th_expansion}.

\begin{proposition} \label{Proposition_highest_derivatives}
Take any partition $\la$ with $|\la|=k$. We have
$$
\left[\prod_{i=1}^{\ell(\lambda)} (\D_i)^{\lambda_i}\right] \exp\bigl( F(x_1,\dots,x_N)\bigr)\Bigr|_{\setzeroes}= b^{\lambda}_{\lambda} \cdot c_F^{\lambda} + \sum_{\mu:\, |\mu|=k,\, \ell(\mu)>\ell(\lambda)} b^{\lambda}_{\mu} \cdot c_F^{\mu}+R+O\left(\frac{1}{N}\right),
$$
where the coefficients $b^\lambda_\mu$  are uniformly bounded in the regime $N\to\infty$, $\theta\to 0$, $\theta N\to \gamma$. In particular,
$$
\lim_{\begin{smallmatrix} N\to\infty,\, \theta\to 0\\ \theta N\to\gamma\end{smallmatrix}}\, b_{\lambda}^{\lambda}= \prod_{i=1}^{\ell(\lambda)} \lambda_i (1+\gamma)_{\lambda_i-1}.
$$
Moreover, $R$ is a homogeneous polynomial of degree $k$ in the coefficients $c_F^\nu$ with $|\nu|<k$, i.e., it does not involve the coefficients $c_F^\nu$ with $|\nu|=k$. Finally, $O\left(\frac{1}{N}\right)$ stands for a linear polynomial in the coefficients $c_F^\nu$ with $|\nu| = k$, whose coefficients are of the order $O\left(\frac{1}{N}\right)$, as $N\to\infty$, $\theta\to 0$, $\theta N\to \gamma$.
\end{proposition}
\begin{proof}
 By Lemma \ref{Lemma_D_as_polynomial},  $\left[\prod_{i=1}^{\ell(\lambda)} (\D_i)^{\lambda_i}\right] \exp\bigl( F(x_1,\dots,x_N)\bigr)\Bigr|_{\setzeroes}$ is a homogeneous polynomial of degree $k$ in the coefficients $c^{\nu}_F$, $|\nu|\le k$.
Hence, its linear component is of the form
$$
\sum_{\mu :\, |\mu|=k} {b^\lambda_\mu \cdot c_F^\mu}.
$$
 Therefore, two steps remain:
 \begin{enumerate}
 \item We need to show that $b^{\lambda}_\mu=O\left(\frac{1}{N}\right)$ unless $\ell(\mu)>\ell(\lambda)$ or $\mu=\lambda$.
 \item We need to find the limit of $b^{\lambda}_\lambda$ as $N\to\infty$, $\theta\to 0$, $\theta N\to \gamma$.
\end{enumerate}

We first claim that the part of $\left[\prod_{i=1}^{\ell(\lambda)} (\D_i)^{\lambda_i}\right] \exp\bigl( F(x_1,\dots,x_N)\bigr)\Bigr|_{\setzeroes}$ involving the coefficients $c^\mu_F$ with $|\mu|=k$ is given by
\begin{equation}
\label{eq_x1}
 \left[\prod_{i=2}^{\ell(\lambda)} (\D_i)^{\lambda_i}\right] \cdot \left[\D_1^{\lambda_1-1}\right] \frac{\pa}{\pa x_1} F(x_1,\dots,x_N) \Bigr|_{\setzeroes}.
\end{equation}
Indeed, the operators $\D_i$ commute, hence, we can apply $\D_1$ first. In the very first application of $\D_1$, the terms $\frac{1}{x_1-x_j}(1-s_{1j})$ can be omitted, since $(1-s_{1j})$ annihilates the symmetric function $\exp(F)$. Hence, the result of the first application of $\D_1$ is $\frac{\pa F}{\pa x_1}\cdot\exp(F)$. Using formula \eqref{eq_Leibnitz}, we see that all the next applications of partial derivatives $\frac{\pa}{\pa x_1}$ should never act on $\exp(F)$, as otherwise we are not getting the terms $c^\mu_F$ with $|\mu|=k$. Similarly, when we further apply  $\prod_{i=2}^{\ell(\lambda)} (\D_i)^{\lambda_i}$, we should not act on $\exp(F)$.
Hence, we can omit $\exp(F)$, as it does not contribute to the computation. Therefore we get \eqref{eq_x1}.

We analyze \eqref{eq_x1} by using the expansion $F(x_1, \dots, x_N) = \sum_{\mu:\, |\mu|\le k}{c^\mu_F \cdot M_\mu(\vx)}$ in monomials and looking at each monomial separately.
Note that each operator $\D_i$ lowers by $1$ the degree of the monomial on which it acts. Since we apply $\frac{\pa}{\pa x_1}$, then $k-1$ operators $\D_i$, and then plug in all variables equal to $0$, the only way to get a non-zero contribution is by acting on a monomial of degree $k$.
We conclude that the coefficient $b^{\lambda}_\mu$ is computed by
\begin{equation}
\label{eq_b_through_monomial}
b^\lambda_\mu = \left. \left[\prod_{i=2}^{\ell(\lambda)} (\D_i)^{\lambda_i}\right] \cdot \left[\D_1^{\lambda_1-1}\right] \frac{\pa}{\pa x_1} M_\mu(\vx) \right|_{\setzeroes} = \left[\prod_{i=2}^{\ell(\lambda)} (\D_i)^{\lambda_i}\right] \cdot \left[\D_1^{\lambda_1-1}\right] \frac{\pa}{\pa x_1} M_\mu(\vx).
\end{equation}
Each $\D_i$ is a sum of $N$ operators. Hence, the operator in \eqref{eq_b_through_monomial} can be represented as a sum of $N^{k-1}$ operators, each of which is a product of the factors $\frac{\pa}{\pa x_i}$ and $\frac{\theta}{x_i-x_j}(1-s_{ij})$.

\smallskip

{\bf Claim A.} Only the terms in which all indices $j$ are distinct and are all larger than $\ell(\lambda)$ contribute to the leading term of \eqref{eq_b_through_monomial}. All others combine together into a remainder of order $O\left(\frac{1}{N}\right)$.

\smallskip

For example, if $\la = (2, 1)$ (so that $\ell(\la) = 2$) and $\mu = (1,1,1)$, then \eqref{eq_b_through_monomial} contains terms of the following types:
\begin{gather*}
\text{(I)}\ \frac{\pa}{\pa x_2}\, \frac{\pa^2}{\pa x_1^2} M_{(1,1,1)}(\vx),\qquad \text{(II)}\ \frac{\theta}{x_2-x_j}(1-s_{2j})\, \frac{\pa^2}{\pa x_1^2} M_{(1,1,1)}(\vx),\quad j\ne 2,\\
\text{(III)}\ \frac{\pa}{\pa x_2}\, \frac{\theta}{x_1-x_k}(1-s_{1k})\, \frac{\pa}{\pa x_1} M_{(1,1,1)}(\vx),\quad k\ne 1,\\
\text{(IV)}\ \frac{\theta}{x_2-x_j}(1-s_{2j})\, \frac{\theta}{x_1-x_k}(1-s_{1k})\, \frac{\pa}{\pa x_1} M_{(1,1,1)}(\vx),\quad j\ne 2,\ k\ne 1.
\end{gather*}
Then Claim A states that the term (II) with $j=1$, the term (III) with $k=2$, and the terms (IV) with $j=k$ or $j=1$ or $k=2$, all combined give a contribution which is smaller, by a factor of $N$, than the contribution of all other terms of these four types, i.e. those with $j, k > 2$ and $j\ne k$. 

\smallskip

Claim A is proven by a simple counting argument. Indeed, the terms with distinct indices are the generic ones: the number of terms where two indices coincide is smaller, by a factor of $N$, than the number of similar terms with distinct indices.

\smallskip

Next, using Claim A, let us take a look at the first application of $\D_2$ after we computed $\D_1^{\la-1} \frac{\pa}{\pa x_1} M_\mu(\vx)$. We could either apply $\frac{\pa}{\pa x_2}$ or we can apply $\frac{\theta}{x_2-x_j}(1-s_{2j})$. But due to symmetry in $x_2$ and $x_j$ the result of the application of the latter operator vanishes. Hence, we have to use $\frac{\pa}{\pa x_2}$. Similarly, in the first application of $\D_3$ we need to use $\frac{\pa}{\pa x_3}$, etc. We conclude that

\begin{equation}
\label{eq_b_through_monomial_2}
b^\lambda_\mu=   \left[\D_{\ell(\lambda)}^{\lambda_{\ell(\lambda)}-1} \frac{\pa}{\pa x_{\ell(\la)}} \right]\cdots \left[\D_2^{\lambda_2-1} \frac{\pa}{\pa x_2}\right] \cdot \left[\D_1^{\lambda_1-1} \frac{\pa}{\pa x_1}\right] \, M_\mu(\vx)+ O\left(\frac{1}{N}\right).
\end{equation}
We analyze the last expression in three steps.

\medskip

{\bf Step 1.} Let us show that if $\ell(\mu)<\ell(\lambda)$, then \eqref{eq_b_through_monomial_2} is $O\left(\frac{1}{N}\right)$. Indeed, if $\ell(\mu)<\ell(\lambda)$, then each monomial in $M_\mu(\vx)$ is missing one of the variables $x_1,\dots,x_{\ell(\lambda)}$. Say, it does not have $x_m$. Then, using the above Claim A, we see that when we apply $\frac{\pa}{\pa x_m}$ in \eqref{eq_b_through_monomial_2}, the expression has no dependence on $x_m$ and, hence, the derivative vanishes.

\medskip

{\bf Step 2.} If $\ell(\mu)>\ell(\lambda)$, then $b^{\lambda}_\mu$ are bounded as $N\to\infty$, $\th\to 0$, $\theta N\to\ga$, by Lemma \ref{Lemma_D_as_polynomial} and we do not need to prove anything else about them.

\medskip

{\bf Step 3.} It remains to study the case $\ell(\mu)=\ell(\lambda)=\ell$. Let us expand $M_\mu(\vx)$ in monomials. If a monomial is missing one of the variables $x_1,\dots,x_{\ell}$, then by the argument of Step 1, it does not contribute to $b^\lambda_\mu$. Hence, since $\ell(\mu)=\ell(\lambda)$, it remains to study the monomials which involve $x_1,x_2,\dots, x_{\ell}$ and no other variables.

Note that for $1\le i \le \ell<j$, we have, using the degree-lowering operators \eqref{eq_lowering_operator}
\begin{multline}\label{eq_x2}
\frac{\theta}{x_i-x_j} (1-s_{ij})  [ x_1^{n_1}\cdots x_\ell^{n_\ell}] = \theta \left[\prod_{a\ne i} x_a^{n_a}\right] \frac{x_i^{n_i} -x_j^{n_i} }{x_i-x_j}\\
= \th \left[\prod_{a\ne i} x_a^{n_a}\right] \left(x_i^{n_{i}-1}+ x_i^{n_i-2} x_j+\dots + x_j^{n_i-1}\right)=\theta d_i  [ x_1^{n_1}\cdots x_\ell^{n_\ell}] + x_j\cdot P,
\end{multline}
where $P$ is a polynomial of degree $n_1+\dots+n_\ell-2$.

Using the above Claim A, one sees that if a factor $x_j$, $j>\ell$ appears in a monomial, then this factor cannot be annihilated by applying any operator $\frac{\pa}{\pa x_i}$, $i\le \ell$, or  any operator $\frac{\theta}{x_i-x_{j'}} (1-s_{ij'})$, $i\le \ell$, $j\ne j'$, unless this application makes the entire monomial vanish. Hence, the only way to get a non-zero contribution is by using the $d_i$ term, but not the $x_j\cdot P$ term in \eqref{eq_x2}. Thus, up to $O\left(\frac{1}{N}\right)$ error, the desired $b^{\lambda}_\mu$ can be alternatively computed as:
\begin{equation*}
b^\lambda_\mu=  \left[ \left(\partial_\ell + \th(N-1) d_\ell\right)^{\lambda_{\ell}-1} \partial_{\ell}\right]\cdots \left[(\partial_1 + \th(N-1) d_1)^{\lambda_1-1} \partial_1\right] \, M_\mu(x_1,\dots,x_\ell)+ O\left(\frac{1}{N}\right).
\end{equation*}
(Above we denoted $\frac{\partial}{\partial x_i}$ by $\partial_i$ for all $i$.) The last operator lowers the degree of $x_1$ by $\lambda_1$, lowers the degree of $x_2$ by $\lambda_2$,\dots, lowers the degree of $x_\ell$  by $\lambda_\ell$. Since $\lambda_1+\dots+\lambda_\ell=\mu_1+\dots+\mu_\ell$, the only way to get a non-zero contribution after these lowerings is by having $\lambda=\mu$. Therefore, $b^{\la}_{\mu} = O(N^{-1})$ if $\ell(\mu) = \ell(\lambda)$ and $\mu\neq\lambda$.

Finally, in the case $\mu = \lambda$, we have
\begin{multline*}
\lim_{\begin{smallmatrix} N\to\infty,\, \theta\to 0\\ \theta N\to\gamma\end{smallmatrix}}{b_\la^\la}
= \lim_{\begin{smallmatrix} N\to\infty,\, \th\to 0\\ \th N\to\gamma\end{smallmatrix}}
\left[ \left(\pa_\ell + (N-1)\th d_\ell\right)^{\la_{\ell}-1} \pa_{\ell}\right]\cdots \left[(\pa_1 + (N-1)\th d_1)^{\la_1-1} \pa_1\right] \, x_1^{\la_1} x_2^{\la_2}\cdots x_\ell^{\la_\ell}\\
= \left[ \left(\pa_\ell + \ga d_l\right)^{\la_{\ell}-1} \pa_{\ell}\right]\cdots \left[(\pa_1 + \ga d_1)^{\la_1-1} \pa_1\right] \, x_1^{\la_1} x_2^{\la_2}\cdots x_\ell^{\la_\ell}\\
= \prod_{i=1}^{\ell} \la_i(\la_i-1+\ga)(\la_i-2+\ga)\cdots(1+\ga). \qedhere
\end{multline*}
\end{proof}

Proposition \ref{Proposition_highest_derivatives} gives the linear part, that is, the first line in \eqref{eq_operator_small_th_expansion}. The next step is to identify $L(\cdot)$ in the second line of \eqref{eq_operator_small_th_expansion}.

\begin{proposition} \label{Proposition_LLN_leading_term}
Take any partition $\lambda$ with $|\lambda|=k$. We have
\begin{equation}\label{eq_x3}
\left[\prod_{i=1}^{\ell(\lambda)} (\D_i)^{\lambda_i}\right] \exp\bigl( F(x_1,\dots,x_N)\bigr)\Bigr|_{\setzeroes}= \prod_{i=1}^{\ell(\lambda)} \left([z^0](\partial+\gamma d+*_g)^{\lambda_i-1} g(z)\right)+R+O\left(\frac{1}{N}\right),
\end{equation}
where $g(z)=\sum_{n=1}^{\infty} n c_F^{(n)} z^{n-1}$, and $\pa$, $d$, $*_g$ are the operators from Definition \ref{def_R_map}.
Moreover, $R$ is a homogeneous polynomial of degree $k$ in $c_F^\nu$ with $|\nu|\le k$, such that each monomial in it involves at least one $\nu$ with $\ell(\nu) > 1$. Finally, $O\left(\frac{1}{N}\right)$ is a (homogeneous of degree $k$) polynomial in  $c_F^\nu$ with $|\nu|\le k$, whose coefficients are $O\left(\frac{1}{N}\right)$ as $N\to\infty$, $\theta\to 0$, $\theta N\to \gamma$.
\end{proposition}
\begin{proof}
We only need to figure out the monomials which involve $c_F^{(n)}$, $n = 1, \dots, k$, and no other coefficients, so we are only interested in the following part of the left-hand side of \eqref{eq_x3} with
  \begin{equation}
  \label{eq_x4}
    \left[\prod_{i=1}^{\ell(\lambda)} (\D_i)^{\lambda_i}\right] \exp\left( \sum_{n=1}^k c_F^{(n)} M_{(n)}(\vec{x})\right)\Biggr|_{\setzeroes}=
    \left[\prod_{i=1}^{\ell(\lambda)} (\D_i)^{\lambda_i}\right]  \prod_{t=1}^N \exp\left( \sum_{n=1}^k c_F^{(n)} (x_t)^n \right)\Biggr|_{\setzeroes}.
  \end{equation}
Next, we recall that each $\D_i$ is a sum of $N$ operators, so that the operator in \eqref{eq_x4} is a sum of $N^{k}$ operators, each of which is a product of the factors $\frac{\pa}{\pa x_i}$ and $\frac{\theta}{x_i-x_j}(1-s_{ij})$. As in Claim A in the proof of Proposition \ref{Proposition_highest_derivatives}, we can and will assume without loss of generality that all indices $j$ are distinct and larger than $\ell(\lambda)$ --- we only accumulate $O\left(\frac{1}{N}\right)$ error by making such assumption.

 There are two consequences of this. First, like in \eqref{eq_b_through_monomial_2}, at this point for each $i$ the very first application of $\D_i$ can be replaced by $\frac{\pa}{\pa x_i}$, since the operators $\frac{\theta}{x_i-x_j}(1-s_{ij})$ act by $0$ due to symmetry in $i$ and $j$. Second, the operators no longer interact with each other in any way and the expression factorizes.
This reasoning is very similar to that in the proof of Proposition \ref{Proposition_highest_derivatives}, so we do not dwell on the details.
As a result, up to $O\left(\frac{1}{N}\right)$ error, \eqref{eq_x4} is equal to
\begin{equation}\label{eq_x5}
\prod_{i=1}^{\ell(\lambda)}  \left( (\D_i)^{\lambda_i-1}\, \frac{\pa}{\pa x_i} \!\left.\left[\prod_{t=1}^N \exp\left( \sum_{n=1}^k c_F^{(n)} (x_t)^n \right)\right]\right|_{\setzeroes}\right).
\end{equation}
 It remains to study the factor in \eqref{eq_x5} corresponding to a single $i$; without loss of generality, let us consider the case $i=1$. We would like to understand
\begin{multline}\label{eq_x6}
(\D_1)^{l-1} \,\frac{\pa}{\pa x_1}\!\left.\left[\prod_{t=1}^N \exp\left( \sum_{n=1}^k c_F^{(n)} (x_t)^n \right)\right]\right|_{\setzeroes}\\
= (\D_1)^{l-1} \left.\left[ g(x_1) \prod_{t=1}^N \exp\left( \sum_{n=1}^k c_F^{(n)} (x_t)^n \right)\right]\right|_{\setzeroes},
\end{multline}
where
$$
g(x_1) := \sum_{n=1}^k n c_F^{(n)} (x_1)^{n-1}.
$$
Note that for any polynomial $H$ we have
\begin{equation}
\label{eq_x8}
\frac{\pa}{\pa x_1}\!\left[H\cdot   \prod_{t=1}^N   \exp\left( \sum_{n=1}^k c_F^{(n)} (x_t)^n \right)\right] = \left( \frac{\pa H}{\pa x_1} + H \cdot g(x_1) \right)\cdot  \prod_{t=1}^N \exp\left( \sum_{n=1}^k c_F^{(n)} (x_t)^n \right)
\end{equation}
and
\begin{multline}
\label{eq_x9}
\frac{\theta}{x_1-x_j}(1-s_{1j}) \!\left[H\cdot   \prod_{t=1}^N \exp\left( \sum_{n=1}^k c_F^{(n)} (x_t)^n \right)\right] \\=\left[ \frac{\theta}{x_1-x_j}(1-s_{1j}) H\right]\cdot   \prod_{t=1}^N \exp\left( \sum_{n=1}^k c_F^{(n)} (x_t)^n \right).
\end{multline}
Combining \eqref{eq_x8} and \eqref{eq_x9}, we can rewrite \eqref{eq_x6} as
\begin{equation}\label{eq_x7}
  (\D_1+*_g)^{l-1} g(x_1)\Bigr|_{\setzeroes},
\end{equation}
where $*_g$ is the operator if multiplication by $g(x_1)$.
It remains to note we can replace each operator $\frac{\theta}{x_1-x_j}(1-s_{1j})$ in $\D_1$ by $\th d_1$.
Indeed, this is done by the exact same reasoning that we used in the proof of Proposition \ref{Proposition_highest_derivatives}, see \eqref{eq_x2}.
After we make this replacement, we conclude that (up to another $O\left(\frac{1}{N}\right)$ error) \eqref{eq_x7} and \eqref{eq_x6} are equal to
$$
\left(\frac{\pa}{\pa x_1}+\gamma d_1+*_g \right)^{l - 1} g(x_1)\Bigr|_{\setzeroes} = [z^0](\pa+\gamma d+*_g)^{l-1} g(z).
$$
Plugging this expression back into \eqref{eq_x5} gives the desired result.
\end{proof}

\bigskip

After all these preparations, it remains to put everything together, as follows.

\begin{proof}[Proof of Theorem \ref{theorem_operators_expansion}]
By Lemma \ref{Lemma_D_as_polynomial} and Corollary \ref{Corollary_reduce_terms}, the left-hand side of \eqref{eq_operator_small_th_expansion} is a homogeneous polynomial in $c^{\nu}_F$, $\nu\le k$, of degree $k$ (if we regard each $c^\nu_F$ as a variable of degree $|\nu|$) with uniformly bounded coefficients as $N\to\infty$, $\theta\to 0$, $\theta N\to\gamma$.
Proposition \ref{Proposition_highest_derivatives} identifies the linear part of this polynomial (corresponding to $c^{\nu}_F$ with $|\nu|=k$) with the first line in the right-hand side of \eqref{eq_operator_small_th_expansion}.
Proposition \ref{Proposition_LLN_leading_term} identifies the polynomial $L$ from the second line of the right-hand side of \eqref{eq_operator_small_th_expansion} and from \eqref{eq_one_row_part}. It remains to note that subtraction of $k(1+\gamma)_{k-1} c^{(k)}_F$ in \eqref{eq_one_row_part} corresponds to the situation when the parts of the polynomial given by Propositions \ref{Proposition_highest_derivatives} and \ref{Proposition_LLN_leading_term} overlap.
\end{proof}

\section{From $\ga$--cumulants to moments}\label{mom_cum_sec}

The goal of this section is to prove Theorem \ref{theorem_cumuls_moms}.
That is, let us begin with any real sequence $\ka_1, \ka_2, \cdots$, consider the power series $g(z) := \sum_{l=1}^{\infty} {\ka_l z^{l-1}}$, and the operators $\pa$, $d$, and $*_g$ from Definition \ref{df:ops}.
We denote the constant term of a power series $h(z)$ by $[z^0]h(z)$.

Recall that $\PP(k)$ denotes the collection of all set partitions of $[k]$, and for each $\pi\in\PP(k)$ we introduced the $\ga$-weight $W(\pi)$ in Definition \ref{W_def}.
If $\pi = B_1\sqcup \cdots\sqcup B_m\in\PP(k)$, then the $B_i$'s are called the \emph{blocks} of $\pi$.
The cardinality of the block $B_i$ is denoted by $|B_i|$.

With these recollections, Theorem \ref{theorem_cumuls_moms} says that for any $k\in\Z_{\ge 1}$, we must have the equality
\begin{equation}\label{eqn_transition_2}
[z^0](\pa + *_g + \ga d)^{k-1} (g(z)) \myeq \sum_{\pi = B_1\sqcup \cdots\sqcup B_m\in\PP(k)}{\!\!\!\!\!\!W(\pi)\prod_{i=1}^m{\ka_{|B_i|}}}.
\end{equation}

This relation and actually a more general version (see Theorem \ref{refined_blocks}) will be proved in this section.

\subsection{A refined combinatorial theorem}\label{sec:refined}

Let $a_1, a_2, \dots$ and $\ka_1, \ka_2, \dots$ be two arbitrary sequences of real numbers.

\begin{df}\label{def_refined_weight}
For any $k\in\Z_{\ge 1}$ and $\pi\in\PP(k)$, we define the quantity $w(\pi)$, that will be called the \emph{refined $\ga$-weight of $\pi$} as follows.\footnote{We omit the dependence on $\ga$ and on the sequences $a_1, a_2, \cdots$ and $\ka_1, \ka_2, \cdots$ from the notation $w(\pi)$.} Suppose that $\pi$ has $m$ blocks and label them $B_1, \cdots, B_m$ in such a way that the smallest element from $B_i$ is smaller than the smallest element from $B_j$ (hence, also smaller than all other elements from $B_j$), whenever $i<j$. Then define
$$
w(\pi) := W(\pi)\cdot \ka_{|B_1|}\prod_{i=2}^m{a_{|B_i|}}.
$$
\end{df}

\medskip

From the formula of $W(\pi)$ in Definition \ref{W_def}, we can give an expanded formula for the refined $\ga$-weight $w(\pi)$ of the set partition $\pi = B_1\sqcup\cdots\sqcup B_m$. Recall that the values $p(i)$, $q(i)$, $1\le i\le m$, are defined by
\begin{gather*}
p(i) := \#\{ j\in\{1, \dots, |B_i| - 1\} \mid \{b^i_j + 1, \dots, b^i_{j+1} - 1\}\cap B_t\neq\emptyset, \text{ for some block }B_t \text{ with } t < i\},\\
q(i) := |B_i|-1-p(i).
\end{gather*}
In particular, $p(1)=0$, $q(1)=|B_1|-1$. The quantity $p(i)$ can be computed by the graphical procedure described in Section \ref{sec_Tcm}.

Define the \emph{weight $w(B_i)$ of the block $B_i$ with respect to $\pi$} by
\begin{equation}
\label{eq_w_block}
w(B_i) :=
\begin{cases} (\ga+1)_{|B_1| - 1}\cdot \ka_{|B_1|}, &\text{ if } i = 1,\\
p(i)!\cdot(\ga+p(i)+1)_{q(i)}\cdot a_{|B_i|}, &\text{ if }i \ge 2.
\end{cases}
\end{equation}
For example, if $i\ge 2$ and $B_i$ is a singleton, then $p(i)=q(i)=0$ and $w(B_i) = a_1$.

The refined $\ga$-weight of the set partition $\pi = B_1\sqcup\cdots\sqcup B_m$ then equals
\begin{align}
w(\pi) =& \prod_{i=1}^m{w(B_i)}\nonumber\\
=& \left((\ga+1)_{|B_1| - 1}\cdot \ka_{|B_1|}\right)\cdot \prod_{i=2}^m{\left( p(i)!\cdot(\ga+p(i)+1)_{q(i)}\cdot a_{|B_i|} \right)}.\label{w_formula}
\end{align}

Observe that under the identifications $a_i \mapsto \ka_i$, for all $i$, we have
$$
\left. w(\pi) \right|_{a_i\mapsto \ka_i} = W(\pi)\prod_{i=1}^m{\ka_{|B_i|}}.
$$
This is why we call $w(\pi)$ the \emph{refined} $\ga$-weight of $\pi$.

\medskip

We show the refined $\ga$-weights for the same examples of set partitions given in Section \ref{sec_Tcm}.
The set partition $\{1, 2, 5, 7\}\sqcup \{3, 4, 6\}\in\PP(7)$ graphically shown in Figure \ref{fig_1} has refined $\ga$-weight
$$w(\pi) = (\ga+1)(\ga+2)^2(\ga+3) \cdot\ka_4 a_3.$$
For the set partition $\{1, 4\}\sqcup\{2, 6\}\sqcup\{3, 5, 7\}\in\PP(7)$ shown in Figure \ref{fig_2}, the refined $\ga$-weight is
$$w(\pi) = 2(\ga+1)\cdot \ka_2 a_2a_3.$$
As a final example, the set partition  $\{1,3,4,5,6\}\sqcup\{2,7\}\in\PP(7)$ shown in Figure \ref{fig_3} has refined $\ga$-weight
$$
w(\pi) = (\gamma+1)(\gamma+2)(\gamma+3)(\gamma+4)\cdot \ka_5a_2.
$$

\begin{thm}\label{refined_blocks}
Set
$$g(z) = \sum_{l=1}^{\infty}{\ka_l z^{l-1}}, \qquad a(z) = \sum_{l=1}^{\infty}{a_l z^{l-1}}.$$
Then we have
\begin{equation}\label{blocks_thm}
[z^0](\pa + *_a + \ga d)^{k-1} (g(z)) = \sum_{\pi = B_1\sqcup \dots\sqcup B_m\in\PP(k)}{w(\pi)}.
\end{equation}
On the left side, we have the constant term of a power series. On the right side, the sum ranges over set partitions of $[k]$, and the refined $\ga$-weight $w(\pi)$ is the one introduced in Definition \ref{def_refined_weight}.
\end{thm}
Note that this result implies Theorem \ref{theorem_cumuls_moms}: indeed, we apply Theorem \ref{refined_blocks} and set $a_i=\kappa_i$.  In the rest of this section we prove Theorem \ref{refined_blocks}.

\subsection{Preliminary lemmas}

\begin{lemma}\label{tech_sums_3}
Let $x, y\in\Z_{\geq 0}$ be arbitrary, and let $z$ be any complex number. Then
\begin{equation}
(y+1)\sum_{i=1}^{x} (z+i)_y = (z+x)_{y+1} - (z)_{y+1}.\label{eqn_tech3}
\end{equation}
\end{lemma}
\begin{proof} The proof is induction on $x$. If $x=0$, then both sides of \eqref{eqn_tech3} vanish. The difference of the left-hand sides of \eqref{eqn_tech3} at $x=t$ and $x=t-1$ is $(y+1)(z+t)_y$. The difference of the right-hand sides of \eqref{eqn_tech3} is the same:
$$
 (z+t)_{y+1}- (z+t-1)_{y+1}= (z+t)_y \bigl(z+t+y- (z+t-1)\bigr)=(z+t)_y \cdot (y+1).\qedhere
$$
\end{proof}

For a sequence of $0$s and $1$s, a \emph{descent} is defined as a substring $10$ in this sequence.  Let $\des(\zeta)$ denote the number of descents in a $0$-$1$ sequence $\zeta$. For instance, $\des(1100)=1$ and $\des(0101010)=3$.

\begin{lemma} \label{Lemma_descent_sum}
 For any two integers $N\ge 1$ and $0\le M\le N$ and any $\gamma\in\mathbb R$, we have
 \begin{equation}
  \sum_{\begin{smallmatrix} \zeta=(\zeta_1,\dots,\zeta_N)\in \{0,1\}^N\\ \sum_{i=1}^N \zeta_i=M\end{smallmatrix}}  \des(\zeta)! \cdot (\gamma+\des(\zeta)+1)_{M-\des(\zeta)}=(\gamma+1+N-M)_M.
 \end{equation}
\end{lemma}
\begin{proof}
Let $K(N,M,d)$ denote the total number of sequences $\zeta\in\{0,1\}^N$ with $\sum_{i=1}^N \zeta_i=M$ and $\des(\zeta)=d$. With this notation, we would like to prove that:
 \begin{equation}
 \label{eq_summation}
  \sum_{d=0}^{M}  K(N,M,d)\cdot  d! \cdot (\gamma+d+1)_{M-d}=(\gamma+1+N-M)_M, \quad N\ge 1,\quad 0\le M\le N.
 \end{equation}
 Our proof is induction on $N$. If $N=1$, then both sides of \eqref{eq_summation} are $1$ at $M=0$ and both sides are $(\gamma+1)$ at $M=1$.
 For the induction step, assume that \eqref{eq_summation} holds for all values $\le N$, and let us prove it for $N+1$, and an arbitrary $0 \le M \le N+1$. We notice that the statement is straightforward at $M=0$. If $M>0$, then we use the following recurrence for $K(N,M,d)$, which is obtained by considering the position of the right-most $1$ in a sequence $\zeta$:
 \begin{equation}
  \label{eq_count_recurrence}
  K(N+1,M,d)=K(N,M-1,d)+\sum_{p=1}^{N-M+1} K(N-p, M-1,d-1).
 \end{equation}
 If $d=M$, then the first term in \eqref{eq_count_recurrence} in not needed; If $d=0$, then the second term in \eqref{eq_count_recurrence} is not needed.
 Hence, the left-hand side of \eqref{eq_summation} for $N$ replaced with $N+1$ can be rewritten using \eqref{eq_count_recurrence} as
 \begin{multline}
 \label{eq_x10}
    \sum_{d=0}^{M}  K(N+1,M,d)\cdot  d! \cdot (\gamma+d+1)_{M-d}\\=
    \sum_{d=0}^{M-1}  K(N,M-1,d)\cdot  d! \cdot (\gamma+d+1)_{M-d}+ \sum_{d=1}^M \sum_{p=1}^{N-M+1}  K(N-p, M-1,d-1) \cdot  d! \cdot (\gamma+d+1)_{M-d}.
 \end{multline}
 For the first sum in the right-hand side of \eqref{eq_x10}, we use $d! \cdot (\gamma+d+1)_{M-d}= d! \cdot (\gamma+d+1)_{M-1-d} \cdot (\gamma+M)$ and induction assumption to evaluate it as
 \begin{equation}
 \label{eq_x11}
  (\gamma+2+N-M)_{M-1}\cdot (\gamma+M).
 \end{equation}
 For the second sum in the right-hand side of \eqref{eq_x10} we change the order of summation and evaluate the sum over $d$ for fixed $p$ using the induction assumption and identity (valid for $d>0$)
 $$
  d! \cdot (\gamma+d+1)_{M-d}= (d-1)! \cdot (\gamma+d)_{M-d} \cdot (\gamma+M) - (d-1)! \cdot (\gamma+1+d)_{M-d} \cdot \gamma.
 $$
 Hence, the $p$th term evaluates to
 \begin{equation}
 \label{eq_x12}
   (\gamma+2+N-p-M)_{M-1} \cdot (\gamma+M)- (\gamma+3+N-p-M)_{M-1} \cdot \gamma.
 \end{equation}
 Combining \eqref{eq_x11} with \eqref{eq_x12}, we transform \eqref{eq_x10} to
\begin{multline*}
 (\gamma+M) \sum_{p=0}^{N-M+1} (\gamma+2+N-p-M)_{M-1} - \gamma \sum_{p=1}^{N-M+1}(\gamma+3+N-p-M)_{M-1}
 \\=  M \sum_{p=0}^{N-M} (\gamma+2+N-p-M)_{M-1}\, +\, (\gamma+M)\cdot (\gamma+1)_{M-1}.
\end{multline*}
We compute the last sum using Lemma \ref{tech_sums_3}, resulting in
$$
(\gamma+2+N-M)_{M}-(\gamma+1)_{M} + (\gamma+M)\cdot (\gamma+1)_{M-1}=(\gamma+2+N-M)_{M}.\qedhere
$$
\end{proof}

\subsection{Proof of Theorem \ref{refined_blocks}}

The proof is by induction on $k$.

\smallskip

{\bf Step 1.} For the base cases, let us consider $k = 1$ and $k = 2$.

\smallskip

\emph{Case $k = 1$.} Clearly, $[z^0]g(z) = [z^0](\ka_1 + \ka_2z + \dots) = \ka_1$.
On the other hand, there is exactly one set partition of $[1]$, namely $\pi = \{1\}$ with $B_1 = \{1\}$. For this partition, $p(1) = q(1) = 0$, so it follows that $w(\pi) = \ka_1$, as needed.

\emph{Case $k = 2$.} Here, $[z^0](\pa + *_a + \ga d)g(z) = [z^0](g'(z) + g(z)a(z) + \ga dg(z)) = \ka_2 + \ka_1a_1 + \ga\ka_2 = \ka_1a_1 + (\ga+1)\ka_2$.
On the other hand, there are two set partitions of $[2]$, namely $\pi_1 = \{1\}\sqcup\{2\}$ with $B_1 = \{1\}$, $B_2 = \{2\}$, and $\pi_2 = \{1,2\}$ with $B_1 = \{1, 2\}$.
For $\pi_1$, $p(1) = p(2) = q(1) = q(2) = 0$, so $w(\pi_1) = \ka_1a_1$.
For $\pi_2$, $p(1) = 0$ and $q(1) = 1$, so $w(\pi_2) = (\ga + 1)\ka_2$.
Therefore, $w(\pi_1) + w(\pi_2) = \ka_1a_1 + (\ga + 1)\ka_2$.

\medskip

{\bf Step 2.}
Suppose that the statement of the theorem is true for certain $k\ge 2$.
For the induction step, we  prove the statement for $k+1$, i.e. we aim to obtain the formula for $[z^0](\pa + *_a + \ga d)^k (g(z))$.

Let $b_1, b_2, b_3, \dots$ be the quantities defined by
$$
(\pa + *_a + \ga d)(g(z)) = b_1 + b_2z + b_3z^2 + \cdots.
$$
From the definition of the operators $\pa$, $*_a$, and $d$, we have
\begin{equation}\label{b_vars}
b_n = \sum_{j=1}^n{\ka_{n+1-j} a_j} + (\ga + n)\ka_{n+1},\quad n\in\Z_{\geq 1}.
\end{equation}
Next, use the induction hypothesis to obtain the combinatorial formula:
\begin{equation}\label{ind_hyp}
[z^0](\pa + *_a + \ga d)^k (g(z)) = [z^0](\pa + *_a + \ga d)^{k-1} (b_1 + b_2 z + b_3 z^2 + \dots) = \sum_{\tilde \pi = \tilde B_1\sqcup \dots\sqcup \tilde B_m\in\PP(k)}{\tilde w(\tilde \pi)},
\end{equation}
where $\tilde w(\tilde \pi)$ is given in the theorem (see \eqref{w_formula}), but instead of the $\ka_i$'s, we should use the $b_i$'s:
\begin{equation} \label{eq_tilde_w}
\tilde w(\tilde \pi) = \left((\ga+1)_{|\tilde B_1|-1}\cdot b_{|\tilde B_1|}\right)\cdot \prod_{i=2}^m{\left(p(i)!\cdot(\ga+p(i)+1)_{q(i)}\cdot a_{|\tilde B_i|}\right)}.
\end{equation}
From \eqref{b_vars}, this equals
\begin{multline}\label{w_form}
\tilde w(\tilde\pi) = \sum_{j = 1}^{|\tilde B_1|} \left( (\ga+1)_{|\tilde B_1|-1}\cdot \ka_{|\tilde B_1|+1-j}\, a_j \right)\cdot \prod_{i=2}^m{\left(p(i)!\cdot(\ga+p(i)+1)_{q(i)}\cdot a_{|\tilde B_i|}\right)}\\
+ \left( (\ga+1)_{|\tilde B_1|-1}\cdot (\ga+|\tilde B_1|)\ka_{|\tilde B_1|+1} \right)\cdot \prod_{i=2}^m{\left(p(i)!\cdot(\ga+p(i)+1)_{q(i)}\cdot a_{|\tilde B_i|}\right)}.
\end{multline}

Our next goal is to obtain a different combinatorial expression for \eqref{ind_hyp}, \eqref{w_form} --- we should get the right-hand side of \eqref{blocks_thm} for $k+1$, namely a formula that involves set partitions of $[k+1]$.

\medskip

{\bf Step 3.}
Given a set partition $\pi$ of $[k+1] = \{1, 2, 3, \cdots, k+1\}$, consider the set partition $\tilde \pi$ of $\{1, 3, 4, \cdots, k+1\}$ that is obtained from $\pi$ by taking the union of the blocks that contain $1$ and $2$, and then removing $2$.
If $\tilde \pi$ is obtained from $\pi$ in this fashion, we say that \emph{$\pi$ maps to $\tilde \pi$} and denote this relation by $\pi\to\tilde \pi$.
Observe that if $1$ and $2$ belong to the same block of $\pi$, then $\pi$ and $\tilde \pi$ have the same number of blocks. On the other hand, if $1$ and $2$ belong to different blocks of $\pi$, then $\tilde \pi$ has one block fewer than $\pi$. For instance, for $k=2$ we have $5$ set partitions of $\{1,2,3\}$ which are mapped to two set partitions of $\{1,3\}$:
$$
   \{1\}\sqcup \{2,3\}\to \{1,3\}, \qquad  \{1,3\}\sqcup \{2\}\to \{1,3\},\qquad  \{1,2,3\}\to \{1,3\},
$$
$$
 \{1\}\sqcup \{2\}\sqcup \{3\}\to \{1\} \sqcup \{3\}, \qquad \{1,2\}\sqcup \{3\}\to \{1\}\sqcup\{3\}.
$$
For a set partition $\tilde \pi$ of $\{1,3,4,\dots, k+1\}$ we define the numbers $p(i)$, $q(i)$ and the weight $\tilde w(\tilde \pi)$ by identifying $\{1,3,4,\dots,k+1\}$ with $\{1,2,\dots,k\}$ in a monotone way and using the previous formula \eqref{eq_tilde_w}. Note that essentially nothing changes in the definition, as the way we compute the numbers $p(i)$, $q(i)$ and the weight $\tilde w(\cdot)$ depends only on the order of the elements of the set that we are partitioning rather than the labels of these elements. Hence, we use the same $\tilde w(\tilde \pi)$ notation no matter whether $\tilde \pi$ is a partition of $\{1,3,4,\dots,k+1\}$ or $\tilde \pi$ is a partition of $\{1,2,\dots,k\}$.

Our goal now is to prove that for each set partition $\tilde \pi$ of $\{1,3,4,\dots, k+1\}$ we have an identity:
\begin{equation}\label{eq_inductive_sum}
 \sum_{\begin{smallmatrix} \pi\in \PP(k+1)\\ \pi\to \tilde \pi\end{smallmatrix}} w(\pi)\stackrel{?}{=}\tilde w(\tilde \pi).
\end{equation}
The last equation together with \eqref{ind_hyp} implies the induction step. We fix $\tilde \pi$ and let its blocks be $\tilde B_1,\dots,\tilde B_m$ ordered, as  before, by their minimal elements, so that $\tilde B_1$ contains $1$. We calculate the sum \eqref{eq_inductive_sum} by splitting the terms into several subsets. Define $T\subseteq \PP(k+1)$ as the subset of those set partitions $\pi$, mapped to $\tilde \pi$, for which $1$ and $2$ belong to the same block of $\pi$.
Next, for any $r\in\{1, 2, \dots, |\tilde B_1|\}$, define $T_r\subseteq \PP(k+1)$ as the subset of those set partitions $\pi$, mapped to $\tilde \pi$, for which $1$ and $2$ belong to distinct blocks of $\pi$ and the block where $2$ belongs is of size $r$. The sets $T$ and $T_1,\dots,T_{|\tilde B_1|}$ are all disjoint; they depend on $\tilde \pi$, but we omit this dependence from the notations. With these notations the desired identity \eqref{eq_inductive_sum} is rewritten
\begin{equation}\label{eq_inductive_sum_2}
 \sum_{\pi\in T\cup T_1\cup T_2\cup\dots\cup T_{|\tilde B_1|} } w(\pi)\stackrel{?}{=}\tilde w(\tilde \pi).
\end{equation}

In step 4 below, we prove that $\sum_{\pi\in T}{w(\pi)}$ is equal to the second line of \eqref{w_form}. In steps 5 and 6, we prove that for any $r\in\{1, 2, \dots, |\tilde B_1|\}$ the sum $\sum_{\pi\in T_r}{w(\pi)}$ is equal to the $j = r$ summand in \eqref{w_form}.
After these steps are done, the identity \eqref{eq_inductive_sum_2} would follow and thus the proof would be complete.

\medskip

{\bf Step 4.} In this step we calculate $\sum_{\pi\in T}{w(\pi)}$. Recall that $T$ contains all set partitions $\pi$ of $[k+1]$ that map to $\tilde \pi$ and such that $1$ and $2$ belong to the same block of $\pi$. In fact, for a given $\tilde \pi$ with blocks $\tilde B_1,\dots\tilde B_m$, there is only one set partition in $T$: it is $\pi=B_1\sqcup \dots\sqcup B_m$, where $B_1 = \tilde B_1\cup \{2\}$ and $B_h = \tilde B_h$, $h=2,3,\dots,m$. Since $1$ and $2$ are adjacent in the ordering of $[k+1]$ and they belong to the same block of $\pi$, it is clear that the weight of $B_i$, $i\ge 2$, in the computation of $w(\pi)$ is the same as the weight of $\tilde B_i$ in the computation of $\tilde w(\tilde \pi)$.
The weight of $B_1$ in the computation of $w(\pi)$ is $(\ga+1)_{|B_1|-1}\cdot \ka_{|B_1|} = (\ga+1)_{|\tilde B_1|}\cdot \ka_{|\tilde B_1|+1}$.
As a result,
\begin{equation*}
w(\pi) = \left((\ga+1)_{|\tilde B_1|}\cdot \ka_{|\tilde B_1|+1}\right) \cdot \prod_{i=2}^m{\left(p(i)!\cdot(\ga+p(i)+1)_{q(i)}\cdot a_{|\tilde B_i|}\right)}.
\end{equation*}
This is exactly the second line of \eqref{w_form}, as desired.

\medskip

{\bf Step 5.} In this step we calculate $\sum_{\pi\in T_r}{w(\pi)}$, $1\le r \le |\tilde B_1|$, in the case when $\tilde \pi$ is a one-block set partition, $\tilde \pi=\tilde B_1=\{1,3,4,\dots,k+1\}$ and $|\tilde B_1|=k$. Hence, set partitions $\pi$ in $T_r$ have two blocks $\pi=B_1\sqcup B_2$, $1\in B_1$, $2\in B_2$ and $|B_2|=r$. Therefore, $|B_1|=k+1-r$.

We can identify elements of $T_r$ with $0$--$1$ sequences $\zeta=(\zeta_1,\zeta_2,\dots,\zeta_{k-1})$ of length $k-1$ and with $\sum_{i=1}^{k-1}\zeta_i=k-r$ through:
$$
 B_1(\zeta)=\{1\}\cup\{i+2\mid \zeta_i=1\},\quad B_2(\zeta)=\{2\}\cup \{i+2\mid \zeta_i=0\}.
$$
In words, the positions where $\zeta_i=1$ encode the elements of $B_1$ from the set $\{3,4,\dots,k+1\}$. With this notation, we rewrite
$$
 \sum_{\pi\in T_r} w(\pi)=\sum_{\begin{smallmatrix} \zeta\in \{0,1\}^{k-1}\\ \sum_{i=1}^{k-1}\zeta_i=k-r\end{smallmatrix}} w\bigl( B_1(\zeta)\sqcup B_2(\zeta)\bigr).
$$
Let us compute $w\bigl( B_1(\zeta)\sqcup B_2(\zeta)\bigr)$. By definition \eqref{eq_w_block}, the weight of block $B_1$ is
$$w(B_1)=(\gamma+1)_{k-r} \,\kappa_{k-r+1}.$$
Further, note that in the notations of Lemma \ref{Lemma_descent_sum}, $p(2)=\des(\zeta)$:
Indeed, in the arc diagrams (as in Figures \ref{fig_1}, \ref{fig_2}, \ref{fig_3})  a roof of the block $B_2$ with no legs intersecting it corresponds to a substring $00$ in $\zeta$, whereas a roof with $m \ge 1$ legs intersecting it corresponds to a substring $011\cdots 110$ (with $m$ ones), and this substring gives exactly one descent in $\zeta$. Hence, by definition \eqref{eq_w_block}, we have
$$w(B_2)=\des(\zeta)! (\gamma+\des(\zeta)+1)_{r-\des(\zeta)-1} a_r.$$
We obtain
\begin{multline}
\label{eq_x13}
 \sum_{\pi\in T_r} w(\pi)=\sum_{\begin{smallmatrix} \zeta\in \{0,1\}^{k-1}\\ \sum_{i=1}^{k-1}\zeta_i=k-r\end{smallmatrix}} (\gamma+1)_{k-r} \, \kappa_{k-r+1} \cdot  \des(\zeta)!\, (\gamma+\des(\zeta)+1)_{r-\des(\zeta)-1}\, a_r\\
 = (\gamma+1)_{k-r} \, \kappa_{k-r+1} \, a_r \sum_{\begin{smallmatrix} \zeta\in \{0,1\}^{k-1}\\ \sum_{i=1}^{k-1}\zeta_i=k-r\end{smallmatrix}}   \des(\zeta)!\, (\gamma+\des(\zeta)+1)_{k-r-\des(\zeta)} \cdot \frac{(\gamma+1)_{r-1}}{(\gamma+1)_{k-r}}
 \\=(\gamma+1)_{r-1} \, \kappa_{k-r+1}  a_r \,  (\gamma+r)_{k-r}= (\gamma+1)_{k-1} \kappa_{k-r+1}  a_r,
\end{multline}
where we used Lemma \ref{Lemma_descent_sum} with $N=k-1$ and $M=k-r$ for the equality between the second and the third lines of \eqref{eq_x13}. Since $|\tilde B_1|=k$ and there are no $\tilde B_h$ with $h>1$, the third line in \eqref{eq_x13} matches the $j=r$ term in \eqref{w_form}, as desired.

\medskip

{\bf Step 6.} We now extend the computation of Step 5 and calculate $\sum_{\pi\in T_r}{w(\pi)}$, $1\le r \le |\tilde B_1|$, for arbitrary $\tilde \pi=\tilde B_1\sqcup\dots\sqcup \tilde B_m$. By definition of $T_r$ each set partition $\pi\in T_r$ has $m+1$ blocks $B_1,\dots,B_{m+1}$ and we have $1\in B_1$, $2\in B_2$, $|B_2|=r$, and $B_i=\tilde B_{i-1}$, $3\le i\le m+1$. The key observation for this step is that for $i\ge 3$ the weight of the block $B_i$, $w(B_i)$, is the same as the weight of the block $\tilde B_{i-1}$, $\tilde w(\tilde B_{i-1})$: this is because the blocks $B_i$ and $\tilde B_i$ coincide as sets and the legs intersecting their roofs in the arc diagrams also coincide. Hence, we have
$$
 w(\pi)=w(B_1) w(B_2) \prod_{i=2}^m{\left(p(i)!\cdot(\ga+p(i)+1)_{q(i)}\cdot a_{|\tilde B_i|}\right)}
$$
It remains to sum the last formula over all possible choices of $B_1$ and $B_2$. Since $B_1\cup B_2=\tilde B_1 \cup \{2\}$ is fixed, this is the same computation as in Step 5, but with $k$ replaced by $|\tilde B_1|$. As a result, we get
$$
  \sum_{\pi\in T_r} w(\pi)= (\gamma+1)_{|\tilde B_1|-1} \kappa_{|\tilde B_1|-r+1}  a_r \prod_{i=2}^m{\left(p(i)!\cdot(\ga+p(i)+1)_{q(i)}\cdot a_{|\tilde B_i|}\right)},
$$
which matches the $j=r$ term in \eqref{w_form}, as desired.

\section{From moments to $\ga$--cumulants}\label{sec_mom_to_cums}

Let $\{m_k\}_{k\ge 1}$ and $\{\ka_l\}_{l\ge 1}$ be real sequences related by $\{\ka_l\}_{l\ge 1} = \Tmc(\{m_k\}_{k\ge 1})$.
In this section, we prove Theorem \ref{thm:mom_cums2}, namely the following identity:
\begin{equation}
\label{eq_x14}
\exp\left(\sum_{l=1}^{\infty}{\frac{\ka_l y^l}{l}}\right) \myeq [z^0]\left( \sum_{n=0}^{\infty}{\frac{(yz)^n}{(\ga)_n}} \right) \exp\left( \gamma \sum_{k=1}^{\infty} \frac{m_k}{kz^k} \right).
\end{equation}
The central idea of our proof is to apply Theorem \ref{thm_small_th} to the measure $\mu_N$ which is the Dirac delta-mass at a single $N$--tuple $(a_1^{(N)},\dots,a_N^{(N)})$.  The Bessel generating function of $\mu_N$ is the multivariate Bessel function $B_{(a_1^{(N)},\dots,a_N^{(N)})}(x_1,\dots,x_N;\theta)$, which allows us to use the known formulas for $B_{(a_1,\dots,a_N)}(y,0^{N-1}; \theta)$ and get the asymptotic expressions for the partial derivatives of the logarithm of the BGF at $0$. We remark that it would be interesting to find a more direct combinatorial proof, explaining how \eqref{eq_x14} matches the expressions of Theorem \ref{theorem_cumuls_moms}.

\bigskip

As our proof of Theorem \ref{thm:mom_cums2} is based on the asymptotic analysis of $B_{(a_1,\dots,a_N)}(y,0^{N-1}; \theta)$, we start by collecting formulas for this function. Assume, as usual, that $a_1\le a_2\le\dots\le a_N$ are real and $y\in\mathbb C$. There are at least three different ways to think about $B_{(a_1,\dots,a_N)}(y,0^{N-1}; \theta)$:
\begin{enumerate}
\item The Taylor series expansion for $B_{(a_1,\dots,a_N)}(y,0^{N-1}; \theta)$ (which is a limit of the binomial formula for Jack polynomials of \cite{Ok_Olsh_shifted_Jack}) reads
\begin{equation}
\label{eq_x15}
 B_{(a_1,\dots,a_N)}(y,0^{N-1}; \theta)=\sum_{k=0}^{\infty} \frac{Q_{(k)}(a_1,\dots, a_N; \, \theta)}{(\theta N)_{k}} y^k,
\end{equation}
where $Q_{(k)}(a_1,\dots, a_N; \, \theta)$ is the value of the $N$--variable Jack symmetric polynomial (with normalization as for $Q$--functions in \cite[Chapter VI, Section 10]{M} or \cite{Ok_Olsh_shifted_Jack}) parameterized by one-row partition $(k)$ at the point $(a_1,\dots,a_N)$. The expansion \eqref{eq_x15} is a particular case of \cite[(4.2)]{Ok_Olsh_shifted_Jack}; that article uses the same parameter $\theta$ for the Jack polynomials, but it is worth mentioning that some other authors (e.g., \cite{Stanley_Jack} or \cite[Section VI.10]{M}) use $\alpha=\theta^{-1}$ instead.

\item The contour integral representation for $B_{(a_1,\dots,a_N)}(y,0^{N-1}; \theta)$ claims that for any complex $y$ with $\Re y>0$ we have
\begin{equation}\label{eqn:int_repr}
B_{(a_1,\dots,a_N)}(y, 0^{N-1}; \theta) = \frac{\Gamma(\theta N)}{y^{\theta N - 1}}\frac{1}{2\pi\ii}
\int_{\mathscr{C}_{\infty}}{\exp(yz)\prod_{j=1}^N{(z - a_j)^{-\theta}}dz},
\end{equation}
where the infinite contour $\mathscr{C}_\infty$ in this formula is positively oriented and is formed by the segment $[M-r\ii, M+r\ii]$ and the horizontal lines $[M+r\ii, -\infty+r\ii)$, $[M-r\ii, -\infty-r\ii)$, for  real numbers $M>a_N$ and $r>0$. The proof of \eqref{eqn:int_repr} can be found in \cite[Theorem 5.1]{C} and the same article contains a complementary integral representation for $\Re y<0$.

\item A stochastic representation for $B_{(a_1,\dots,a_N)}(y,0^{N-1}; \theta)$ reads
\begin{equation}
\label{eq_x16}
 B_{(a_1,\dots,a_N)}(y,0^{N-1}; \theta)= \E\left[ \exp\left({y\sum_{i=1}^N a_i \eta_i}\right)\right],
\end{equation}
where $(\eta_1,\dots,\eta_N)$ is a Dirichlet-distributed random vector with all parameters equal to $\theta$. The proof of \eqref{eq_x16} can be found in \cite[Proposition 5.1]{AN}, although in some forms this statement was known before, see, e.g., \cite[Remark 8.3]{OV}.

\end{enumerate}

Either of the above three approaches can be used to establish a formula for the generating function of derivatives of $B_{(a_1,\dots,a_N)}(y,0^{N-1}; \theta)$ at $0$:

\begin{prop} \label{Theorem_Bessel_derivatives} For any $\theta>0$, $N\in\mathbb Z_{>0}$, $y\in \mathbb C$, and $a_1\le a_2\le \dots\le a_N$ we have the expansion
\begin{equation}
\label{eq_x17}
 B_{(a_1,\dots,a_N)}(y,0^{N-1}; \theta)=\sum_{k=0}^{\infty} \frac{c_k}{(\theta N)_k} y^k,
\end{equation}
where the numbers $c_k$ are found from the following Taylor series expansion:
\begin{equation}
\label{eq_x18}
  \sum_{k=0}^{\infty} c_k z^k = \prod_{i=1}^{N} (1-a_i z)^{-\theta}.
\end{equation}
The series \eqref{eq_x17} is uniformly convergent over $y$ in compact subsets of $\mathbb C$.
\end{prop}
\begin{proof}
 According to \eqref{eq_x15}, the coefficients $c_k$ in \eqref{eq_x17} are computed as $c_k=Q_{(k)}(a_1,\dots, a_N; \, \theta)$. The generating function for the one-row Jack polynomials is well-known, see \cite[Section VI.10, top formula on page 378]{M} or \cite[(9)]{Stanley_Jack}. We have:
 $$
   \sum_{k=0}^{\infty} Q_{(k)}(a_1,\dots, a_N; \, \theta) z^k = \prod_{i=1}^{N} (1-a_i z)^{-\theta},
 $$
 which proves \eqref{eq_x18}. Finally, uniform convergence of \eqref{eq_x17} follows either from the fact that we deal with a Taylor series expansion of an entire function, or from bounds $c_k< C r^k$ for some $C>0$, $r>0$, which can be extracted from \eqref{eq_x18}.
\end{proof}

In view of Theorem \ref{thm_small_th}, the desired identity \eqref{eq_x14} of Theorem \ref{thm:mom_cums2} becomes the limit as $N\to\infty$, $\theta N\to \gamma$ of Theorem \ref{Theorem_Bessel_derivatives}, as we now explain.

\begin{proof}[Proof of Theorem \ref{thm:mom_cums2}] {\bf Step 1.} We first show that the formulas \eqref{eq_cums_moments} and \eqref{eq_cums_moments_2} are equivalent. Indeed, \eqref{eq_cums_moments} says that the coefficient of $y^n$ in $\exp\left(\sum_{l=1}^{\infty}  \tfrac{\kappa_l}{l} y^l\right)$ can be computed as the constant term of $$\frac{z^n}{(\gamma)_n} \exp\left(\gamma \sum_{k=1}^{\infty} \frac{m_k}{k} z^{-k}\right).$$ On the other hand, \eqref{eq_cums_moments_2} says that the same coefficient can be computed as $\tfrac{1}{(\gamma)_n}$ times the coefficient of $z^n$ in $$\exp\left(\gamma \sum_{k=1}^{\infty} \frac{m_k}{k} z^{k}\right).$$ Clearly, the latter and the former are two ways to compute the same number.

\medskip

{\bf Step 2.} Next, let us assume that $m_1$, $m_2$, $m_3$, \dots\, are moments of a compactly supported probability measure $\mu$, i.e., there exists $r>0$ and a probability measure $\mu$ supported inside $[-r,r]$, such that
\begin{equation}
\label{eq_x19}
 m_k=\int_{-r}^r x^k \mu(dx),\qquad k=1,2,\dots.
\end{equation}
As it is true for any probability measure, $\mu$ can be approximated by discrete measures with atoms of weight $1/N$ as $N\to\infty$, and we can choose these measures to be supported inside $[-r,r]$. Let us fix such an approximation for $\mu$, that is,
we choose real numbers $-r\le a_1^{(N)}\le a_2^{(N)}\le \dots \le a_N^{(N)}\le r$, such that
$$
 \lim_{N\to\infty} \frac{1}{N} \sum_{i=1}^N \delta_{a_i^{(N)}}=\mu. \qquad \text{(Weak convergence of measures.)}
$$
In particular, this implies
$$
 \lim_{N\to\infty} \frac{1}{N} \sum_{i=1}^N (a_i^{(N)})^k=m_k,\qquad k=1,2,\dots.
$$
Let $\mu_N$ be the Dirac delta-mass at the point $(a_1^{(N)},\dots,a_N^{(N)})$ --- this is a probability measure on $\mathbb R^N$. Then the measures $\mu_N$ satisfy the Law of Large Numbers in the sense of Definition \ref{Definition_LLN_sat_ht} with sequence $m_k$ given by \eqref{eq_x19}. The BGF of the measure $\mu_N$ is $B_{(a_1^{(N)},\dots,a_N^{(N)})}(x_1,\dots,x_N;\, \theta)$. Hence, Theorem \ref{thm_small_th} yields that
\begin{equation}
\label{eq_x22}
 \lim_{\begin{smallmatrix} N\to\infty,\, \theta\to 0 \\ \theta N\to \gamma \end{smallmatrix}} \frac{\partial^l}{\partial y^l} \ln B_{(a_1^{(N)},\dots,a_N^{(N)})}(y, 0^{N-1};\, \theta)\Bigr|_{y=0}= (l-1)!\cdot \kappa_l,
\end{equation}
where $\{\ka_l\}_{l\ge 1} = \Tmc(\{m_k\}_{k\ge 1})$.

On the other hand, note that the formulas \eqref{eq_x17} and \eqref{eq_x18} have a limit in the regime $N\to \infty$, $\theta\to 0$, $\theta N\to\gamma$, which reads:
\begin{equation}
\label{eq_x20}
\lim_{\begin{smallmatrix} N\to\infty,\, \theta\to 0 \\ \theta N\to \gamma \end{smallmatrix}} B_{(a_1^{(N)},\dots,a_N^{(N)})}(y,0^{N-1}; \theta)=\sum_{k=0}^{\infty} \frac{c_k}{(\gamma)_k} y^k,
\end{equation}
where the numbers $c_k$ are found from the following  Taylor series expansion:
\begin{multline}
\label{eq_x21}
  \sum_{k=0}^{\infty} c_k z^k = \lim_{\begin{smallmatrix} N\to\infty,\, \theta\to 0 \\ \theta N\to \gamma \end{smallmatrix}} \prod_{i=1}^{N} (1-a_i^{(N)} z)^{-\theta}=\lim_{\begin{smallmatrix} N\to\infty,\, \theta\to 0 \\ \theta N\to \gamma \end{smallmatrix}} \exp\left[-\theta \sum_{i=1}^{N} \ln\left(1-a_i^{(N)} z\right)\right]\\=\lim_{\begin{smallmatrix} N\to\infty,\, \theta\to 0 \\ \theta N\to \gamma \end{smallmatrix}} \exp\left[\theta N \sum_{k=1}^{\infty} \frac{z^k}{k} \frac{1}{N}\sum_{i=1}^{N} (a_i^{(N)})^k\right]= \exp\left[ \gamma \sum_{k=1}^{\infty} \frac{m_k z^k}{k}\right].
\end{multline}
Because $|a_i^{(N)}|\le r$ for all $1\le i \le N$, for any $\eps>0$ the convergence in \eqref{eq_x21} is uniform over $|z|\le r^{-1}-\eps$. We claim that the convergence in \eqref{eq_x20} is uniform over $y$ in compact subsets of $\mathbb C$. Indeed, the term-by-term
convergence of \eqref{eq_x17} to \eqref{eq_x20} is evident from \eqref{eq_x21}, while a tail bound on the series can be obtained from a uniform $N$--independent bound on the coefficients $c_k=c_k^{(N)}$ in \eqref{eq_x17}, \eqref{eq_x18} of the form $|c_k^{(N)}|\le C\cdot r^{-k}$ for $C>0$, which follows from the Cauchy integral formula applied to \eqref{eq_x18}.

Comparing \eqref{eq_x22} with \eqref{eq_x20} and noting that uniform convergence of analytic functions implies convergence of their derivatives, we conclude that
\begin{equation}
\label{eq_x23}
 \exp\left(\sum_{l=1}^{\infty} \frac{\kappa_l}{l} y^l\right)=\sum_{k=0}^{\infty} \frac{c_k}{(\gamma)_k} y^k.
\end{equation}
The last identity together with \eqref{eq_x21} give \eqref{eq_cums_moments_2}.

\medskip

{\bf Step 3.} It remains to study the case when $\{m_k\}_{k\ge 1}$ is an arbitrary sequence of real numbers, rather than a sequence of moments of a compactly supported probability measure $\mu$ as in \eqref{eq_x19}. Note that the relation $\{\ka_l\}_{l\ge 1} = \Tmc(\{m_k\}_{k\ge 1})$ is equivalent to saying that for certain polynomials $Q^l$, we have
$$
 \ka_l= Q^l(m_1,m_2,\dots,m_l), \quad l=1,2,\dots.
$$
On the other hand, \eqref{eq_cums_moments} of Theorem \ref{thm:mom_cums2} is equivalent to saying that for certain polynomials $\tilde Q^l$, we have
$$
 \ka_l= \tilde Q^l(m_1,m_2,\dots,m_l), \quad l=1,2,\dots.
$$
Hence, in order to prove Theorem \ref{thm:mom_cums2}, we need to show that the polynomials $Q^l$ and $\tilde Q^l$ coincide for each $l=1,2,\dots$. Note that two polynomials in $l$ variables coincide if and only if they coincide as functions on a non-empty open set $D\subset \mathbb R^{l}$, and this $D$ can be chosen in an arbitrary way. Fix $l$ and define:
$$
 D:=\left\{\left(\tfrac{1}{l}\sum_{i=1}^{l} d_i,\, \tfrac{1}{l}\sum_{i=1}^{l} (d_i)^2,\dots, \tfrac{1}{l}\sum_{i=1}^{l} (d_i)^l\right) \, \Bigg|\, d_1<d_2<\dots<d_l\right\}\subset \mathbb R^{l}.
$$
The set $D$ is an image of an open set $\{(d_1,\dots,d_l)\subset \mathbb R^l \mid d_1<\dots<d_l \}$ under a smooth map with non-vanishing Jacobian (equal to $l^{-l} \prod_{i<j}  (d_i-d_j)$). Hence, $D$ is open. By Step 2 the polynomials  $Q^l$ and $\tilde Q^l$ coincide as functions on $D$, because each element in $D$ is a moment sequence of the discrete probability measure with atoms $\tfrac{1}{l}$ at points $d_1,\dots,d_l$. Therefore, polynomials $Q^l$ and $\tilde Q^l$ coincide.
\end{proof}

\section{Limits of the maps $\Tcm$ and $\Tmc$ as $\ga\to 0$ and $\ga\to\infty$}

\label{Section_semifree}

In this section we investigate the behavior of the $\gamma$-cumulants and $\gamma$--convolution as $\gamma\to 0$ and $\gamma\to\infty$. We will see that in the former case the conventional cumulants and conventional convolution appear, while in the latter case we link to the free probability counterparts.

\subsection{$\gamma\to 0$ limit}
\label{Section_limit_to_0}

Let us recall the definition of the classical cumulants.

\begin{definition} Given a sequence of moments $\{m_k\}_{k\ge 1}$ we define the corresponding cumulants $\{c_l\}_{l\ge 1}= \tilde T_{m\to c}(\{m_k\}_{k\ge 1})$ through the identity for the generating functions:
\begin{equation}
\label{eq_cumulants_gen}
C(z) = \ln (M(z)),\qquad  \text{where } \quad M(z) := 1 + \sum_{k=1}^\infty{\frac{m_k}{k!}z^k},\quad C(z) := \sum_{l=1}^\infty{\frac{c_l}{l!} z^l}.
\end{equation}
In the opposite direction, given a sequence of cumulants $\{c_l\}_{l\ge 1}$, we define the corresponding sequence of moments $\{m_k\}_{k\ge 1}=\tilde T_{c\to m}(\{c_l\}_{l\ge 1})$ through a combinatorial formula:
\begin{equation}\label{eq_cumulants_combinatorial}
m_k := \sum_{\pi=B_1\sqcup \dots\sqcup B_m \in\PP(k)}\, \prod_{i=1}^m{c_{|B_i|}},\quad k = 1, 2, \cdots.
\end{equation}
\end{definition}
The two definitions \eqref{eq_cumulants_gen} and \eqref{eq_cumulants_combinatorial} are well-known to be equivalent and the maps $\tilde T_{m\to c}$ and $\tilde T_{c\to m}$ are inverse to each other, see, e.g.,  \cite[Sections 1.1--1.2]{MS}.

\begin{thm}\label{Theorem_gamma_to_0_limit}
 Given a sequence of numbers $\{m_k\}_{k\ge 1}$, let
 $$
 \{\kappa^0_l\}_{l\ge 1}=\lim_{\gamma\to 0} \Tmc(\{m_k\}_{k\ge 1}), \qquad  \qquad \{c_l\}_{l\ge 1}= \tilde T_{m\to c}(\{m_k\}_{k\ge 1}).
 $$
Then for each $l=1,2,\dots$,  we  have
$\displaystyle \kappa^0_l= \tfrac{1}{(l-1)!}c_l$.
\end{thm}
\begin{remark}\label{gamma_0_reform}
The second relation implies $\tilde{T}_{c\to m}(\{c_l\}_{l\ge 1}) = \{m_k\}_{k\ge 1}$.
The first one means that if we let $\Tmc(\{m_k\}_{k\ge 1}) =: \{ \ka_l^{\ga} \}_{l\ge 1}$ (the $\ka_l^{\ga}$'s depend on $\ga$), then $\lim_{\ga\to 0}{\ka_l^\ga} = \ka_l^0 = \tfrac{1}{(l-1)!}c_l$, $l\ge 1$.
As $\Tcm$ is the inverse of $\Tmc$, then $\{m_k\}_{k\ge 1} = \Tcm(\{ \ka_l^{\ga} \}_{l\ge 1})$.
Observe that $\Tcm(\{ \ka_l^{\ga} \}_{l\ge 1})$ is a sequence where each entry is a polynomial (the one from Theorem \ref{theorem_cumuls_moms}) in the variables $\ka_l^{\ga}$, $l\ge 1$, and the coefficients of these polynomials have limits as $\ga\to 0$, thus $\{m_k\}_{k\ge 1} = \lim_{\ga\to 0}\Tcm(\{ \ka_l^{\ga} \}_{l\ge 1}) = \lim_{\ga\to 0}\Tcm(\{ \tfrac{1}{(l-1)!}c_l \}_{l\ge 1})$. Hence we have two equations for $\{m_k\}_{k\ge 1}$:
$$
\tilde{T}_{c\to m}(\{c_l\}_{l\ge 1}) = \{m_k\}_{k\ge 1},\qquad
\lim_{\ga\to 0}\Tcm(\{ \tfrac{1}{(l-1)!}c_l \}_{l\ge 1}) = \{m_k\}_{k\ge 1}.
$$
In the same fashion, one can show that these two relations imply the ones in the theorem.
Hence the theorem is equivalent to the statement that these last two definitions of $\{m_k\}_{k\ge 1}$ coincide.
\end{remark}
\begin{proof}[Proof of Theorem \ref{Theorem_gamma_to_0_limit}] We use the description of the map $\Tmc$ of Theorem \ref{theorem_cumuls_moms} and send $\gamma\to 0$ in the weight $W(\pi)$ of \eqref{W_formula}. Note that
$$
 \lim_{\gamma\to 0} \bigl[ p(i)! (\gamma+p(i)+1)_{q(i)}\bigr]=(p(i)+q(i))!
$$
According to definitions of Section \ref{sec_Tcm}, $p(i)+q(i)+1=|B_i|$, i.e., the size of the $i$th block in $\pi$. Hence, the $\gamma\to 0$ limit of Theorem \ref{theorem_cumuls_moms} gives
\begin{equation}
\label{eq_x24}
 m_k=\sum_{\pi=B_1\sqcup\dots\sqcup B_m\in\PP(k)} \prod_{i=1}^m\left[(|B_i|-1)!\, \kappa^0_{|B_i|}\right], \qquad k=1,2,\dots.
\end{equation}
Comparing with \eqref{eq_cumulants_combinatorial} and noting that the relations \eqref{eq_x24} uniquely determine $\{\kappa^0_l\}_{l\ge 1}$, we conclude that $ \kappa^0_l= \tfrac{1}{(l-1)!}c_l$.
\end{proof}

As a corollary, we obtain the $\gamma\to 0$ behavior of the $\gamma$--convolution of Theorem \ref{Theorem_gamma_convolution}.
\begin{corollary}\label{Corollary_convolution_at_0}
 Take two sequences of real numbers $\{m_k^{\mathbf a}\}_{k\ge 1}$ and $\{m_k^{\mathbf b}\}_{k\ge 1}$. Define
 $$
   \{\tilde m_k\}_{k\ge 1}:=\lim_{\gamma\to 0}\left[ \{m_k^{\mathbf a}\}_{k\ge 1}\boxplus_\gamma \{m_k^{\mathbf b}\}_{k\ge 1} \right].
 $$
 Then with the agreement $m_0^{\mathbf a}=m_0^{\mathbf b}=1$, we have
 \begin{equation}
 \label{eq_sum_moments}
  \tilde m_k=\sum_{s=0}^k {k\choose s} m_s^{\mathbf a} m_{k-s}^{\mathbf b}, \quad k=1,2,\dots.
 \end{equation}
\end{corollary}
\begin{remark}
 Suppose that we are given two independent random variables $\mathbf a$ and $\mathbf b$, such that
 $$
   m_k^{\mathbf a} =\E \mathbf a^k, \quad m_k^{\mathbf b}=\E \mathbf b^k,\quad k=1,2,\dots.
 $$
 Then the formula \eqref{eq_sum_moments} says that $\tilde m_k=\E(\mathbf a+\mathbf b)^k$.
\end{remark}
\begin{proof}[Proof of Corollary \ref{Corollary_convolution_at_0}] By \eqref{eq_convolution_cumulants} we have for each $\gamma>0$
$$
 \Tmc(  \{m_k^{\mathbf a}\}_{k\ge 1}\boxplus_\gamma \{m_k^{\mathbf b}\}_{k\ge 1} )= \Tmc(  \{m_k^{\mathbf a}\}_{k\ge 1}) + \Tmc( \{m_k^{\mathbf b}\}_{k\ge 1}).
$$
Taking the limit $\gamma\to 0$ and using Theorem \ref{Theorem_gamma_to_0_limit}, we get
$$
 \tilde T_{m\to c} \bigl( \{\tilde m_k\}_{k\ge 1}\bigr)=  \tilde T_{m\to c}\bigl(  \{m_k^{\mathbf a}\}_{k\ge 1}\bigr) +  \tilde T_{m\to c}\bigl(  \{m_k^{\mathbf b}\}_{k\ge 1}\bigr).
$$
Taking into account \eqref{eq_cumulants_gen}, the last identity is equivalent to \eqref{eq_sum_moments}.
\end{proof}

\subsection{$\gamma\to\infty$ limit}
\label{Section_limit_to_infinity}

Let us recall the definition of the free cumulants.

A set partition $\pi\in\PP(k)$ is said to be \emph{crossing} if there exist two distinct blocks $B, B'\in\pi$ and integers $1\le x < y < z < w \le k$ such that $x, z\in B$ and $y, w\in B'$. A set partition $\pi\in\PP(k)$ is said to be a \emph{non-crossing set partition} if it is not crossing. In the notations of Section \ref{sec_Tcm}, the non-crossing set partitions are those which have $p(i)=0$ for all $i$; in other words, there should be no crossings of roofs and legs.
We denote the collection of all non-crossing set partitions of $[k]$ by $NC(k)$. For instance, out the seven partitions of $[4]$ with two blocks, mentioned in Section \ref{sec_Tcm}, the following six are the non-crossing ones:
$$
 \{1\}\sqcup \{2,3,4\},\quad \{1,3,4\}\sqcup \{2\},\quad \{1,2,4\}\sqcup\{3\},\quad \{1,2,3\}\sqcup\{4\},
$$
$$
 \{1,2\}\sqcup \{3,4\},\quad \{1,4\}\sqcup \{2,3\}.
$$

\begin{definition} Given a sequence of moments $\{m_k\}_{k\ge 1}$ we define the corresponding free cumulants $\{r_l\}_{l\ge 1}= \tilde T^\infty_{m\to r}(\{m_k\}_{k\ge 1})$ through the identity for the generating functions:
\begin{equation}
\label{eq_free_cumulants_gen}
G\bigl(R(z)+z^{-1}\bigr) = z,\qquad  \text{where } \quad G(z) := z^{-1} + \sum_{k=1}^\infty \frac{m_k}{z^{k+1}} ,\quad R(z) := \sum_{l=1}^\infty r_l z^{l-1}.
\end{equation}
In the opposite direction, given a sequence of free cumulants $\{r_l\}_{l\ge 1}$, we define the corresponding sequence of moments $\{m_k\}_{k\ge 1}=\tilde T^\infty_{r\to m}(\{r_l\}_{l\ge 1})$ through a combinatorial formula:
\begin{equation}\label{eq_free_cumulants_combinatorial}
m_k := \sum_{\pi=B_1\sqcup \dots\sqcup B_m \in NC(k)}\, \prod_{i=1}^m{r_{|B_i|}},\quad k = 1, 2, \cdots.
\end{equation}
\end{definition}
The relation $G\bigl(R(z)+z^{-1}\bigr) = z$ can be rewritten as $R(z)=G^{(-1)}(z)-z^{-1}$ and $R(z)$ defined in this way is called the \emph{Voiculescu $R$--transform.}
The equivalence between \eqref{eq_free_cumulants_gen} and \eqref{eq_free_cumulants_combinatorial} is explained in \cite[Sect. 2.4]{MS}.

\begin{thm}\label{Theorem_gamma_to_infinity_limit}
 Given a sequence of numbers $\{m_k\}_{k\ge 1}$ define
 $$
 \{\kappa^\gamma_l\}_{l\ge 1}= \Tmc(\{m_k\}_{k\ge 1}), \qquad  r_l=\lim_{\gamma\to\infty} \gamma^{l-1} \kappa^\gamma_l.
 $$
Then we  have
$\displaystyle \{r_l\}_{l\ge 1}= \tilde T^\infty_{m\to r} (\{m_k\}_{k\ge 1})$.
\end{thm}
\begin{remark} \label{Remark_inverse_infinity}
Just like in Remark \ref{gamma_0_reform}, the theorem can be equivalently stated by saying that the following two definitions
$$
\{m_k\}_{k\ge 1} = \lim_{\gamma\to \infty} \Tcm(\{\gamma^{1-l}r_l\}_{l\ge 1}),\qquad
\{m_k\}_{k\ge 1} = \tilde T^\infty_{r\to m}(\{r_l\}_{l\ge 1}),
$$
of the sequence $\{m_k\}_{k\ge 1}$ coincide.
\end{remark}
\begin{proof}[Proof of Theorem \ref{Theorem_gamma_to_infinity_limit}] We use the reformulation of Remark \ref{Remark_inverse_infinity} together with the description of the map $\Tcm$ of Theorem \ref{theorem_cumuls_moms} and send $\gamma\to \infty$ in the weight $W(\pi)$ of \eqref{W_formula}. Note that
$$
 \lim_{\gamma\to \infty } \bigl[  \gamma^{-p(i)-q(i)} p(i)! (\gamma+p(i)+1)_{q(i)}\bigr]=\begin{cases} 1,& \text{if }p(i)=0,\\ 0, & \text{otherwise.} \end{cases}
$$
According to the definitions of Section \ref{sec_Tcm}, $p(i)+q(i)+1=|B_i|$, so $\ga^{1-|B_i|}r_{|B_i|} = \ga^{-p(i)-q(i)}r_{|B_i|}$. Hence, the $\gamma\to \infty$ limit of Theorem \ref{theorem_cumuls_moms} gives
\begin{multline}
\label{eq_x25}
 m_k=\lim_{\gamma\to\infty} \sum_{\pi=B_1\sqcup\dots\sqcup B_m\in\PP(k)} \prod_{i=1}^m\left[  \gamma^{-p(i)-q(i)} p(i)! (\gamma+p(i)+1)_{q(i)}   r_{|B_i|}\right]\\= \sum_{\pi=B_1\sqcup \dots\sqcup B_m\in NC(k)} \prod_{i=1}^m  r_{|B_i|}, \quad k=1,2,\dots.
\end{multline}
This is the same expression as \eqref{eq_free_cumulants_combinatorial}.
\end{proof}

As a corollary, we obtain the $\gamma\to \infty$ behavior of the $\gamma$--convolution of Theorem \ref{Theorem_gamma_convolution}.
\begin{corollary}\label{Corollary_convolution_at_infty}
 Take two sequences of real numbers $\{m_k^{\mathbf a}\}_{k\ge 1}$ and $\{m_k^{\mathbf b}\}_{k\ge 1}$. Define
 $$
   \{\tilde m_k\}_{k\ge 1}:=\lim_{\gamma\to \infty}\left[ \{m_k^{\mathbf a}\}_{k\ge 1}\boxplus_\gamma \{m_k^{\mathbf b}\}_{k\ge 1} \right].
 $$
 Then the sequence $\{\tilde m_k\}_{k\ge 1}$ is uniquely fixed by the identity
 \begin{equation}
 \label{eq_sum_free_moments}
   \tilde T^{\infty}_{m\to r}\bigl(\{\tilde m_k\}_{k\ge 1}\bigr)=  \tilde T^{\infty}_{m\to r}\bigl(\{m^{\mathbf a}_k\}_{k\ge 1}\bigr)+  \tilde T^{\infty}_{m\to r}\bigl( \{m^{\mathbf b}_k\}_{k\ge 1}\bigr)
 \end{equation}
\end{corollary}
\begin{remark}
 Suppose that we are given two random variables $\mathbf a$ and $\mathbf b$, such that
 $$
   m_k^{\mathbf a} =\E \mathbf a^k, \quad m_k^{\mathbf b}=\E \mathbf b^k,\quad k=1,2,\dots.
 $$
 Then the formula \eqref{eq_sum_free_moments} says that $\tilde m_k$ are the moments of the \emph{free convolution} of $\mathbf a$ and $\mathbf b$.
\end{remark}
\begin{proof}[Proof of Corollary \ref{Corollary_convolution_at_infty}] By \eqref{eq_convolution_cumulants} we have for each $\gamma>0$
$$
 \Tmc\bigl(  \{m_k^{\mathbf a}\}_{k\ge 1}\boxplus_\gamma \{m_k^{\mathbf b}\}_{k\ge 1} \bigr)= \Tmc\bigl(  \{m_k^{\mathbf a}\}_{k\ge 1}\bigr) + \Tmc\bigl( \{m_k^{\mathbf b}\}_{k\ge 1}\bigr).
$$
Taking the limit $\gamma\to \infty$ and using Theorem \ref{Theorem_gamma_to_infinity_limit}, we get \eqref{eq_sum_free_moments}.
\end{proof}

\section{Appendix: Law of Large Numbers for fixed temperature}

\label{Section_Appendix_LLN}

The aim of this Appendix is to probe the possibility of a version of Theorem \ref{thm_small_th} in which $\th > 0$ is fixed and is not changing with $N$.
The following claim is an analogue of \emph{one direction} of Theorem \ref{thm_small_th}.

\begin{claim}[LLN for finite temperature]\label{Claim_finite_th}
Let $\{\mu_N\}_{N \geq 1}$ be a sequence of exponentially decaying probability measures on tuples $a_1\le \cdots \le a_N$.
For each $N$, let $G_{N; \th}(x_1, \cdots, x_N) := G_\th(x_1, \cdots, x_N; \mu_N)$ be the BGF of $\mu_N$.
Assume that the sequence $\{G_{N; \th}\}_N$ satisfies the following conditions:
\begin{enumerate}[label=(\alph*)]
\item $\displaystyle \lim_{N\to\infty}\left.\frac{1}{N}\cdot\frac{\pa^l}{\pa x_i^l}\ln{(G_{N; \th})}\right|_{\setzeroes} =
(l - 1)!\cdot c_l$,\ for all $l\in\Z_{\geq 1}$.

\item $\displaystyle \lim_{N\to\infty}\left.\frac{1}{N}\cdot\frac{\partial}{\partial x_{i_1}}\cdots\frac{\partial}{\partial x_{i_r}}\ln{(G_{N; \th})}\right|_{\setzeroes} = 0$,\ for all $i_1, \dots, i_r\in\Z_{\geq 1}$ such that the set $\{i_1, \dots, i_r\}$ contains at least two distinct indices.
\end{enumerate}

\noindent Then the sequence $\{\mu_N\}_{N \geq 1}$ satisfies the following LLN (compare to Definition \ref{Definition_LLN_sat_ht}):
there exist real numbers $\{m_{k}\}_{k\geq 1}$ such that for any
$s=1,2,\dots$ and any $k_1, \dots, k_s\in\Z_{\geq 1}$, we have
$$\lim_{N\to\infty} \E_{\mu_N} \prod_{i=1}^s  \left( \frac{1}{N} \sum_{i=1}^N\left( \frac{a_i}{N} \right)^{k_i} \right) = \prod_{i=1}^{s}  m_{k_i}.$$

\medskip

If this occurs, $\{c_l\}_{l\geq 1}$ and $\{m_k\}_{k\ge 1}$ are related by either
\begin{equation*}
m_k = \sum_{\pi\in NC(k)}\, \prod_{B\in\pi}{\left( \theta^{|B| - 1}c_{|B|} \right)},\quad k = 1, 2, \cdots,
\end{equation*}
or, equivalently,
\begin{equation*}
m_k = \frac{1}{k+1} \cdot [z^{-1}] \left( \left( z^{-1} + \sum_{l = 1}^{\infty} \th^{l-1}c_l z^{l - 1} \right)^{\!\!k+1} \right)\!,\quad k\in\Z_{\geq 1}.
\end{equation*}
\end{claim}

\bigskip

We do not present a proof of the claim, but probably one can prove it with the same techniques that we have used in Section \ref{sec_proof_LLN}. At $\theta=1$, another approach is by a degeneration of \cite[Theorem 5.1]{BuG1}, see also \cite{NovakM}.
This claim would prove the one-sided implication
$$
\text{conditions (a) and (b) (fixed $\theta$ version of LLN--appropriateness)} \Longrightarrow \text{LLN--satisfaction}.
$$
Based on our Theorem \ref{thm_small_th}, on \cite[Theorem 2.6]{BuG3} (which studies the CLT at fixed $\theta = 1$), and on the classical theorem that relates the weak convergence of measures to the convergence of their characteristic functions, the reader may be inclined to believe that the reverse implication is also true and that this kind of ``if and only if'' results are always expected.

However, this turns out to be wrong. The naive analogue of Theorem \ref{thm_small_th} does not hold for fixed $\theta$: the ``expected if and only if statement" is false. Here is a counter-example.

Let us consider a sequence of probability measures $\mu_N$ such that a random $\mu_N$--distributed vector $a_1 \ge \cdots \ge a_N$ has $a_1 = \cdots = a_N$ almost surely and this common value $a$ is distributed according to a Gaussian measure of mean $0$ and variance $N$.
In this case, the random variable
$$
p_k^N = \frac{1}{N}\sum_{i=1}^N{\left( \frac{a_i}{N} \right)^k} = \left( \frac{a}{N} \right)^k
$$
is distributed as the $k$-th power of a Gaussian random variable of mean $0$ and variance $1/N$.
Consequently, the sequence $\{\mu_N\}_N$ satisfies a LLN, and all $m_k$'s are equal to zero.
On the other hand, by using $B_{(a, \dots, a)}(x_1, \dots, x_N; \th) = \exp(a\sum_{i=1}^N{x_i})$, it follows that the BGF of $\mu_N$ equals
$$
G_{N; \th}(x_1, \cdots, x_N) = \int_{-\infty}^{\infty}{\frac{e^{-a^2/(2N)}}{\sqrt{2\pi N}} B_{(a, \dots, a)}(x_1, \dots, x_N; \th) da}
= \exp\left( \frac{N}{2} \left(\sum_{i=1}^N{x_i}\right)^{\!\!\!2\,} \right).
$$
The sequence $\{G_{N; \th}\}_N$ of BGFs then satisfies
$$
\lim_{N\to\infty} \left. \frac{1}{N}\cdot\frac{\pa^2}{\pa x_1\pa x_2}\ln(G_{N;\th}) \right|_{\setzeroes} = 1,
$$
therefore contradicting condition (b) from Claim \ref{Claim_finite_th}.

\bigskip

A more refined question is whether \emph{some} ``if and only if for LLN'' statement holds, if one modifies somehow the conditions (a) and (b) of Claim \ref{Claim_finite_th}.
Based on small calculations (obtained when trying to reverse-engineer the proof of Theorem \ref{thm_small_th}), it is plausible that the answer is yes.

Indeed, based on Proposition \ref{proposition_moments_through_operators}, we must study the limits of the expressions
\begin{equation}\label{P_lambda}
N^{-|\la|-\ell(\la)}\cdot \left[\prod_{i=1}^{\ell(\la)}{\P_{\la_i}}\right]\! G_{N; \th},
\end{equation}
where $\la$ ranges over the set of all partitions of a given size $k$.
We have performed calculations for $k=2, 3$; they indicate that the conditions on second-order derivatives in (a) and (b) from Claim \ref{Claim_finite_th} should be replaced by:
\begin{align*}
\bullet& \lim_{N\to\infty} \left.\frac{1}{N}\left\{\frac{\pa^2}{\pa x_1^2} - \frac{\pa^2}{\pa x_1\pa x_2} \right\} \ln(G_{N;\th}) \right|_{\setzeroes} = \th^{-1}\cdot c_2,\\
\bullet& \lim_{N\to\infty} \left.\frac{1}{N^2} \frac{\pa^2 \ln(G_{N;\th})}{\pa x_1\pa x_2} \right|_{\setzeroes} = 0,
\end{align*}
and the conditions on third-order derivatives should be replaced by:
\begin{align*}
\bullet& \lim_{N\to\infty} \left.\frac{1}{N}\left\{\frac{1}{2}\cdot\frac{\pa^3}{\pa x_1^3} - \frac{3}{2}\cdot\frac{\pa^3}{\pa x_1^2\pa x_2} + \frac{\pa^3}{\pa x_1\pa x_2 \pa x_3} \right\} \ln(G_{N;\th}) \right|_{\setzeroes} = \th^{-2}\cdot c_3,\\
\bullet& \lim_{N\to\infty} \left.\frac{1}{N^2}\left\{ \frac{\pa^3}{\pa x_1^2\pa x_2} - \frac{\pa^3}{\pa x_1 \pa x_2 \pa x_3} \right\} \ln(G_{N;\th}) \right|_{\setzeroes} = 0,\\
\bullet& \lim_{N\to\infty} \left.\frac{1}{N^3}\frac{\pa^3 \ln(G_{N;\th})}{\pa x_1\pa x_2\pa x_3} \right|_{\setzeroes} = 0.
\end{align*}
These relations are much more involved than conditions (a) and (b) from Claim \ref{Claim_finite_th}, or than the conditions from Definition \ref{Definition_LLN_appr_ht}. What should be the correct ``if and only if'' relations for $k>3$? This is an interesting open question for future research.

\end{document}